\documentclass[acmjacm]{acmsmall}
\pdfoutput=1

\usepackage{latexsym,xspace,calc,amsthm}
\usepackage{amsmath,amssymb} 

\usepackage{lmodern}
\usepackage{microtype}
\usepackage{tgtermes}
\usepackage{fix-cm}

\makeatletter
\def\@biblabel#1{[#1]} 
\def\thebibliography#1{%
    \footnotesize
    \refsection*{{\refname}
        \@mkboth{\uppercase{\refname}}{\uppercase{\refname}}%
    }
    \list{\@biblabel{\@arabic\c@enumiv}}
       {\settowidth\labelwidth{\@biblabel{#1}}%
        \leftmargin\labelwidth
        \advance\leftmargin\bibindent
        \itemindent-\bibindent
        \itemsep2pt
        \parsep \z@
        \usecounter{enumiv}
        \let\p@enumiv\@empty
        \renewcommand\theenumiv{\@arabic\c@enumiv}%
    }%
    \let\newblock\@empty
    \sloppy
    \sfcode`\.=1000\relax
}
\makeatother

 \renewenvironment{thebibliography}[1]{%
   \begin{odlthebibliography}{#1}%
     \setlength{\parskip}{0ex}%
     \setlength{\itemsep}{3pt}%
     \fontsize{9.5}{9.5} 
     \selectfont
}%
 {%
   \end{odlthebibliography}%
 }

\usepackage{tikz}
\usetikzlibrary{calc}
\usetikzlibrary{positioning}
\usetikzlibrary{snakes,automata,chains}

\usepackage{amssymb,stmaryrd,amsmath}
\usepackage{mdwlist} 
\usepackage{float}   

\usepackage[colorlinks]{hyperref}
\usepackage{url}
\usepackage{breakurl}  

\newtheoremstyle{jamiestyle}
  {4pt}
  {0pt}
  {\it}
  {0pt}
  {\sc}
  {.}
  { }
  {}
\theoremstyle{jamiestyle}
\newtheorem{thrm}{Theorem}[subsection]
\newtheorem{prop}[thrm]{Proposition}
\newtheorem{lemm}[thrm]{Lemma}
\newtheorem{corr}[thrm]{Corollary}
\newtheoremstyle{jamienfstyle}
  {4pt}
  {0pt}
  {\normalfont}
  {0pt}
  {\sc}
  {.}
  { }
  {}
\theoremstyle{jamienfstyle}
\newtheorem{nttn}[thrm]{Notation}
\newtheorem{defn}[thrm]{Definition}
\newtheorem{xmpl}[thrm]{Example}
\newtheorem{rmrk}[thrm]{Remark}

\usepackage{color}
\definecolor{mygreen}{rgb}{0,0.6,0}
\definecolor{mygray}{rgb}{0.5,0.5,0.5}
\definecolor{mymauve}{rgb}{0.58,0,0.82}

\usepackage{listings}
 
\definecolor{gray}{RGB}{128, 128, 128}
\definecolor{lightgray}{RGB}{200, 200, 200}
\definecolor{cyan}{RGB}{0, 255, 255}
\definecolor{blue}{RGB}{0, 0, 255}
\definecolor{red}{RGB}{255, 0, 0}
\definecolor{pink}{RGB}{255, 128, 128}
\definecolor{green}{RGB}{0, 128, 0}
\definecolor{lightyellow}{RGB}{255, 255, 200}
\definecolor{purple}{RGB}{128, 0, 128}

\lstdefinestyle{all}
    {basicstyle=\ttfamily\scriptsize,
     keywordstyle=\color{blue}\ttfamily\scriptsize,
     commentstyle=\color{green}\ttfamily\scriptsize,
     stringstyle=\color{red}\ttfamily\scriptsize}

\lstdefinelanguage{hask}{%
    frame=none,
    xleftmargin=2pt,
    belowcaptionskip=\bigskipamount,
    captionpos=b,
    tabsize=2,
    numbers=left,
    numberstyle=\tiny\color{gray},
    emphstyle={\bf},
	morecomment=[s][\color{green}]{\{-}{-\}},
    stringstyle=\mdseries\rmfamily,
    commentstyle=\color{green},
    keywords={},
    keywords=[1]{case, of, data, if, then, else, where, let, in, do},
    keywords=[2]{Chip, Config, CurrencySymbol, TokenName, PubKeyHash, Integer, Value, State, Action, Text, Maybe, Void, TxConstraints,  Contract},
    keywords=[3]{HasNative},
    keywords=[4]{=>},
    keywords=[5]{Just, Nothing, MkChip, MkConfig, SetPrice, Buy},
    keywordstyle=[1]\mdseries\sffamily\color{red},
    keywordstyle=[2]\mdseries\sffamily\color{blue},
    keywordstyle=[3]\mdseries\sffamily\color{green},
    keywordstyle=[4]\mdseries\sffamily,
    keywordstyle=[5]\mdseries\sffamily\color{purple},
    columns=flexible,
    basicstyle=\small\sffamily,
    showstringspaces=false,
    breaklines=false,
    showspaces=false,
    escapeinside={--}{\^^M},escapebegin={\color{green}--},escapeend={},
    literate= {+}{{$+$}}1 {/}{{$/$}}1 {*}{{$*$}}1 {=}{{$=$}}1
              {>}{{$>$}}1 {<}{{$<$}}1 {\\}{{$\lambda$}}1
              {\\\\}{{\char`\\\char`\\}}1
              {->}{{$\rightarrow$}}2 {>=}{{$\geq$}}2 {<-}{{$\leftarrow$}}2
              {<=}{{$\leq$}}2 {=>}{{$\Rightarrow$}}2
              {\ .}{{$\circ$}}2 {\ .\ }{{$\circ$}}2
              {>>}{{>>}}2 {>>=}{{>>=}}2
              {|}{{$\mid$}}1
              {\_}{{\underline{\hspace{2mm}}}}2
}

\lstdefinelanguage{solidity}{%
    frame=none,
    xleftmargin=2pt,
    belowcaptionskip=\bigskipamount,
    captionpos=b,
    tabsize=2,
    numbers=left,
    numberstyle=\tiny\color{gray},
    emphstyle={\bf},
	morecomment=[s][\color{green}]{\{-}{-\}},
    stringstyle=\mdseries\rmfamily,
    commentstyle=\color{green},
    keywords={},
    keywords=[1]{pragma, solidity, contract, event, constructor, require, function, return, emit},
    keywords=[2]{address, uint, mapping},
    keywords=[3]{public, payable, external, view, returns},
    keywords=[4]{=>, +=, -=, =, <=, ==},
    keywords=[5]{msg, sender, transfer, value},
    keywordstyle=[1]\mdseries\sffamily\color{red},
    keywordstyle=[2]\mdseries\sffamily\color{blue},
    keywordstyle=[3]\mdseries\sffamily\color{green},
    keywordstyle=[4]\mdseries\sffamily,
    keywordstyle=[5]\mdseries\sffamily\color{purple},
    columns=flexible,
    basicstyle=\small\sffamily,
    showstringspaces=false,
    breaklines=false,
    showspaces=false,
    escapeinside={--}{\^^M},escapebegin={\color{green}--},escapeend={},
    literate= {+}{{$+$}}1 {/}{{$/$}}1 {*}{{$*$}}1 {=}{{$=$}}1
              {>}{{$>$}}1 {<}{{$<$}}1 {\\}{{$\lambda$}}1
              {\\\\}{{\char`\\\char`\\}}1
              {->}{{$\rightarrow$}}2 {>=}{{$\geq$}}2 {<-}{{$\leftarrow$}}2
              {<=}{{$\leq$}}2 {=>}{{$\Rightarrow$}}2
              {\ .}{{$\circ$}}2 {\ .\ }{{$\circ$}}2
              {>>}{{>>}}2 {>>=}{{>>=}}2
              {|}{{$\mid$}}1
              {\_}{{\underline{\hspace{2mm}}}}2
}

\newcommand\mbot{\mathsf{e}}
\newcommand\mtop{\mathsf{f}}

\newcommand\deffont[1]{{\bfseries #1}}
\newcommand\powerset{\f{pow}}

\newcommand\finpow{\f{fin}}

\newcommand\f[1]{\mathit{#1}}
\newcommand\tf[1]{\mathsf{#1}}
\newcommand\ns[1]{\mathsf{#1}}
\newcommand\finsubseteq{\mathbin{\subseteq_{\text{\it fin}}}}
\newcommand\at{\text{@}}
\newcommand\tx{\f{tx}}
\newcommand\ty{\f{ty}}
\newcommand\ctx{\f{ctx}}
\newcommand\txs{\f{txs}}
\newcommand\utxi{\f{utxi}}
\newcommand\utxo{\f{utxo}}
\newcommand\stx{\f{stx}}

\newcommand\liff{\Longleftrightarrow}

\newcommand\ssm{{{:}\text{=}}}

\newcommand\plus{{\text{+}}}

\newcommand\pos{\f{pos}}
\newcommand\posi{\f{posi}}
\newcommand\atoms{{\mathbb A}}
\newcommand\Forall[1]{\forall #1.}
\newcommand\Exists[1]{\exists #1.}
\usepackage{graphics} 

\newcommand\supp{\f{supp}}
\newcommand\lmodel{[\hspace{-0.2em}[}
\newcommand\rmodel{]\hspace{-0.2em}]}
\newcommand\model[1]{{\lmodel #1 \rmodel}}

\newcommand\lF{F}
\newcommand\rG{G}
\newcommand\lf{f}
\newcommand\rg{g}
\newcommand\mone{{\text{-}1}}

\makeatletter
\newcommand*\bigcdot{\mathpalette\bigcdot@{.5}}
\newcommand*\bigcdot@[2]{\mathbin{\vcenter{\hbox{\scalebox{#2}{$\m@th#1\bullet$}}}}}
\makeatother

\newcommand\pact{{\cdot}}
\newcommand\mact{\mathbin{\bigcdot}}

\begin{document}
\title{Algebras of UTxO blockchains}
\newcommand\titlerunning{\emph{Algebras of UTxO blockchains}}
\newcommand\authorrunning{\emph{Murdoch J. Gabbay}}
\author{Murdoch J. Gabbay}
\begin{abstract}
We condense the theory of UTxO blockchains down to a simple and compact set of four type equations (Idealised EUTxO), and to an algebraic characterisation (abstract chunk systems), and exhibit an adjoint pair of functors between them.
This gives a novel account of the essential mathematical structures underlying blockchain technology, such as Bitcoin. 
\end{abstract}
\maketitle
\thispagestyle{empty}

\tableofcontents

\section{Introduction}

Blockchain is a young field --- young enough that no consensus has yet developed as to its underlying mathematical structures. 
There are many blockchain implementations, but what (mathematically) are they implementations \emph{of}?

Two major blockchain architectures exist: 
\begin{itemize*}
\item
\emph{UTxO-based} blockchains (like Bitcoin), and 
\item
\emph{accounts-based} blockchains (like Ethereum).
\end{itemize*}
We consider UTxO-style blockchains in this paper, and specifically the \emph{Extended} UTxO-style model~\cite{chakravarty:extum}, which as the name suggests extends the UTxO structure (how, is described in Remark~\ref{rmrk.utxo}) of Bitcoin, which is still the canonical blockchain application.

So our question becomes: what, mathematically speaking, is an EUTxO blockchain?

In the literature, Figure~3 of~\cite{chakravarty:extum} exhibits Extended UTxO as an inductive datatype, designed with implementation and formal verification in mind.
Most blockchains exist only in code, so to have an inductive specification to work from in a published academic paper is a luxury for which we can be grateful.

However, this does not answer our question. 

It would be like answering \emph{``What are numbers?''} with the inductive definition $\mathbb N= 1 + \mathbb N$:
this an important structure (and to be fair, it yields an important inductive principle) but this does not tell us that $\mathbb N$ is a ring; or about primes and the fundamental theorem of arithmetic; or that $\mathbb N$ is embedded in $\mathbb Q$ and $\mathbb R$; or even about binary representations. 
In short: Figure~3 of~\cite{chakravarty:extum} 
gives us the raw data structure of one particular blockchain implementation, which is certainly important, but this is not a \emph{mathematics of blockchains}.

As we shall see, there is more to be said here.

\subsection{Map of the paper}
\label{subsect.map}

A map of this paper, and our answer, is as follows:
\begin{enumerate}
\item
In Section~\ref{sect.ieutxo} we present \emph{Idealised EUTxO} (Definition~\ref{defn.ieutxo.equations}), which is four type equations (Figure~\ref{fig.ieutxo}).

This captures the essence of~\cite[Figure~3]{chakravarty:extum}, but far more succinctly --- four lines vs. one full page.\footnote{This is not a criticism of the original inductive definition.}

So: EUTxO is a solution to the IEUTxO equations in Figure~\ref{fig.ieutxo}. 
\item
The approach to blockchain in this paper is novel --- we concentrate not on \emph{blockchains} but on \emph{blockchain segments}, which we call \emph{chunks} (Definition~\ref{defn.chunk}).

Chunks have many properties that blockchains do not have: if you cut a blockchain into pieces you get chunks, not blockchains; and chunks have more structure, e.g. they form a partially-ordered partial monoid (Theorem~\ref{thrm.popm}) which communicate across \emph{channels} (much like the $\pi$-calculus~\cite{Milner:comms}) and they display resource separation properties reminiscent of known systems such as separation logic (Remark~\ref{rmrk.separation.logic}).

A blockchain is the special case of a chunk with no active input channels (Definition~\ref{defn.blockchain}).

So: EUTxO is a system of chunks (a partially-ordered partial monoid with channels).
\item
IEUTxO models form a category (Definition~\ref{defn.ieutxo.category}).

So now EUTxO is the category of partially-ordered partial monoid solutions to the IEUTxO models, and arrows between them. 
\item
Our answers are still quite concrete, in the sense that objects are solutions to type equations.
To go further, we use algebra.

We introduce \emph{abstact chunk systems} (Definition~\ref{defn.acs.system}), which are oriented atomic monoids of chunks (Definitions~\ref{defn.oriented}, \ref{defn.atomic}, and~\ref{defn.monoid.of.chunks}).
These too form a category (Definition~\ref{defn.acs}), with objects and arrows.

Thus, we extract relevant properties of what makes solutions to the IEUTxO equations in Figure~\ref{fig.ieutxo} \emph{interesting}, as explicit and testable algebraic properties (see the discussion in Subsection~\ref{subsect.tests}).

Several design choices exist in this space: we discuss some of them in Remarks~\ref{rmrk.why.factorisation} and~\ref{rmrk.factor} and Proposition~\ref{prop.why.pure}.

So now EUTxO is a bundle of abstract algebraic, testable properties, which exists in a clean design space which could be explored in future work. 
\item
Finally, we pull this all together by constructing functors between the categories of IEUTxO models and of abstract chunk systems (Definitions~\ref{defn.lF} and~\ref{defn.rG}), and we exhibit a cycle of categorical embeddings between them (Theorem~\ref{thrm.adjoints}).

So finally, EUTxO becomes a pair of categories --- one of concrete solutions to some type equations, and this embedded in a category of abstract algebras --- related by adjoint functors mapping between them; or if the reader prefers, it becomes a loop of embeddings (illustrated in Remark~\ref{rmrk.loop.of.embeddings}), cycling between the concrete and abstract algebras.
\end{enumerate}

There are many definitions and results in this paper, and this leads to a broader point of it: that the mathematics we observe is \emph{possible}.
There is a mathematics of blockchain here which we have not seen commented on before.
 
This may help make blockchain more accessible and interesting to a mathematical audience, and improve communication --- since there is no more effective language for handling complexity than mathematics.
Furthermore, as we argue in Subsections~\ref{subsect.brief.discussion} and~\ref{subsect.tests}, our analysis of blockchain structure in this paper is not just of interest to mathematicians --- it may also be of practical interest to programmers --- by suggesting ways to structure and transform code, and establishing properties for unit tests, property-based testing, and formal verification of correctness --- and to designers of new UTxO-based systems, wishing to attain good design and security by working from a (relatively compact) mathematical reference model. 

More exposition and discussion is in the body of the paper and in Section~\ref{sect.discussion}, including discussions of future work.\footnote{Tip: search the pdf for `future work'.}

\section{Some background}
\label{sect.background}

\subsection{What this paper is (not) about}
\label{subsect.what.its.about}

\emph{We are parametric across possible data.}

Blockchains are best-known as stores of value but it is widely appreciated in the industry that a blockchain is just a particular kind of distributed database, and that the data stored on it, and consistency conditions imposed on its transactions, can vary with the application. 

This paper (in common with many real-life blockchain implementations) is parametric in the type of data stored on it.
Specifically, we include as a parameter an uninterpreted type $\beta$ of `data' in our `Idealised IEUTxO', as a user-determined black box; and we admit a choice of admissible transactions and validators (this is the subset inclusions in the last two lines of Figure~\ref{fig.ieutxo}).\footnote{E.g. currencies and smart contracts are eminently compatible with our models --- and most likely other as-yet-unimagined extensions too.  Algebra is good at accommodating examples; how many boolean algebras, groups, or vector spaces does the average person encounter in a lifetime (whether they know it or not)?  So if this maths stimulates the reader to think \emph{``Obviously this model could also accommodate \emph{X}; why did the author not mention this?''}, then everything will be just as it should be.}

\emph{We do not consider networks or security.}

Real blockchains rightly work hard to be efficient and secure. 
In the real world we hash data, network latency matters, as do good cryptography, incentives, permissions, user training, passwords, privacy, and more. 

Again we are parametric in these concerns.
We include an uninterpreted type $\alpha$ of `keys', but do not force any particular cryptographic content on them.\footnote{It's not that we don't care, or think these issues are trivial.  But consider: Java disallows pointer arithmetic and presents the user with an abstraction of infinite memory addressed by uninterpreted abstract pointers.  This does not imply that Java programmers believe that memory actually \emph{is} infinite, or that RAM is not linearly addressed.  It's just more productive --- and also actually safer --- to handle memory-management as a distinct design issue.}

\emph{We do have validators, but we elide their computational content.}

The UTxO model of adding a block to a chain is that the chain has `unspent outputs' --- meaning output ports that have not yet engaged in interaction --- and at each unspent output $o$ is located some \emph{data} $d$ and a \emph{validator} $v$.

A \emph{validator} is a machine that takes data-key pairs as input, and decides whether they are `good' or `bad'.
Good data-key pairs let the user interact with the blockchain; bad ones get rejected.
In more detail, to append a block to the blockchain, we 
\begin{itemize*}
\item
find a validator $v$ on an unspent output $o$ with data $d$, and 
\item
present it with a suitable key $k$ such $v$ considers $(d,k)$ to be `good'; 
\end{itemize*}
the output is now considered `spent' and our block is attached to the chain.\footnote{In practice, the code of $d$ and $v$ is usually public and can be read directly off the blockchain.  The cleverness of cryptography is to devise systems such that from $d$ and $v$ it is not trivial to deduce a $k$ such that $v$ will accept $(d,k)$ --- unless you already know a cryptographic secret.}

This is much like putting a key into a lock to open a door, except that appending a block is irreversible: an output, once spent, cannot be un-spent.\footnote{When we write `irreversible' what we mean in real life is so-called \emph{probabilistic finality}, that the practical chances of a block append getting reversed decreases rapidly and exponentially as other blocks are added after it.}
The idea is that our newly-attached block may introduce fresh outputs with validators which we designed and to which we have the keys;%
\footnote{The case where we deliberately attach a block that consumes an output and introduces no fresh outputs on the new block to replace it, is in some contexts called a \emph{burn}, because its effect is to destroy unspent outputs without replacing them with new ones.} 
these are new `unspent' outputs on the chain.%
\footnote{So for instance, if we had an oracle that could solve the cryptographic puzzles currently attached to validators on bitcoin, then we would be able to `spend' bitcoin by attaching new blocks to the chain.  As discussed, in the real world these puzzles are designed to be practically impossible to solve unless you have the key.  Fun fact: somewhere in a rubbish dump in Newport, UK is a hard drive with the keys to 7,500 bitcoin.}

We do model validators in this paper, but mathematically --- meaning that a validator is not a code-script (as it would be in real life).
Instead, we identify a validator directly with the set of data-key pairs that it accepts.\footnote{Contrast with~\cite[Figure~3]{chakravarty:extum}, which is truer to how an efficient implementation would actually work: it assumes an ``(opaque) type of scripts'' and a generic script application function $\model{\text{-}}$ which operates on a general-purpose type $\tf{Data}$ of data, into which it is assumed that inputs get encoded.

The validators-as-graphs-of-function paradigm of this paper rests on two basic observations: 1. we can identify a program-script with the function that it computes (then throw away the script and keep the function); and 2. a function can be represented as its \deffont{graph} $\{(x,y)\mid y=f(x)\}$.
(We could further insist that a validator be a \emph{computable} graph, but we don't need to for our results here, so we don't.)

We spell out points 1. and 2. here \emph{precisely because} they may be taken as obvious by some readers, but \emph{saying this} appears to be novel in the peer-reviewed blockchain literature; and I know from conversations that either point may be novel to some readers.}
\begin{rmrk}
\label{rmrk.in.summary}
In summary:
\begin{quote}
\emph{Idealised (E)UTxO is a mathematical model of the structural act of UTxO block combination and validation}
\end{quote}
and
\begin{quote}
\emph{Abstract chunk systems are an algebraic rendering of the same idea.}
\end{quote}

This basic idea of block combination and validation might not seem like much.\footnote{\dots neither did `plus one' when I learned to count in school, and look what comes of \emph{that} in university.  Basic ideas can be tricky like that.} 
However, this is \emph{the} fundamental operation of the UTxO architecture, and we shall find many interesting and unexpected observations to make about it.

So if there is one question that this paper addresses, it is this: 
\begin{quote}
\emph{What does it mean, mathematically, when we append a transaction to a UTxO blockchain?}
\end{quote}
\end{rmrk}

Our answer occupies the rest of this paper.

In the rest of this Section we will set up some basic mathematical machinery.
The reader is welcome to skip or skim it, and refer back to it as required.

\subsection{Basic data structures}

\begin{defn}\leavevmode
\label{defn.zfa.atoms}
\begin{enumerate}
\item
Fix a countably infinite set $\atoms=\{a,b,c,p,\dots\}$ of \deffont{atoms}.
\item
A \deffont{permutation} is a bijection on $\atoms$; write $\pi,\pi'\in\f{Perm}$ for permutations.
\end{enumerate}
\end{defn}

\begin{rmrk}
Following the ideas in~\cite{gabbay:equzfn,gabbay:newaas-jv} atoms will be the atoms of ZFA of Zermelo-Fraenkel set theory with Atoms\footnote{Also called \emph{urelemente} or \emph{urelements} in the set-theoretic literature.} --- this is a fancy way to say that $\atoms$ is a type of atomic identifiers. 

We will use atoms to locate inputs and outputs on a blockchain.
More on this in Subsection~\ref{subsect.permutations}.
\end{rmrk}

\begin{nttn}\leavevmode
\label{nttn.pointed}
\begin{enumerate*}
\item
Write $\mathbb N=\{0,1,2,\dots\}$.
\item\label{pointed.finite.set}
If $\ns X$ is a set, write $\finpow(\ns X)$ for the finite powerset of $\ns X$, and $\finpow_!(\ns X)$ for the \deffont{pointed finite powerset}. 
In symbols: 
$$
\finpow_!(\ns X) = \{ (X,x)\in \finpow(\ns X)\times \ns X\mid x\in X\}.
$$
Above, $(X,x)$ is a pair, and $\finpow(\ns X)\times\ns X$ is a cartesian product.
\item\label{pointed.functorial}
If $\ns X$ and $\ns Y$ are sets, we use a convenient shorthand in Figure~\ref{fig.ieutxo} by writing
$$
(\finpow(\ns X),\ns Y)_!
\quad\text{as shorthand for}\quad 
(\finpow(\ns X)_!,\ns Y) .
$$
That is, we take $(\text{-})_!$ to act on a pair functorially, on the first component.
We do this in Figure~\ref{fig.ieutxo} when we write $\tf{Transaction}_!$ in the definition of $\tf{Validator}$.
\item\label{list.concat}
If $\ns X$ is a set then write $[\ns X]$ for the set of (possibly empty) finite lists of elements from $\ns X$.
We write $\mact$ for list concatenation, so $l\mact l'$ is $l$ followed by $l'$.

More generally, we will write $\mact$ for any monoid composition; list concatenation is one instance.
It will always be clear what is intended.
\item\label{sublist.inclusion}
If $\ns X$ is a set then order $l,l'\in [\ns X]$ by the \deffont{sublist inclusion} relation, where $l\leq l'$ when $l$ can be obtained from $l'$ by deleting (but not rearranging) some of its elements. 
\item
If $\ns X$ is a set and $x\in\ns X$ then we may call the one-element list $[x]\in[\ns X]$ a \deffont{singleton}.
\item
If $V\in \ns X\to\tf{Bool}$ and $x\in\ns X$ then we may write $V(x)$ or $V\,x$ for $V\,x = \tf{True}$. 
\end{enumerate*}
\end{nttn}

\subsection{The permutation action}
\label{subsect.permutations}

\begin{rmrk}
We spend this Subsection introducing permutations and their action on elements.
We will need this most visibly in two places:
\begin{enumerate}
\item
To state the key Definition~\ref{defn.posi}.
\item
To prove Lemma~\ref{lemm.perm.utxi}, and thus Proposition~\ref{prop.pos.eq.posi}.
\end{enumerate}
Because we assume atoms and are working in a ZFA universe, everything has a standard permutation action.
We describe it in Definition~\ref{defn.perm}.
Programmers can think of the permutation action as a \emph{generic} definition in the ZFA universe (given below in this Remark), which is sufficiently generic that it exists for all the datatypes considered in this paper.
By this perspective, Definition~\ref{defn.perm} specifies how this generic action interacts with the specific type-formers of interest for this paper.
\end{rmrk}

\begin{defn}
\label{defn.zfa.perm}
For reference we write out the ZFA generic definition, which is by $\epsilon$-induction on the sets universe:
$$
\begin{array}{r@{\ }l@{\qquad}l}
\pi\pact a=&\pi(a) & a\in\atoms
\\
\pi\pact X =& \{\pi\pact x\mid x\in X\} &X\text{  a set}.
\end{array}
$$ 
\end{defn}
More information on this sets inductive definition is in~\cite{gabbay:equzfn,gabbay:thesis}.
Definition~\ref{defn.perm} can be usefully viewed as a collection of concrete instances of Definition~\ref{defn.zfa.perm} for the datatypes of interest in this paper:

\begin{defn}
\label{defn.perm}
Permuations $\pi$ act concretely as follows:
\begin{enumerate}
\item
If $\pi\in\f{Perm}$ and $a\in\mathbb A$ then $\pi$ acts on $a$ as a function:
$$
\pi\pact a=\pi(a) .
$$
\item
If $\pi\in\f{Perm}$ and $X$ is any set then $\pi$ acts \deffont{pointwise} on $X$ as follows:
$$
\pi\pact X=\{\pi\pact x\mid x\in X\} .
$$
Note as a corollary of this that $x\in X \liff \pi\pact x\in\pi\pact X$.
\item
If $\pi\in\f{Perm}$ and $(x_1,\dots,x_n)$ is a tuple then $\pi$ acts \deffont{pointwise} on $(x_1,\dots,x_n)$ as follows:
$$
\pi\pact(x_1,\dots,x_n) = (\pi\pact x_1,\dots,\pi\pact x_n) .
$$ 
Note therefore that $\pi\pact((x_1,\dots,x_n)!!i)=\pi\pact x_i$, where $1\leq i\leq n$ and $!!$ indicates lookup.
\item
If $\pi\in\f{Perm}$ and $i\in\mathbb N$ then $\pi\pact i = i$.
\item
If $\pi\in\f{Perm}$ and $(X,x)$ is a pointed set (Notation~\ref{nttn.pointed}(\ref{pointed.finite.set}))
then $\pi$ acts \deffont{pointwise} on $(X,x)$ as follows:
$$
\pi\pact (X,x) = (\pi\pact X,\pi\pact x) .
$$
(This is indeed just a special case of the previous case, for tuples.)
\item
If $\pi\in\f{Perm}$ and $f$ is a function then $\pi$ has the \deffont{conjugation action} on $f$ as follows:
$$
(\pi\pact f)(x) = \pi\pact(f(\pi^\mone\pact x)) .
$$
Note therefore that $\pi\pact (f(x))=(\pi\pact f)(\pi\pact x)$, and $\pi\pact f$ maps $\pi\pact x$ to $\pi\pact (f(x))$.
\item
In particular, $\pi$ acts as the above on the inputs, outputs, sets of inputs, sets of outputs, transactions, and validators from Figure~\ref{fig.ieutxo}.
\item
If $\pi\in\f{Perm}$ and $R$ is a relation, then $\pi$ acts \deffont{pointwise} such that
$$
\pi\pact R = \{(\pi\pact x,\pi\pact y)\mid (x,y)\in R\}
$$
so that
$$
x\mathrel{\pi\pact R} y
\liff
\pi^\mone\pact x\mathrel{R} \pi^\mone\pact y .
$$
\end{enumerate}
\end{defn}

We use Definition~\ref{defn.fix} in Definition~\ref{defn.posi}, but it is a useful concept so we include it here:
\begin{defn}
\label{defn.fix}
If $a\in\atoms$ then write $\f{fix}(a)$ for the set of permutations $\pi\in\f{Perm}$ such that $\pi(a)=a$.
In symbols:
$$
\f{fix}(a) = \{\pi\in\f{Perm} \mid \pi(a)=a\} .
$$
\end{defn}

\begin{defn}
\label{defn.equivariant}
Call an element \deffont{equivariant} when $\pi\pact x=x$ for every $\pi\in\f{Perm}$.
Concretely:
\begin{enumerate}
\item
$\atoms$ is equivariant, and no individual atom $a\in\atoms$ is equivariant.
\item\label{equivariant.set}
A set $X$ is equivariant when 
$$
\Forall{\pi{\in}\f{Perm}}(x\in X \liff \pi\pact x\in X).
$$
In words: 
\begin{quote}
\emph{A set is equivariant precisely when it is closed in the orbits of its elements under the permutation action.}
\end{quote}
\item
$(x_1,\dots,x_n)$ is equivariant precisely when $x_i$ is equivariant for every $1\leq i\leq n$.
\item
$\mathbb N$ is equivariant, and every $i\in\mathbb N$ is equivariant.
\item
A pointed set $(X,x)$ is equivariant precisely when $X$ and $x$ are equivariant.
\item\label{equivariant.function}
A function $f$ is equivariant when $\pi\pact (f(x))=f(\pi\pact x)$ for every $x$ and every $\pi$.
\item
A relation $R$ is equivariant when 
$$
\Forall{\pi{\in}\f{Perm}}\Forall{x,y}(x\mathrel{R} y \liff \pi\pact x\mathrel{R}\pi\pact y) .
$$
\end{enumerate}
\end{defn}

\section{Idealised EUTxO: \texorpdfstring{$\tf{IEUTxO}$}{IEUTxO}}
\label{sect.ieutxo}

\subsection{IEUTxO equations and solutions}

\begin{defn}
\label{defn.ieutxo.equations}
Let \deffont{idealised EUTxO} be the type equations in Figure~\ref{fig.ieutxo}\ ($\tf{Transaction}_!$ is from Notation~\ref{nttn.pointed}(\ref{pointed.functorial})).
\end{defn}

\begin{figure}[tp]
$$
\begin{array}{r@{\ }l}
\tf{Input}=&\atoms\times\alpha  
\\
\tf{Output}=&\mathbb A\times\beta\times\tf{Validator}  
\\
\tf{Transaction}\subseteq&\bigl(\finpow(\tf{Input})\times\finpow(\tf{Output})\bigr)\setminus\{(\varnothing,\varnothing)\}
\\
\tf{Validator}\subseteq&\powerset(\beta\times\tf{Transaction}_!)  
\end{array}
$$
\caption{Type equations of Idealised EUTxO}
\label{fig.ieutxo}
\end{figure}

\begin{defn}
\label{defn.solution}
Let a \deffont{solution} or \deffont{model} 
of the IEUTxO type equations in Figure~\ref{fig.ieutxo} be a tuple 
$$
\mathbb T=(\alpha,\beta,\tf{Transaction},\tf{Validator},\nu : \tf{Validator}\hookrightarrow \powerset(\beta\times\tf{Transaction}_!))
$$
where:
\begin{enumerate*}
\item
$\alpha$ is an uninterpreted\footnote{\emph{Uninterpreted} means: ``Pick whatever you want.  We'll never look inside it so we don't care.  Just make sure this set is large enough and has whatever structure that \emph{you} need for \emph{your} application, and the maths will work for that choice.''} equivariant (Definition~\ref{defn.equivariant}(\ref{equivariant.set})) set of \emph{keys}.
\item
$\beta$ is an uninterpreted equivariant set of \emph{data}.
\item
$\tf{Validator}$ is an equivariant set of \emph{validators}.
\item\label{transaction.subset}
$\tf{Transaction}$ is an equivariant subset of $\finpow(\atoms\times\alpha)\times\finpow(\atoms\times\beta\times\tf{Validator})$ and we disallow the empty transaction, having no inputs or outputs. 
\item\label{validator.injection}
$\nu$ is an equivariant injective function (Definition~\ref{defn.equivariant}(\ref{equivariant.function})) from $\tf{Validator}$ to $\powerset(\beta\times\tf{Transaction}_!)$.
\end{enumerate*} 
\end{defn}


\newcommand\nameOaa{$a$}
\newcommand\nameOab{$b$}
\newcommand\nameOac{$c$}
\newcommand\nameIba{\nameOab}
\newcommand\nameObb{$d$}
\newcommand\nameIca{\nameOaa}
\newcommand\nameOca{$e$}
\newcommand\nameOcb{$f$}
\newcommand\nameOcc{$g$}
\newcommand\nameIda{\nameOca}
\newcommand\nameIdb{\nameOcb}
\newcommand\nameIdc{\nameObb}
\newcommand\nameOda{$h$}
\newcommand\nameOdb{$i$}
\newcommand\nameOdc{$j$}
\newcommand\nameOdd{$k$}

\newcommand\drawCircle[2]{
    \path[fill, #1]    (#2) circle[radius=0.3];
    \path[fill, white] (#2) circle[radius=0.15];
}
\newcommand\drawBox[2]{
        \draw[ultra thick, fill=lightgray]
            ($#1 + (-1,-1)$)
            rectangle ($#1 + (1,1)$);
        \node at #1 {$#2$};
}
\newcommand\drawLine[4]{
        \draw[ultra thick] (#1) to[out=0, in=180] node[#2]{#3} (#4);
}
\newcommand\drawLineToCircle[5]{
        \drawLine{#1}{#2}{#3}{#4}
        \drawCircle{#5}{#4}
}
\newcommand\drawLineFromCircle[5]{
        \drawLine{#1}{#2}{#3}{#4}
        \drawCircle{#5}{#1}
}

\newcommand\mkBackbone{
        \coordinate (a) at (0*\ma,0);    
        \coordinate (d) at (3*\ma,0);    

        \coordinate (ab1) at (\mb, \mc);  
        \coordinate (ab2) at (\mb, 0);  
        \coordinate (ab3) at (\mb,-\mc);  

        \coordinate (bc1) at (1.5*\ma,0.5*\mc); 

        \coordinate (de1) at (3.5*\ma, \mc);
        \coordinate (de2) at (3.5*\ma, 1/3*\mc);
        \coordinate (de3) at (3.5*\ma,-1/3*\mc);
        \coordinate (de4) at (3.5*\ma, -\mc);
        
        \coordinate (x2) at (2.5*\ma,-\mc);

        \drawBox{(a)}{\tx_1}
        \drawBox{(b)}{\tx_2}
        \drawBox{(c)}{\tx_3}
        \drawBox{(d)}{\tx_4}
        
        \drawLine{$(a) + (1, 0.5)$}{above, xshift=-1, yshift=1}{\nameOaa}{ab1}
        \drawLineToCircle{$(d) + (1, 0.6)$}{above}{\nameOda}{de1}{red}
        \drawLineToCircle{$(d) + (1, 0.2)$}{above right}{\nameOdb}{de2}{red}
        \drawLineToCircle{$(d) + (1,-0.2)$}{below right}{\nameOdc}{de3}{red}
        \drawLineToCircle{$(d) + (1,-0.6)$}{below}{\nameOdd}{de4}{red}
        \drawLine{$(c) + (1, 0.5)$}{above, xshift=-1}{\nameOca}{cd1}
        \drawLineToCircle{$(c) + (1,-0.5)$}{below, xshift=-1}{\nameOcc}{cd3}{red}

}
\newcommand\mkBackboneBC[1]{
        \coordinate (bc2) at (1.5*\ma-#1*\ma,-\mc);    
        \coordinate (bc3) at (1.5*\ma, 0);

        \coordinate (cd4) at (2.5*\ma+#1*\ma,-\mc);
        
        \coordinate (cd1) at (2.5*\ma, \mc);  
        \coordinate (cd2) at (2.5*\ma, 0);    
        \coordinate (cd3) at (2.5*\ma-#1*\ma,-2/3*\mc+#1*\mc);    

        \coordinate (b) at (1*\ma,0);    
        \coordinate (c) at (2*\ma,0);    
        \coordinate (x1) at (1.5*\ma+#1*\ma,\mc);     

        \mkBackbone

        \drawLineToCircle{$(a) + (1,-0.5)$}{below, xshift=-1, yshift=-1}{\nameOac}{ab3}{red}

        \drawLine{ab2}{above, yshift=-1}{\nameIba}{$(b) + (-1, 0)$}
        \drawLineToCircle{$(a) + (1, 0  )$}{above, yshift=-1}{\nameOab}{ab2}{blue}
   
        \drawLine{$(c) + (1, 0  )$}{above, yshift=-1.5}{\nameOcb}{cd2};

}
\newcommand\mkBackboneCB[1]{
        \coordinate (bc2) at (1.5*\ma,-\mc);    

        \coordinate (bd1) at (2.5*\ma, \mc);
        \coordinate (bd2) at (2.5*\ma,-\mc);
        \coordinate (bd3) at (2.5*\ma, 3/5*\mc);
        \coordinate (bd4) at (2.5*\ma,-1/3*\mc);
 
        \coordinate (cd1) at (1.5*\ma-#1*\ma, \mc);
        \coordinate (cd2) at (1.5*\ma-#1*\ma, 3/5*\mc);
        \coordinate (cd3) at (1.5*\ma-#1*\ma,-2/5*\mc);
        
        \coordinate (c) at (1*\ma,0);       
        \coordinate (b) at (2*\ma,0);       
        \coordinate (x1) at (1.5*\ma,\mc);  

        \mkBackbone

        \drawLineToCircle{$(a) + (1,-0.5)$}{below, xshift=-1, yshift=-1}{\nameOab}{ab3}{red}
        
        \drawLineToCircle{$(a) + (1, 0  )$}{above, yshift=-1}{\nameOac}{ab2}{red}

        \drawLine{$(c) + (1, 0  )$}{below, yshift=1}{\nameOcb}{cd2}

        \coordinate (cb1) at (1.5*\ma+#1*\ma, -1*\mc);
        \drawLine{cb1}{above, xshift=-2}{\nameIba}{$(b) + (-1, 0)$}
}

\begin{figure}[b]
    \centering
    \newcommand\ma{10}
    \newcommand\mb{5}
    \newcommand\mc{3}
    \begin{tikzpicture}[scale=0.25]
        \coordinate (tx) at (  0, 0);
        \coordinate (i1) at (-\ma, \mb);
        \coordinate (i2) at (-\ma, 0);
        \coordinate (i3) at (-\ma,-\mb);
        \coordinate (o1) at (0.95*\ma, \mc);
        \coordinate (o2) at (0.95*\ma,-\mc);

        \drawBox{(tx)}{tx}

        \drawLineFromCircle{i1}{above right}{$(a,x_1)$}{$(tx) + (-1, 0.5)$}{red}
        \drawLineFromCircle{i2}{above, xshift=-10}{$(b,x_2)$}{$(tx) + (-1, 0  )$}{red}
        \drawLineFromCircle{i3}{below right}{$(c,x_3)$}{$(tx) + (-1,-0.5)$}{red}
        \drawLineToCircle{$(tx) + (1, 0.3)$}{below right, xshift=2, yshift=5}{$(d,y_1,v_1)$}{o1}{red}
        \drawLineToCircle{$(tx) + (1,-0.3)$}{above right, xshift=2, yshift=-5}{$(e,y_2,v_2)$}{o2}{red}

        \coordinate (ty) at (2.00*\ma, 0);
        \coordinate (i4) at (1.05*\ma, \mc);
        \coordinate (i5) at (1.05*\ma,-\mc);
        \coordinate (o3) at (3.00*\ma, \mc);

        \drawBox{(ty)}{ty}

        \drawLineFromCircle{i4}{above right, xshift=-10}{$(d,x_4)$}{$(ty) + (-1, 0.5)$}{red}
        \drawLineFromCircle{i5}{below right, xshift=-10}{$(e,x_5)$}{$(ty) + (-1, 0  )$}{red}
        \drawLineToCircle{$(ty) + (1, 0.3)$}{below right, xshift=2, yshift=5}{$(f,y_1\plus y_2,v_1)$}{o3}{red}
    \end{tikzpicture}
    \caption{A pair of transactions $\tx$ and $\ty$} 
    \label{fig.transaction}
    \ \\[5mm]
    \begin{tikzpicture}[scale=0.25]
        \coordinate (tx) at (  0, 0);
        \coordinate (i1) at (-\ma, \mb);
        \coordinate (i2) at (-\ma, 0);
        \coordinate (i3) at (-\ma,-\mb);
        \coordinate (o1) at ( \ma, \mc);
        \coordinate (o2) at ( \ma,-\mc);

        \drawBox{(tx)}{tx}

        \drawLineFromCircle{i1}{above right}{$(a,x_1)$}{$(tx) + (-1, 0.5)$}{red}
        \drawLineFromCircle{i2}{above, xshift=-10}{$(b,x_2)$}{$(tx) + (-1, 0  )$}{red}
        \drawLineFromCircle{i3}{below right}{$(c,x_3)$}{$(tx) + (-1,-0.5)$}{red}
        \drawLineToCircle{$(tx) + (1, 0.3)$}{below right, xshift=2, yshift=5}{$(d,y_1,v_1)$}{o1}{red}
        \drawLineToCircle{$(tx) + (1,-0.3)$}{above right, xshift=2, yshift=-5}{$(e,y_2,v_2)$}{o2}{red}

        \coordinate (ty) at (2.0*\ma, 0);
        \coordinate (i4) at (1.0*\ma, \mc);
        \coordinate (i5) at (1.0*\ma,-\mc);
        \coordinate (o3) at (3.0*\ma, \mc);

        \drawBox{(ty)}{ty}

        \drawLineFromCircle{i4}{above right, xshift=-10}{$(d,x_4)$}{$(ty) + (-1, 0.5)$}{blue}
        \drawLineFromCircle{i5}{below right, xshift=-10}{$(e,x_5)$}{$(ty) + (-1, 0  )$}{blue}
        \drawLineToCircle{$(ty) + (1, 0.3)$}{below right, xshift=2, yshift=5}{$(f,y_1\plus y_2,v_1)$}{o3}{red}
    \end{tikzpicture}
    \caption{A pair of transactions $\tx$ and $\ty$, successfully validated and combined} 
    \label{fig.transaction.combined}
\end{figure}

Some notation will be helpful:
\begin{nttn}
\label{nttn.ty.points.to}
If $\tx=(I,O)\in\tf{Transaction}$ and $i\in I$ then write $\tx\at i\in\tf{Transaction}_!$ for the input-in-context $((I,i),O)$ obtained by pointing $I$ at $i\in I$ (Notation~\ref{nttn.pointed}).
In symbols:
$$
\text{if }\tx=(I,O)\in\tf{Translation}
\ \text{and}\  i\in I 
\quad\text{then}\quad
\tx\at i = ((I,i),O) \in \tf{Transaction}_!
$$
\end{nttn}

\begin{xmpl}
In Figures~\ref{fig.transaction} and~\ref{fig.transaction.combined} we take a short diagrammatic tour of Definition~\ref{defn.solution}, before continuing with the technical development.
Since we have yet to build our machinery, this discussion is informal and intuitive.
We include precise forward pointers to later definitions where appropriate; additional diagrams will be discussed in detail in Example~\ref{xmpl.example.transactions}. 

$\tx\in\tf{Transaction}$ in Figure~\ref{fig.transaction} is a pair of a set of three inputs, and a set of two outputs:
$$
\tx\ =\ \Bigl(\{(a,x_1),(b,x_2),(c,x_3)\}\ ,\ \{(d,y_1,v_1),(e,y_2,v_2)\}\Bigr) .
$$
Above:
\begin{itemize}
\item
$a$, $b$, $c$, $d$, and $e$ are atoms in $\mathbb A$.
They serve as unique labels for the inputs and outputs.

Note that every input and output in $\tx$ has a distinct label.
This is not directly enforced in the raw datatype in Figure~\ref{fig.ieutxo}, but it will follow as a consequence of well-formedness conditions which we introduce later, in Definition~\ref{defn.chunk}.

Since each input is labelled with a unique atom, and similarly for each output, we may treat atoms as \emph{positions}, \emph{locations}, or \emph{channels}.
Thus for example the input $(a,x_1)$ is labelled with $a$ and so we can think intuitively that it is located at position $a$; or waiting to communicate its data $x_1$ on channel $a$.
\item
$x_1$ and $x_2$ are elements of $\alpha$.
This is just an uninterpreted datatype, but a good intuition is that these are cryptographic keys.
\item
$y_1$ and $y_2$ are elements of $\beta$.
This is just an uninterpreted datatype: the intuition is that $y_1$ and $y_2$ are fragments of state data.

Note that state data is stored per-output and not per-transaction.
We suppose for concreteness that $\beta=\mathbb N$, so state data can be summed (as we do in the output $f$ of $\ty$; more in this shortly).

Note that $\tf{Transaction}$ is a \emph{subset} of $\finpow(\tf{Input})\times\finpow(\tf{Output})$ in Figure~\ref{fig.transaction}; which subset, is a parameter of the model selected.
So for example we could enforce that all state data on all inputs and outputs should be equal, and this would in effect ensure that state data is a per-transaction quantity.
This may or may not be what we want: the definition allows us to choose.
\item
$v_1$ and $v_2$ are validators.
Their role is, intuitively, to decide which transactions $\tx$ will interact with (precise definition in Definition~\ref{defn.chunk}(\ref{blockchain.the.point})).
\end{itemize} 
$\ty\in\tf{Transaction}$ is another transaction, with two inputs and one output.
Note that its inputs are located at $d$ and $e$, meaning that the inputs of $\ty$ \emph{point to} the outputs of $\tx$ (Definition~\ref{defn.i.points.to.o}).
Note also that it performs an addition, in the sense that the state data of its single output is the sum of the state data on its two inputs. 

If in addition the following validation conditions are satisfied ($\at$ from Notation~\ref{nttn.ty.points.to})
\begin{itemize*}
\item
$(y_1,\ty\at (d,x_4))\in \nu(v_1)$ 
\\
(in words: validator $v_1$ in state $y_1$ validates the pointed transaction $\ty\at (d,x_4))$), and
\item
$(y_2,\ty\at (e,x_5))\in \nu(v_2)$ 
\\
(in words: validator $v_2$ in state $y_2$ validates the pointed transaction $\ty\at (e,x_5)$)
\end{itemize*}
then the two-element sequence $[\tx,\ty]$ is considered to be a valid combination.
In the terminology we define later, this is called a \emph{chunk} (Definition~\ref{defn.chunk}).

Supposing this is so.
Then we can join the two transactions as illustrated in Figure~\ref{fig.transaction.combined}, using blue circles to indicate successful validations.
Congratulations: we have performed our first blockchain concatenation.

Note that \emph{both} validations must succeed for the combination $[\tx,\ty]$ to be considered a valid chunk.
More diagrams and discussion follow in Example~\ref{xmpl.example.transactions}.
We now return to the definitions. 
\end{xmpl}

\begin{nttn}
\label{nttn.equivariant}
Sets that do not include atoms --- including $\mathbb N$, inductive types built using $\mathbb N$, and function types built using $\mathbb N$ --- are automatically equivariant. 
Thus, if the reader unfamiliar with nominal techniques and ZFA wonders whether particular choices they need for $\alpha$, $\beta$, and $\tf{Validator}$ are equivariant --- then the answer is `yes'. 

We may elide equivariance conditions Definition~\ref{defn.solution} henceforth.
Any such type-like definition will be equivariant --- i.e. closed under taking orbits of the permutation action --- unless stated otherwise.
\end{nttn}

\begin{rmrk}
Equivariance comes from the underlying ZFA universe.
Notation~\ref{nttn.equivariant} can be viewed as an assertion that the definition exists in the category of equivariant ZFA sets and equivariant functions between them (or if the reader prefers: sets with a permutation action, and equivariant functions between them).

This paper will be light on sets and categorical foundations:
we use just enough so that readers from various backgrounds get a hook on the ideas that speaks to them, and so it is always clear what is meant and how it could be made fully formal.

Note that just because a set is equivariant does not mean all its elements must be; for instance, $\mathbb A$ is equivariant (and consists of a single permutation orbit), but none of its elements $a\in\atoms$ are equivariant.
\end{rmrk}

\begin{rmrk}
Compare Figure~\ref{fig.ieutxo} with~\cite[Figure~3]{chakravarty:extum}:\footnote{These correspondences are not intended to be read as mathematically precise; they illustrate how a shared intuitive underlying structure manifests itself across the two different definitions.} 
$\alpha$ here corresponds to \emph{redeemer} there, $\beta$ here corresponds to \emph{datum} there (though $\alpha$ here lives on inputs and $\beta$ lives on outputs, whereas there both redeemer and datum live on inputs); the by-hash referencing there is replaced here with a nominal treatment using atoms; and validators exist on the output here and on the input there.
There is less to this latter difference (validator moving from input in \cite{chakravarty:extum} to output here) than meets the eye: what is key is the interaction between an output and a later input, so it a matter of perspective and our convenience whether we view the output as validating the input, or vice versa, or indeed both.
\end{rmrk}

\begin{rmrk}
\label{rmrk.NQR}
Definition~\ref{defn.solution}(\ref{transaction.subset}) uses a subset inclusion, whereas Definition~\ref{defn.solution}(\ref{validator.injection}) uses an injection $\nu$.
Why?

First, in practice we would expect $\nu(v)$ to be a \emph{computable} subset of $\beta\times\tf{Transaction}_!$, since we have implementations in mind (though nothing in the mathematics to follow will depend on this).
 
Also, sets are well-founded, so a pure subset inclusion solution to Figure~\ref{fig.ieutxo}, for both clauses, would be difficult; we use $\nu$ to break the downward chain of sets inclusions.%
\footnote{A universe of non-wellfounded sets~\cite{aczel:nonwfs} would also solve this.  That might be overkill for this paper, but this might become relevant if the theorems of this paper are implemented in a universe with more direct support for non-wellfounded objects than ZFA provides.} 

Yet just as we may write 
$$
\mathbb N\subseteq\mathbb Q\subseteq\mathbb R, 
$$
eliding (or neglecting to consider) that their realisations may differ (finite cardinals vs. equivalence classes of pairs vs. Dedekind cuts), so --- since $\nu$ is an injection --- it may be convenient to treat $\tf{Validator}$ as a literal subset of $\powerset(\beta\times\tf{Transaction}_!)$.\footnote{This can also depend on which aspects of a structure we wish to emphasise.  For example, $\mathbb N$ and $\omega$ are isomorphic, as are $\powerset(X)$ and $2^X$: the usage chosen in context suggests which aspects of the underlying entity (arithmetic vs. ordinals; sets vs. functions) we find most relevant.}
This is standard, provided we clearly state what is intended and are confident that we could unroll the injections if required.

Thus Definition~\ref{defn.ieutxo.model} rephrases Definition~\ref{defn.solution}, with a focus on extensional sets behaviour rather than on internal structure (and for reference, see Definition~\ref{defn.rG} for an example of where the internal structure is required): 
\end{rmrk}

\begin{defn}
\label{defn.ieutxo.model}
An IEUTxO model $\mathbb T$ from Definition~\ref{defn.solution} 
can be presented modulo Remark~\ref{rmrk.NQR} as a tuple 
$$
\mathbb T=(\alpha_{\mathbb T},\beta_{\mathbb T},\tf{Transaction}_{\mathbb T},\tf{Validator}_{\mathbb T})
$$ 
of sets that solves the equations in Figure~\ref{fig.ieutxo}. 
We may drop the $\mathbb T$ subscripts where these are clear from the context.
\end{defn}

\subsection{Positions}

\begin{rmrk}
The \emph{positions} of a transaction are intuitively the interfaces or channels along which it may connect with other transactions to form a chunk (see next subsection).
The notion of connection is called \emph{pointing to} (Definition~\ref{defn.i.points.to.o}).

In Figure~\ref{fig.transaction.combined}, the positions of $\tx$ and $\ty$ are $\{a,b,c,d,e\}$ and $\{d,e,f\}$ respectively, and input $d$ of $\ty$ points to output $d$ of $\tx$.
Also, $\tx$ is \emph{earlier} than $\ty$, and $\ty$ is \emph{later} than $\tx$. 
\end{rmrk}

\begin{nttn}\leavevmode
\label{nttn.points.to}
\begin{enumerate*}
\item
If $\tx\in\tf{Transaction}$ appears in $\txs\in[\tf{Transaction}]$ then write $\tx\in\txs$.
\item
If $\tx,\tx'\in\txs$ and $\tx$ appears before $\tx'$ in $\txs$, then call $\tx$ \deffont{earlier} than $\tx'$ and $\tx'$ \deffont{later} than $\tx$ (in $\txs$). 
\item
If $\tx=(I,O)\in\tf{Transaction}$ and $o\in\tf{Output}$, say $o$ \deffont{appears in} $\tx$ and write $o\in \tx$ when $o\in O$; similarly for an input $i\in\tf{Input}$.

We may silently extend this notation to larger data structures, writing e.g. $i\in\txs$ when $i\in\tx\in\txs$ for some $\tx$. 
\end{enumerate*}
\end{nttn}

\begin{figure}[t]
$$
\begin{array}{r@{\ }l@{\qquad}r@{\ }l}
i=(p,k)\in&\tf{Input} & 
\pos(i)=&\{p\} 
\\
o=(p,d,V)\in&\tf{Output} & 
\pos(o)=&\{p\}
\\
\tx=(I,O)\in&\tf{Transaction} &
\pos(\tx)=&\bigcup\{\pos(x)\mid x\in I\cup O\} 
\\
\ctx=((I,i),O)\in&\tf{Transaction}_!  &
\pos(\ctx)=&\pos(I,O) 
\\
\txs\in&[\tf{Transaction}] &
\pos(\txs)=&\bigcup\{\pos(\tx)\mid \tx\in\txs\} 
\end{array}
$$
\caption{Positions of}
\label{fig.positions}
\end{figure}

\begin{defn}
\label{defn.pos}
Suppose $\mathbb T$ is an IEUTxO model.
We define \deffont{positions of}, written $\pos$, as in Figure~\ref{fig.positions} (see also Definition~\ref{defn.inputs.and.outputs}).
\end{defn}

\begin{rmrk}
Intuitively, $\pos(x)$ collects the positions mentioned explicitly on the inputs or outputs of a structure.
Note that validators may also act depending on positions of their inputs, but this information is not detected by $\pos$.
For instance, consider a (arguably odd, but imaginable) output $o$ having the form 
$$
o=(b,0,\{i\in\tf{Input} \mid \pos(i)=\{a\}\}) .
$$ 
So this output is at position $b$, $\beta=\mathbb N$ and $o$ carries data $0$, and $o$ has a validator that validates an input precisely when it is at position $a$.
Then $\pos(o)=\{b\}$.\footnote{Note to experts in nominal techniques: so the support of validators, if any, is not counted in $\pos$.  See also the discussion in Subsection~\ref{subsect.nominal}.}
\end{rmrk}

\begin{lemm}
\label{lemm.supp.empty.empty}
Suppose $\mathbb T$ is an IEUTxO model and $\txs\in[\tf{Transaction}]$.
Then 
$$
\pos(\txs)=\varnothing
\quad\text{implies}\quad
\txs=[].
$$
\end{lemm}
\begin{proof}
Recall from Definition~\ref{defn.solution}(\ref{transaction.subset}) that the empty transaction is disallowed.
Now examining Figure~\ref{fig.positions} we see that the only way $\txs$ can mention no positions \emph{at all}, is by having no transactions and so being empty.
\end{proof}

\begin{defn}
\label{defn.i.points.to.o}
\begin{enumerate*}
\item\label{i.points.to.o}
Suppose $i\in\tf{Input}$ and $o\in\tf{Output}$.
Then say that $i$ \deffont{points to} $o$ and write $i\mapsto o$ when they share a position:
$$
i\mapsto o \quad\text{when}\quad \pos(i)=\pos(o) .
$$
(The use of `point' here is unrelated to the `pointed sets' from Notation~\ref{nttn.pointed}.)
\item\label{i.points.to.o.validates}
Recall the notation $\tx\at i$ from Notation~\ref{nttn.ty.points.to}.
Suppose that:
$$
\begin{array}{l}
i=(p,k)\in\tf{Input}
\\
i\in\tx\in\tf{Transaction}
\ \ \text{and}\ \ 
\\
o=(p,d,V)\in\tf{Output}. 
\end{array}
$$
Then write 
\begin{equation}
\label{eq.validates}
\f{validates}(o,\tx\at i)
\quad\text{when}\quad
(d,\tx\at i) \in V 
\end{equation}
and say that the output $o$ \deffont{validates} the input-in-context $\tx\at i$.
\end{enumerate*}
\end{defn}

\subsection{Why we have \texorpdfstring{$\alpha$}{alpha} and \texorpdfstring{$\beta$}{beta}}

We are now ready to more rigorously explain the roles of $\alpha$ and $\beta$ in Figure~\ref{fig.ieutxo}.
We touched on this already in Subsection~\ref{subsect.what.its.about} ($\alpha$ is `keys'; $\beta$ is `data'), but now we can be more specific in our discussion: 
\begin{itemize}
\item 
$\alpha$ is intuitively a datatype for \emph{keys}.

If $i=(a,k)$, then $k$ is supposed to be a cryptographic secret that we need to attach to a suitable unspent output (e.g. a solution to a cryptographic puzzle posted by its validator).
\item
$\beta$ is intuitively a datatype of abstract data.

If $o=(p,d,V)$ then $d$ stores some kind of state (for instance, an account balance).
\end{itemize}
We should note:
\begin{enumerate*}
\item
Nothing mathematical enforces this usage. 
$\alpha$ and $\beta$ are abstract type parameters, to instantiate as we please.
\item 
$\beta$ is redundant and can be isomorphically removed. 

We briefly sketch a suitable isomorphism: We set $\tf{Output}=\atoms\times\tf{Validator}$ and $\tf{Validator}\subseteq\powerset(\tf{Transaction}_!)$ in Figure~\ref{fig.ieutxo}.
$\alpha$ is replaced by $\alpha\times\beta$, and information that was stored as $d\in\beta$ on an output would instead be stored in the validator of that output, which would be set to accept only inputs with $\beta$-component $d$.  
\item
If we only care about countable datatypes and computable data (a reasonable simplification), then everything could be G\"odel encoded into $\mathbb N$. 
So we could then just fix $\alpha=\mathbb N$. 
\end{enumerate*}
But there is a design tension here: we want something that is compact, but also implementable (e.g. as proof-of-concept reference code; see Remark~\ref{rmrk.hypothetical}).
Having explicit types of keys $\alpha$ and data $\beta$ is useful for clarity,
so even though $\alpha$ and $\beta$ introduce (a little) redundancy, the cost is minimal and the practical returns worthwhile.

Furthermore, in the real world validators are generated by code.
This validator code is typically given state data which is carried on the same output as the validator is stored, to help the validator decide whether to validate prospective input.
So when we write $\tf{Validator}\subseteq\powerset(\beta\times\tf{Transaction}_!)$, this effectively makes the validator into a function over \emph{local state data} and a transaction input. 

We see this happen concretely in equation~\eqref{eq.validates} in Definition~\ref{defn.i.points.to.o}: we have $o=(p,d,V)$ and we pass the state data $d$ of $o$ to $V$ the validator of $o$ --- 
$$
(d,\tx\at i)\in V 
$$ 
--- even though this is mathematically redundant, since $d$ and $V$ are both in $o$ so the $d$ could in principle be curried into the $V$.
We retain $\beta$ and set $\tf{Validator}\subseteq\powerset(\beta\times\tf{Transaction}_!)$ rather than just $\tf{Validator}\subseteq\powerset(\tf{Transaction}_!)$, to reflect this practical architecture.

\begin{rmrk}
\label{rmrk.hypothetical}
The discussion above is more than hypothetical and corresponds to the experience of creating executable Haskell code.
The IEUTxO type equations in this paper have been realised as a Haskell implementation; a reference system has been implemented following it; and some of the results of this paper converted into QuickCheck properties.  
The package is on Hackage and GitHub~\cite{gabbay:idealisedeutxo}.\footnote{%
There is no claim to have created a practical blockchain implementation (there is no cryptographic assurance, for example), but the Haskell code comes out as a surprisingly direct translation of the maths in this paper, and demonstrates that this paper has enough executable content to be a plausible reference implementation, for example, within the framework of its own abstractions.

In particular, the design of Figure~\ref{fig.ieutxo} was informed by the experience of creating executable code, and having types $\alpha$ and $\beta$ explicitly available, was useful.}
\end{rmrk}

\subsection{Chunks and blockchains}

\subsubsection{Chunks}

Definition~\ref{defn.chunk} (chunks) is a central concept in this paper, so we will summarise it twice: once now, and again after the definition:
\begin{quote}
A list of transactions is a \emph{chunk} when all input positions are distinct, all output positions are distinct, and all inputs point to at most one earlier validating output.
\end{quote}
Examples are illustrated in Example~\ref{xmpl.example.transactions} (along with examples of blockchains).
In formal detail:
\begin{defn}
\label{defn.chunk} 
Suppose $\mathbb T=(\alpha,\beta,\tf{Transaction},\tf{Validator})$ is an IEUTxO model.
Call a transaction-list $\txs\in[\tf{Transactions}]$ a \deffont{chunk} when:
\begin{enumerate*}
\item\label{blockchain.at.most.one.output}
Distinct outputs appearing in $\txs$ have distinct positions.
\item\label{blockchain.in.out}
Distinct inputs appearing in $\txs$ have distinct positions.  

It may be that the position of an input $i\in\tx\in\txs$ equals that of some output $o\in\tx'\in\txs$, or using the notation of Definition~\ref{defn.i.points.to.o}(\ref{i.points.to.o}): $i\mapsto o$.
If so, by condition~\ref{blockchain.at.most.one.output} there is at most one such. 
We write the \deffont{unique output pointed to by $i$} 
$$
\txs(i)\in\tf{Output}
$$ 
where it exists, so $o=\txs(i)\in\tx'\in\txs$.
\item\label{chunk.earlier}
For every input $i\in\tx\in\txs$, if the output $\txs(i)$ is defined then that output must occur in a transaction $\tx'$ that is strictly earlier than $\tx$ in $\txs$.
\item\label{blockchain.the.point}
For every $i\in\tx\in\txs$, if $\txs(i)$ is defined then $\f{validates}(\txs(i),\tx\at i)$ (Definition~\ref{defn.i.points.to.o}(\ref{i.points.to.o.validates})).
\end{enumerate*}
\end{defn}

\begin{rmrk}
So to sum Definition~\ref{defn.chunk} up in a single line: 
\begin{quote}
a \emph{chunk} is a list of transactions such that a position can only ever be shared between a single pair of an \emph{earlier} output $o$ to a \emph{later} input $i$, which $o$ \emph{validates} --- and otherwise positions are distinct.
\end{quote}
\end{rmrk}

\begin{nttn}
\label{nttn.chunk}
\begin{enumerate}
\item
Write 
$$
\tf{Chunk}_{\mathbb T}\subseteq[\tf{Transaction}_{\mathbb T}]
\quad\text{for the set of chunks over $\mathbb T$.}
$$
We may drop the $\mathbb T$ subscripts; the meaning will always be clear.
\item
We may also call a transaction-list \deffont{valid}, when it is a chunk.
That is: chunks are precisely the \emph{valid} transaction-lists.
\end{enumerate}
\end{nttn}

\begin{rmrk}
\label{rmrk.chunks.and.channels}
The way we have formulated the structure of chunks in Definition~\ref{defn.chunk} 
reminds this author of the $\pi$-calculus, where positions correspond to $\pi$-calculus channel names (and outputs are outputs and inputs are inputs).

When we have this intuition in mind, we may occasionally call positions \deffont{channels}, as in the \emph{blocked channels} of Subsection~\ref{subsect.blocked.channel}.
See also the discussion in Remark~\ref{rmrk.lambda.pi}.
\end{rmrk}

With the intuition of Remark~\ref{rmrk.chunks.and.channels} in mind, we give a simple definition, which refines Definition~\ref{defn.pos}:
\begin{defn}
\label{defn.inputs.and.outputs}
Suppose $\mathbb T=(\alpha,\beta,\tf{Transaction},\tf{Validator})$ is an IEUTxO model and suppose $\tx=(I,O)\in\tf{Transaction}$ is a transaction.
Then define 
\begin{itemize*}
\item
the \deffont{input channels} or \deffont{positions} $\f{input}(\tx)\finsubseteq\atoms$ of $\tx$ to be the finite set of atoms that are positions of inputs in $\tx$, and 
\item
the \deffont{output channels} or \deffont{positions} $\f{output}(\tx)\finsubseteq\atoms$ of $\tx$ to be the finite set of atoms that are positions of outputs in $\tx$.
\end{itemize*}
In symbols:
$$
\begin{array}{r@{\ }l}
\f{input}(\tx) =& \{p \mid (p,k)\in I\}
\\
\f{output}(\tx) =& \{p \mid (p,d,V)\in O\} .
\end{array}
$$ 
\end{defn}

An important special case of Notation~\ref{nttn.chunk} is when the chunk is a singleton list, i.e. it contains just one transaction:
\begin{lemm}
\label{lemm.singleton.chunk.valid}
Suppose $\mathbb T$ is an IEUTxO model and suppose $\tx\in\tf{Transaction}$.
Then 
$$
[\tx]\in\tf{Chunk}
\quad\text{if and only if}\quad 
\f{input}(\tx)\cap\f{output}(\tx)=\varnothing.
$$
\end{lemm}
\begin{proof}
We consider the conditions in Definition~\ref{defn.chunk} and see that condition~\ref{chunk.earlier} forces the input and output channels of the transaction to be disjoint, and then none of the other conditions are relevant.
\end{proof}

\begin{rmrk}
Definition~\ref{defn.inputs.and.outputs} refines Definition~\ref{defn.pos}, and another way to phrase Lemma~\ref{lemm.singleton.chunk.valid} is that $[\tx]$ is a chunk precisely when $\pos(\tx)=\f{inputs}(\tx)\uplus\f{outputs}(\tx)$, where $\uplus$ denotes disjoint sets union.
\end{rmrk}


\begin{figure}
    \centering 
    \newcommand\ma{10}
    \newcommand\mb{5}
    \newcommand\mc{3}
    \begin{tikzpicture}[scale=0.3]

        \mkBackboneBC{0}

        \drawLine{bc2}{}{}{x2}
        \drawLineToCircle{$(b) + (1,-0.25)$}{below, xshift=-1, yshift=-1}{\nameObb}{bc2}{blue}

        \drawLineFromCircle{ab1}{}{}{x1}{blue}
        \drawLine{x1}{above, xshift=1, yshift=1}{\nameIca}{$(c) + (-1, 0.25)$}

        \drawLineFromCircle{cd1}{above, xshift=1}{\nameIda}{$(d) + (-1, 0.5)$}{blue}
        \drawLineFromCircle{cd2}{above, xshift=-2, yshift=-1}{\nameIdb}{$(d) + (-1, 0  )$}{blue}
        \drawLine{x2}{below}{\nameIdc}{$(d) + (-1,-0.5)$}

    \end{tikzpicture}
    \caption{A blockchain $\mathcal B=[\tx_1,\tx_2,\tx_3,\tx_4]$}
    \label{fig.blockchain}
    \ \\[5mm]
    \begin{tikzpicture}[scale=0.3]
        \mkBackboneBC{0.05}

        \drawLineToCircle{$(b) + (1,-0.25)$}{below, xshift=-1, yshift=-1}{\nameObb}{bc2}{red}
        \drawLineFromCircle{x1}{above, xshift=1, yshift=1}{\nameIca}{$(c) + (-1, 0)$}{red}
        \drawLineFromCircle{cd1}{above, xshift=1, yshift=0}{\nameIda}{$(d) + (-1, 0.5)$}{blue}
        \drawLineFromCircle{cd2}{above,xshift=-1, yshift=-2}{\nameIdb}{$(d) + (-1, 0  )$}{blue}
        \drawLineFromCircle{cd4}{below}{\nameIdc}{$(d) + (-1,-0.5)$}{red}

        \drawCircle{red}{ab1}

    \end{tikzpicture}
    \caption{$\mathcal B$ chopped up as a blockchain $[\tx_1,\tx_2]$ and a chunk $[\tx_3,\tx_4]$}
    \label{fig.chunks1}
    \ \\[5mm]
    \begin{tikzpicture}[scale=0.3]
        \mkBackboneCB{0.05}

        \drawLine{bd2}{below, xshift=1}{\nameIdc}{$(d) + (-1,-0.5)$}
        \drawLineToCircle{$(b) + (1,-0.25)$}{below, xshift=-1, yshift=-1}{\nameObb}{bd2}{blue}
        \drawLineFromCircle{ab1}{above, xshift=1, yshift=1}{\nameIca}{$(c) + (-1, 0.5)$}{blue}

        \drawLineFromCircle{bd1}{above, xshift=1}{\nameIda}{$(d) + (-1, 0.5)$}{red}
        \drawLineFromCircle{bd3}{below, yshift=0}{\nameIdb}{$(d) + (-1, 0  )$}{red}

        \drawCircle{red}{cb1}
        \drawCircle{red}{cd1}
        \drawCircle{red}{cd2}
    \end{tikzpicture}
    \caption{$\mathcal B$ chopped up as a blockchain $[\tx_1,\tx_3]$ and a chunk $[\tx_2,\tx_4]$}
    \label{fig.chunks2}
    \ \\[5mm]
    \begin{tikzpicture}[scale=0.3]
        \mkBackboneCB{0}

        \drawLineFromCircle{ab1}{above, xshift=1, yshift=1}{\nameIca}{$(c) + (-1, 0.5)$}{blue}
        \drawLine{bd1}{above, xshift=1}{\nameIda}{$(d) + (-1, 0.5)$}
        \drawLine{bd3}{below, yshift=0}{\nameIdb}{$(d) + (-1, 0  )$}
        \drawLine{bd2}{below, xshift=1}{\nameIdc}{$(d) + (-1,-0.5)$}
        \drawLineToCircle{$(b) + (1,-0.25)$}{below, xshift=-1, yshift=-1}{\nameObb}{bd2}{blue}

        \drawLineFromCircle{ab3}{}{}{cb1}{blue}
        \drawLineFromCircle{cd1}{}{}{bd1}{blue}
        \drawLineFromCircle{cd2}{}{}{bd3}{blue}

    \end{tikzpicture}
    \caption{The blockchain $\mathcal B'=[\tx_1,\tx_3,\tx_2,\tx_4]$}
    \label{fig.blockchain'}
\end{figure}


\subsubsection{UTxOs, UTxIs \dots}

\begin{defn}
\label{defn.utxo.utxi}
Suppose $\mathbb T$ is an IEUTxO model and $\txs\in[\tf{Transaction}]$.
\begin{enumerate}
\item
If $i\in\tx\in\txs$ and $\txs(i)$ is not defined\footnote{\dots so $i$ does not point to an earlier validating output in $\txs$\dots} then call the unique atom $a\in\pos(i)\subseteq\atoms$ an \deffont{unspent transaction input}, or \deffont{UTxI}.
Write 
$$
\utxi(\txs)\finsubseteq\atoms\quad\text{for the UTxIs of $\txs$.}
$$
\item
If $o\in\tx\in\txs$ and $o\neq\txs(i)$ for all later $i\in\tx\in\txs$\,\footnote{\dots so $o$ does not validate some later input in $\txs$ \dots} then call the unique atom $a\in\pos(o)\subseteq\atoms$ an \deffont{unspent transaction output}, or \deffont{UTxO}.
Write 
$$
\utxo(\txs)\finsubseteq\atoms\quad\text{for the set of UTxOs of $\txs$.}
$$
\item
If $o\in\tx\in\txs$ and $o=\txs(i)$ for some later $i\in\tx\in\txs$, then call the atom $a\in\pos(o)\subseteq\atoms$ a \deffont{spent transaction channel}, or \deffont{STx}.
Write 
$$
\stx(\txs)\finsubseteq\atoms\quad\text{for the set of STxs of $\txs$.}
$$
\end{enumerate}
\end{defn}

\begin{rmrk}
Some comments on the interpretation of $\utxi$ and $\utxo$ and $\stx$ from Definition~\ref{defn.utxo.utxi}:

$\utxi(\txs)$, $\utxo(\txs)$, and $\stx(\txs)$ are all finite sets of atoms, but we interpret them somewhat differently:
\begin{enumerate}
\item
Intuitively, an atom $a\in\utxo(\txs)$ identifies an output $o\in\tx\in\txs$ with position $a$.
So $\utxo(\txs)$ is a finite set of names of outputs in $\txs$.
\item
Intuitively, an atom $a\in\utxi(\txs)$ identifies an input-in-context $\tx\at i$, for $i\in\tx\in\txs$ with $\pos(i)=a$.

We say this because the validator of an output takes as argument an input-in-context $\tx\at i\in\tf{Transaction}_!$.
So $\utxi(\txs)$ is a finite set of names for inputs-in-contexts.
\item
Intuitively, an atom $a\in\stx(\txs)$ identifies a pair of an output and the input-in-context that spends it.
Thus $a$ could be thought of as this pair, or $a$ could be thought of as an edge in a graph that joins a node representing the output, to a node representing the input.  

So $\stx(\txs)$ is a finite set of internal names of already-spent communications between outputs and inputs-in-context within $\txs$.
\end{enumerate}
\end{rmrk}

\subsubsection{\dots and blockchains}

With the machinery we have now have, it is quick and easy to define blockchains:
\begin{defn}
\label{defn.blockchain}
A \deffont{blockchain} is a chunk $\f{ch}\in\tf{Chunk}$ such that $\utxi(\f{ch})=\varnothing$. 
In words: 
\begin{quote}
a blockchain is a chunk with no unspent inputs. 
\end{quote}
Diagrammatic examples follow in Example~\ref{xmpl.example.transactions}:
\end{defn}

\begin{xmpl}
\label{xmpl.example.transactions}
Recall we observed in Subsection~\ref{subsect.what.its.about} that (in the terminology that we now have) the key operation of an IEUTxO is to attach a transaction's inputs to a chunk's outputs.
 
Example transaction-lists, blockchains, and chunks are illustrated in Figures~\ref{fig.transaction}, \ref{fig.transaction.combined}, \ref{fig.blockchain}, \ref{fig.chunks1}, \ref{fig.chunks2}, and~\ref{fig.blockchain'}.\footnote{These diagrams are adapted from~\cite{gabbay:utxabs}, with my coauthor's agreement.}

In the Figures, a blue circle denotes a validator on an output at some position ($a,b,c,\dots$) that has accepted an input and connected to it; and a red circle denotes an unspent input or output, meaning one that has not connected up with a validator to form a spent output-input pair:
\begin{enumerate}
\item
$\mathcal B$, $\mathcal B'$, $[\tx_1,\tx_2]$ and $[\tx_1,\tx_3]$ are blockchains, because they have unspent outputs (in red) but no unspent inputs.

In these blockchains, $\tx_1$ is what is called the \deffont{genesis block}, meaning the first block in the chain.
It follows from the definitions that the genesis block has no inputs.\footnote{
If we permitted loops, i.e. connections from a later output to an input that is on the same block or earlier, then a genesis block might have an input that addresses an output on the same or a later block.
See Subsection~\ref{subsect.future.work}(\ref{loops}).
But for the definitions as set up in this paper, that is not allowed.

Note that our definitions admit a blockchain with two genesis blocks; it suffices to have more than one transaction with no inputs.  Whether this is a feature or a bug depends on the application. 
}
\item
$[\tx]$ (Figure~\ref{fig.transaction}) and $[\tx_1]$, $[\tx_3,\tx_4]$ and $[\tx_2,\tx_4]$ are chunks, but not blockchains because they have unspent inputs (in red).
\item
$[\tx_2,\tx_1]$ is neither a blockchain nor chunk, because the \nameIba-input of $\tx_2$ points to the later \nameOab-output of $\tx_1$. 
It is just a list of transactions.
\end{enumerate}
\end{xmpl}

We note two alternative characterisation of blockchains (Definition~\ref{defn.blockchain}):
\begin{lemm}
A chunk is a blockchain when \dots
\begin{enumerate*}
\item
\dots the `at most one' in Definition~\ref{defn.chunk}(\ref{blockchain.in.out}) is strengthened to `precisely one'.
\item
\dots the function $i\mapsto\txs(i)$ (Definition~\ref{defn.chunk}(\ref{blockchain.in.out}))
is a total function on the inputs in $\txs$ (so that every input points to precisely one output in an earlier transaction).
\end{enumerate*} 
\end{lemm}

\begin{rmrk}
\label{rmrk.why.so.slow}
We step back to reflect on Definition~\ref{defn.blockchain}.
This is supposed to be a paper about blockchains; why did it take us this long to get to them?
Because they are a special case of something better and more pertinent: chunks.

A blockchain is just a left-closed chunk.
There is nothing wrong with blockchains, but mathematically, chunks seem more interesting:
\begin{enumerate*}
\item
A sublist of a blockchain is a chunk, not a blockchain (we prove this in a moment, in Corollary~\ref{corr.sublist.inclusion.chunks}).
\item
A composition of blockchains is possible, but uninteresting; whereas composition of chunks is clearly an interesting operation.\footnote{Blockchains have no unspent inputs, so if composed they just sit side-by-side and do not communicate.  Contrast this with chunks, for which the partial monoidal composition is clearly natural.}
\item
If we cut a blockchain into $n$ pieces then we get one \emph{blockchain} (the initial segment) \dots and $n-1$ \emph{chunks}. 
\item
Chunks can in any case be viewed as a natural generalisation of blockchains, to allow UTxIs as well as UTxOs. 
\end{enumerate*}
Definitions and results like Definition~\ref{defn.popm}, Theorem~\ref{thrm.popm}, and Lemma~\ref{lemm.fresh.chunks} inhabit a universe of chunks, not blockchains.

Even in implementation, where we care about real blockchains on real systems, a lot of development work goes into allowing users in practice to download only partial histories of the blockchain rather than having to download and store a complete record --- the motivation here is practical, not mathematical --- and in the terminology of this paper, we would say that for efficiency we may prefer to work with chunks where possible, because they can be partial and so can be more lightweight.

So the focus of this paper is on chunks: they generalise blockchains, have better mathematical structure; and chunks are in any case where we arrive even if we \emph{start off} asserting (de)compositional properties of blockchains; and finally --- though this is not rigorously explored in this paper, but we would suggest that --- chunks are also where we arrive when we consider space-efficient blockchain implementations.

Finally, we mention that blockchains have a right monoid action given by concatenating chunks.
Thus, by analogy here with rings and modules, we could imagine for future work a mathematics of blockchains generalising Definition~\ref{defn.blockchain} such that a `blockchain set' is just any set with a suitable chunk action.
\end{rmrk}

\subsection{Properties of chunks and blockchains}

\subsubsection{Algebraic and closure properties of chunks}

Lemma~\ref{lemm.invalidity.must.be.somewhere} expresses that a list of transactions is a valid chunk if and only if every sublist of it of length at most two, is a valid chunk.
In this sense, (in)validity is a \emph{local} phenomenon:
\begin{lemm}
\label{lemm.invalidity.must.be.somewhere}
Suppose $\mathbb T=(\alpha_{\mathbb T},\beta_{\mathbb T},\tf{Transaction}_{\mathbb T},\tf{Validator}_{\mathbb T})$ is an IEUTxO model, and suppose $[\tx_1,\dots,\tx_n]\in[\tf{Transaction}]$.
Then the following conditions are equivalent: 
\begin{itemize*}
\item
$[\tx_i,\tx_j]\in\tf{Chunk}$ for every $1\leq i< j\leq n$\,\footnote{This condition holds trivially if $n=0$, i.e. for the empty list.} 
\item
$[\tx_1,\dots,\tx_n]\in\tf{Chunk}$
\end{itemize*}
\end{lemm}
\begin{proof}
We note of the well-formedness conditions on chunks from Definition~\ref{defn.chunk} that they all concern relationships involving at most two transactions. 
\end{proof}

\begin{corr}[Validity is down-closed]
\label{corr.sublist.inclusion.chunks}
Suppose we have an IEUTxO model $\mathbb T=(\alpha,\beta,\tf{Transaction},\tf{Validator})$ and $l,l'\in[\tf{Transaction}]$.

Recall from Notation~\ref{nttn.pointed} that we order lists by sublist inclusion, so $l'\leq l$ when $l'$ is a sublist of $l$.
Then 
$$
\f{ch}\in\tf{Chunk}\ \land\ l'\in[\tf{Transaction}]\ \land\ l'\leq \f{ch} 
\quad\text{implies}\quad
l'\in\tf{Chunk} . 
$$
In words: every sublist of a chunk, is itself a chunk.
\end{corr}
\begin{proof}
From Lemma~\ref{lemm.invalidity.must.be.somewhere}.
\end{proof}

We can wrap up Corollary~\ref{corr.sublist.inclusion.chunks} in a nice mathematical package:
\begin{defn}
\label{defn.popm}
Suppose $(X,\leq,\mact,\mbot)$ has the following structure:
\begin{enumerate*}
\item
$X$ is a set.
\item
$(X,\leq)$ is a partial order, for which $\mbot$ is a bottom element.
\item
$\mact$ is a \deffont{partial monoid} action on $X$, meaning that $(x\mact y)\mact z$ exists if and only if $x\mact (y\mact z)$ exists, and if both exist then they are equal.\footnote{A slightly weaker possibility is that $(x\mact y)\mact z$ and $x\mact (y\mact z)$ need not exist or not exist together, but \emph{if} they both exist, then they are equal.  This is the natural notion \emph{if} we bind names of spent output-input pairs.  More discussion in Subsection~\ref{subsect.garbage-collection}.}
\item
$\mact$ is \deffont{down-closed}, meaning that if $x'\leq x$ and $x\mact y$ exists, then so does $x'\mact y$, and similarly for $y\mact x$ and $y\mact x'$. 
\item
$\mact$ is \deffont{monotone} where defined, meaning that if $x'\leq x$ then $x'\mact y\leq x\mact y$ (provided $x\mact y$ exists), and similarly for $y\mact x$ and $y\mact x'$.
\end{enumerate*}
In this case, call $(X,\leq,\mact,\mbot)$ a \deffont{partially-ordered partial monoid}.
\end{defn}

\begin{thrm}
\label{thrm.popm}
Suppose $\mathbb T$ is an IEUTxO model (Definition~\ref{defn.ieutxo.model}).

Then its set of chunks $\tf{Chunk}_{\mathbb T}$ (Definition~\ref{defn.chunk}) forms a partially-ordered partial monoid (Definition~\ref{defn.popm}), 
where 
\begin{itemize*}
\item
$\leq$ is sublist inclusion, 
\item
$\mact$ is list concatenation, and 
\item
the unit element is $[]$ the empty set.
\end{itemize*}
\end{thrm}
\begin{proof}
By facts of lists, and Corollary~\ref{corr.sublist.inclusion.chunks}.
\end{proof}

\subsubsection{Some observations on observational equivalence}
\label{subsect.obs.obs}

\begin{rmrk}
Lemmas~\ref{lemm.fresh.chunks.defined} and~\ref{lemm.fresh.chunks} apply to IEUTxO models and essentially give criteria for observational equivalence when positions are disjoint.

We find them echoed in the theory of abstract chunk systems as Definitions~\ref{defn.oriented}(\ref{oriented.fresh.defined}) and~\ref{defn.oriented}(\ref{oriented.fresh.commute}), and we need them for Proposition~\ref{prop.partial.converse.fresh}.
\end{rmrk}

\begin{lemm}
\label{lemm.fresh.chunks.defined}
Suppose $\mathbb T=(\alpha,\beta,\tf{Transaction},\tf{Validator})$ is an IEUTxO model and $\f{ch},\f{ch}'\in\tf{Chunk}_{\mathbb T}$.
Then 
$$
\pos(\f{ch})\cap\pos(\f{ch}')=\varnothing
\quad\text{implies}\quad
\f{ch}\mact\f{ch}'\in\tf{Chunk}_{\mathbb T}.
$$
\end{lemm}
\begin{proof}
By routine checking of possibilities, using the fact that if $\pos(\f{ch})\cap\pos(\f{ch}')=\varnothing$ (Definition~\ref{defn.pos}) then they have no positions in common, so no output in one can be called on to validate an input in the other.
\end{proof}

\begin{defn}
\label{defn.ch.obs.equiv}
Suppose $\f{ch},\f{ch}'\in\tf{Chunk}_{\mathbb T}$.
Then call $\f{ch}$ and $\f{ch}'$ \deffont{commuting} when
$$
\begin{array}{l}
\f{ch}\mact\f{ch}'\in\tf{Chunk}_{\mathbb T}
\liff
\f{ch}'\mact\f{ch}\in\tf{Chunk}_{\mathbb T} .
\end{array}
$$
\end{defn}

\begin{rmrk}
Definition~\ref{defn.ch.obs.equiv} is clearly a notion of observational equivalence between $\f{ch}\mact\f{ch}'$ and $\f{ch}'\mact\f{ch}$ where the observable is `forms a valid chunk with'.
This observable does not depend on internal structure, so we will develop it further once we have abstract chunk systems; see Definition~\ref{defn.circ}.

For now, Definition~\ref{defn.ch.obs.equiv} gives us just enough of the background theory of observational equivalence, to state and prove Lemma~\ref{lemm.fresh.chunks}, Proposition~\ref{prop.fresh.chunks.iff}, and Theorem~\ref{thrm.practical}.
\end{rmrk}

\begin{lemm}
\label{lemm.fresh.chunks}
Suppose $\mathbb T=(\alpha,\beta,\tf{Transaction},\tf{Validator})$ is an IEUTxO model.
Then:
\begin{enumerate}
\item
If $\tx,\tx'\in\tf{Transaction}$ and $\pos(\tx)\cap\pos(\tx')=\varnothing$ (Definition~\ref{defn.pos}) then the following all hold:
$$
[\tx,\tx']\in\tf{Chunk}
\ \liff\ 
[\tx',\tx]\in\tf{Chunk}
\ \liff\ 
[\tx],[\tx']\in\tf{Chunk} 
$$
\item\label{fresh.chunks.obs}
As a corollary, if $\f{ch},\f{ch}'\in\tf{Chunk}$ and $\pos(\f{ch})\cap\pos(\f{ch}')=\varnothing$ then $\f{ch}$ and $\f{ch}'$ are commuting (Definition~\ref{defn.ch.obs.equiv}).
\end{enumerate}
\end{lemm}
\begin{proof}
\begin{enumerate}
\item
By routine checking of possibilities, using the fact that if $\pos(\tx)\cap\pos(\tx')=\varnothing$ 
then they have no positions in common, so no output in one can be called upon to validate an input in the other.
\item
It is a fact that if $\tx\in l$ then $\pos(\tx)\subseteq\pos(l)$ and similarly for $\tx'\in l'$.
The corollary now follows by a routine argument from part~1 of this result and Lemma~\ref{lemm.invalidity.must.be.somewhere}.
\end{enumerate}
\end{proof}

\subsubsection{Properties of UTxOs and UTxIs}

We return to Definition~\ref{defn.utxo.utxi}:
Lemma~\ref{lemm.utxi.utxo.empty} uses Definition~\ref{defn.utxo.utxi} to note some simple properties of Definition~\ref{defn.chunk}.
\begin{lemm}
\label{lemm.utxi.utxo.empty}
Suppose $\mathbb T$ is an IEUTxO model and $\f{ch},\f{ch}'\in\tf{Chunk}_{\mathbb T}$
Then:
\begin{enumerate}
\item\label{utxi.cap.utxo} 
$\f{utxi}(\f{ch})\cap\f{utxo}(\f{ch})=\varnothing$
\item\label{utxi.utxo} 
If $\f{ch}\mact\f{ch}'\in\tf{Chunk}_{\mathbb T}$
then 
$\pos(\f{ch})\cap\pos(\f{ch}')\subseteq\utxo(\f{ch})\cap\utxi(\f{ch}')$.
\item\label{utxo.cup.utxo}
$\begin{array}[t]{r@{\ }l}
\varnothing=&\f{utxi}(\f{ch})\cap\f{stx}(\f{ch})
\\
\varnothing=&\f{utxo}(\f{ch})\cap\f{stx}(\f{ch})
\\
\pos(\f{ch})=&\f{utxi}(\f{ch})\uplus\f{utxo}(\f{ch})\uplus\f{stx}(\f{ch})
\qquad\text{\it $\uplus$ is disjoint union}
\end{array}
$
\end{enumerate}
\end{lemm}
\begin{proof}
\begin{enumerate}
\item
An input cannot point to a later output, because of Definition~\ref{defn.chunk}(\ref{chunk.earlier}), and if it points to an earlier output then by construction in Definition~\ref{defn.utxo.utxi} this position no longer labels a UTxO or UTxI.
Furthermore a position can be used at most once in an input-output pair, by Definition~\ref{defn.chunk}(\ref{blockchain.in.out}).
\item
From Definition~\ref{defn.chunk}(\ref{chunk.earlier}), as for the previous case.
\item
All facts of Definition~\ref{defn.utxo.utxi} and Figure~\ref{fig.positions}. 
\end{enumerate}
\end{proof}

\begin{rmrk}
Proposition~\ref{prop.fresh.chunks.iff} can be viewed as a stronger version of Lemma~\ref{lemm.fresh.chunks.defined}.
It is an important result because it relates the following apparently different observables:
\begin{enumerate}
\item
A statically observable property, that $\f{ch}$ and $\f{ch}'$ mention disjoint sets of positions.
\item
A locally observable property, that $\f{ch}$ and $\f{ch}'$ compose in both directions.
\item
An abstract global observable, that $\f{ch}\mact\f{ch}'$ and $\f{ch}'\mact\f{ch}$ can be commuted in any larger chunk.
\end{enumerate}
Compare also with Proposition~\ref{prop.partial.converse.fresh}, which is a similar result but for the differently-constructed abstract chunk systems.
\end{rmrk}
 
\begin{prop}
\label{prop.fresh.chunks.iff}
Suppose $\mathbb T=(\alpha,\beta,\tf{Transaction},\tf{Validator})$ is an IEUTxO model and $\f{ch},\f{ch}'\in\tf{Chunk}_{\mathbb T}$.
Then the following are equivalent:
\begin{enumerate}
\item 
$\pos(\f{ch})\cap\pos(\f{ch}')=\varnothing$
\item
$\f{ch}\mact\f{ch}'\in\tf{Chunk}_{\mathbb T} \land \f{ch}'\mact\f{ch}\in\tf{Chunk}_{\mathbb T}$.
\item
$\f{ch}$ and $\f{ch}'$ are commuting.
\end{enumerate}
\end{prop}
\begin{proof}
The top-to-bottom implication is Lemma~\ref{lemm.fresh.chunks.defined}.

For the bottom-to-top implication, suppose that $\f{ch}\mact\f{ch}',\f{ch}'\mact\f{ch}\in\tf{Chunk}_{\mathbb T}$.
From Lemma~\ref{lemm.utxi.utxo.empty}(\ref{utxi.utxo}) we have 
$$
\pos(\f{ch})\cap\pos(\f{ch}')\subseteq (\utxo(\f{ch})\cap\utxi(\f{ch}'))\cap(\utxo(\f{ch}')\cap\utxi(\f{ch})) .
$$
We can rearrange this: 
$$
\pos(\f{ch})\cap\pos(\f{ch}')
\subseteq 
(\utxi(\f{ch})\cap\utxo(\f{ch}))\cap(\utxi(\f{ch}')\cap\utxo(\f{ch}')) .
$$
Now we use Lemma~\ref{lemm.utxi.utxo.empty}(\ref{utxi.cap.utxo}).

The final part is direct from Lemma~\ref{lemm.fresh.chunks}(\ref{fresh.chunks.obs}).
\end{proof}

\begin{rmrk}
\label{rmrk.separation.logic}
Proposition~\ref{prop.fresh.chunks.iff} is a \emph{resource separation} result: if two chunks depend on disjoint resources (disjoint sets of positions) then they commute.
This is in the spirit of separation logic~\cite{reynolds:seplls}, which is a family of logics for reasoning about programs with resource separation in programs --- intuitively, that if two programs depend on disjoint resources (e.g. channels or pointers), then they should not interfere with one another, just as we see in Proposition~\ref{prop.fresh.chunks.iff}.

It is also in the spirit of a short paper~\cite{nester:fouls} (which appeared after this paper went into initial review) which uses monoidal categories to reason on resources in a broadly similar spirit, albeit using different methods.
A quote from that paper makes a related point: \emph{We have seen how the resource theoretic interpretation of monoidal categories, and in particular their string diagrams, captures the sort of material history that concerns ledger structures for blockchain systems.}
More on this in the Conclusions.
\end{rmrk}

We conclude with Lemma~\ref{lemm.perm.pos}, which we will need later in Lemma~\ref{lemm.perm.utxi}:
\begin{lemm}
\label{lemm.perm.pos}
Suppose $\mathbb T\in\tf{IEUTxO}$ is an IEUTxO model and $\pi\in\f{Perm}$ is a permutation of atoms and $\txs\in[\tf{Transaction}_{\mathbb T}]$.
Then:
$$
\begin{array}[t]{r@{\ }l}
f(\pi\pact\txs) =&\pi\pact f(\txs) 
\\
=& \{\pi(a)\mid a\in f(\txs)\}
\end{array}
\quad\text{for}\  f\in\{\utxi,\utxo,\stx,\pos\} 
$$ 
In the terminology of Definition~\ref{defn.equivariant}:
$\utxi$, $\utxo$, $\stx$, and $\pos$ are all equivariant.
\end{lemm}
\begin{proof}
Direct from Figure~\ref{fig.positions} and Definitions~\ref{defn.utxo.utxi} and~\ref{defn.perm}.
\end{proof}

\subsection{An application: UTxO systems are `Church-Rosser', in a suitable sense}

We now come to Theorem~\ref{thrm.practical} which is an application of our machinery so far:

\begin{thrm}[Church-Rosser for UTxO]
\label{thrm.practical}
Suppose $\mathbb T\in\tf{IEUTxO}$ and $y,x,x'\in\tf{Chunk}_{\mathbb T}$.
Suppose further that
$$
y\mact x\mact x'\in\tf{Chunk}_{\mathbb T}
\quad\text{and}\quad
\utxi(y\mact x') = \utxi(y\mact x\mact x') .
$$ 
Then we have that:
\begin{enumerate*}
\item
$x$ and $x'$ are commuting (Definition~\ref{defn.ch.obs.equiv}).
\item
$y\mact x'\mact x\in\tf{Chunk}_{\mathbb T}$. 
\end{enumerate*}
\end{thrm}
\begin{proof}
We know by Corollary~\ref{corr.sublist.inclusion.chunks} (because $y\mact x\mact x'$ is a chunk) that $x\mact x'$ is a chunk, so $\pos(x)\cap\pos(x')\subseteq\utxo(x)\cap\utxi(x')$.

We also know that $\utxi(y\mact x')=\utxi(y\mact x\mact x')$ and it follows from Definition~\ref{defn.utxo.utxi} that ($\utxo(y)\cap\utxi(x)=\varnothing$ and) $\utxo(x)\cap\utxi(x')=\varnothing$.\footnote{We really use here the fact that names can only be used to link an output to an input \emph{once}.}

Therefore $\pos(x)\cap\pos(x')=\varnothing$.
By Proposition~\ref{prop.fresh.chunks.iff}, $x$ and $x'$ are commuting, and it follows (since $y\mact x'\mact x\in\tf{Chunk}_{\mathbb T}$) that $y\mact x\mact x'\in\tf{Chunk}_{\mathbb T}$.
\end{proof}

Recall the definition of a blockchain (Definition~\ref{defn.blockchain}) as being a chunk with empty $\utxi$.
Then we can specialise Theorem~\ref{thrm.practical} as follows:
\begin{corr}
\label{corr.pure}
Suppose $\mathbb T\in\tf{IEUTxO}$ and $y,x,x'\in\tf{Chunk}_{\mathbb T}$ and suppose $y\mact x'$ is a blockchain.
Then:
\begin{enumerate*}
\item
If $y\mact x\mact x'$ is a blockchain then $x$ and $x'$ commute.
\item
If $x$ and $x'$ do not commute then $y\mact x\mact x'$ is not a blockchain.
\end{enumerate*}
\end{corr}
\begin{proof}
Direct from Theorem~\ref{thrm.practical}, for the case that $\utxi(y\mact x')=\utxi(y\mact x\mact x')=\varnothing$.
\end{proof}

\begin{rmrk}
Corollary~\ref{corr.pure} models a situation where someone designs a chunk $x'$ that will successfully attach to a blockchain $y$, but because this is a distributed system, somebody else gets in and attaches $x$ first.

Then this can only happen if $x$ and $x'$ are commuting; conversely, if $x$ and $x'$ are not commuting then one of $y\mact x\mact x'$ or $y\mact x'$ must fail to be a blockchain. 

What makes this interesting is that it is an important correctness property from the point of view of the user who created $x'$:
\emph{if} $x'$ is accepted onto both $y$ and $y\mact x$ without failure, then it \emph{does not matter} that $x$ got in first --- $y\mact x\mact x'$ and $y\mact x'\mact x$ are equivalent up to observable behaviour.

Note that:
\begin{itemize*}
\item
Theorem~\ref{thrm.practical} and Corollary~\ref{corr.pure} do not state that $\tf{IEUTxO}$ systems are insensitive to order, and 
\item
Theorem~\ref{thrm.practical} and Corollary~\ref{corr.pure} do not state that all transactions always commute (this is simply not what is written on the page). 
\end{itemize*}

Theorem~\ref{thrm.practical} and Corollary~\ref{corr.pure} can be viewed as a \emph{purity} property, in the sense of functional programming: if $x'$ successfully combines with $y$, then inputs to $x'$ may not be modified by an intervening environment $x$ --- there are no side-effects!

More specifically, these results can be viewed as playing a role analogous to a `Church-Rosser' or `confluence' property.
To see why, contrast with the situation in an accounts-based system --- corresponding to an imperative paradigm --- where a transaction may be successfully appended \emph{even if} parameters to it on the blockchain get modified by intervening transactions.  

To take a concrete scenario: I could check my bank account, observe I have enough money for a purchase, submit my transaction --- and then go into overdraft and be subject to overdraft fees, because a direct debit happened to arrive in-between (a) my checking my balance and designing my purchase transaction and (b) the payment request for the purchase transaction arriving at my account.
This error clearly comes from the use of a stateful, imperative programming style, and Theorem~\ref{thrm.practical} expresses a rigorous sense in which a corresponding phenomenon is impossible in a UTxO-style system.\footnote{In practice, programmers of smart contracts on accounts-based blockchains may write explicit tests into their smart contracts that double-check values of input variables \emph{at time of attachment to the blockchain}, if they anticipate this might be an issue.  More discussion of this is in~\cite{gabbay:utxabs}.} 

With this comparison in mind, we see that Theorem~\ref{thrm.practical} is a purity result: state is local, and composition of chunks succeeds or fails locally. 
\end{rmrk}

\subsection{The category \texorpdfstring{$\tf{IEUTxO}$}{IEUTxO} of IEUTxO models}

We can organise our IEUTxO models into a category:
\begin{defn}
\label{defn.ieutxo.category}
Let $\tf{IEUTxO}$ be a category such that: 
\begin{enumerate}
\item
Objects $\mathbb S,\mathbb T$ are IEUTxO models (Definition~\ref{defn.solution}).
\item\label{ieutxo.arrow}\label{condition.a}
An arrow $\lf:\mathbb S\to\mathbb T$ is a map 
$$
\lf:\tf{Transaction}_{\mathbb S}\to\tf{Chunk}_{\mathbb T}
$$ 
such that
if $\tx,\tx'\in\tf{Transaction}_{\mathbb S}$ then
\begin{equation}
\label{eq.condition.a}
[\tx,\tx']\in\tf{Chunk}_{\mathbb S}
\quad\text{implies}\quad
\lf(\tx)\mact \lf(\tx')\in\tf{Chunk}_{\mathbb T}.
\end{equation}

Above, $\mact$ denotes monoid composition, which on chunks is list concatenation; see Notation~\ref{nttn.pointed}(\ref{list.concat}).
\item
The identity arrow maps $\tx$ to $[\tx]$.
\item\label{ieutxo.arrow.composition}
Composition of arrows is pointwise, meaning that 
if 
$$
\begin{array}{r@{\ }l@{\qquad}r@{\ }l}
\lf:&\mathbb S\to\mathbb S'
&
\tx\in&\tf{Transaction}_{\mathbb S} 
\\
\lf':&\mathbb S'\to\mathbb S''
&
\lf(\tx)=&[\tx'_1,\dots,\tx'_n]
\end{array}
$$ 
then $\lf'\lf :\mathbb S\to\mathbb S''$ is such that
$$
\tx\in\tf{Transaction}_{\mathbb S} \longmapsto \lf'(\tx'_1)\mact\ldots\mact \lf'(\tx'_n)\in\tf{Chunk}_{\mathbb S''} .
$$ 
We prove this mapping is indeed an arrow --- thus, it maps to chunks --- in Corollary~\ref{corr.composition.arrows.defined.ieutxo}.
\end{enumerate}
\end{defn}

\begin{lemm}
\label{lemm.validity.preserved}
An arrow $\lf:\mathbb S\to\mathbb T$ (Definition~\ref{defn.ieutxo.category}(\ref{ieutxo.arrow}))
induces a mapping of chunks $\tf{Chunk}_{\mathbb S}\to\tf{Chunk}_{\mathbb T}$, by acting on the individual transactions and composing the results:
$$
\lf([\tx_1,\dots,\tx_n]) = \lf(\tx_1)\mact \ldots\mact \lf(\tx_n) .
$$
\end{lemm}
\begin{proof}
The nontrivial part is to check that 
\begin{itemize*}
\item
if $\f{ch}=[\tx_1,\dots,\tx_n]$ is a valid chunk in $\tf{Chunk}_{\mathbb S}$, 
\item
then $\lf(\tx_1)\mact\ldots\mact \lf(\tx_n)$ is a valid chunk in $\tf{Chunk}_{\mathbb T}$.
\end{itemize*}
This follows by combining condition~\ref{ieutxo.arrow} of Definition~\ref{defn.ieutxo.category} with Lemma~\ref{lemm.invalidity.must.be.somewhere} and Corollary~\ref{corr.sublist.inclusion.chunks}.
\end{proof}

\begin{corr}
\label{corr.ab}
Condition~\ref{ieutxo.arrow} of Definition~\ref{defn.ieutxo.category} is equivalent to either of the following conditions: 
\begin{enumerate*}
\item
$[\tx_1,\dots,\tx_n]\in\tf{Chunk}_{\mathbb S}$
implies
$\lf(\tx_1)\mact\ldots\mact\lf(\tx_n)\in\tf{Chunk}_{\mathbb T}$.
\item
$f$ induces a monoid homomorphism on chunks.
\end{enumerate*}
\end{corr}
\begin{proof}
The equivalence of conditions~1 and~2 above is routine, given that every chunk factors into singletons.
Then condition~\eqref{eq.condition.a} in part~\ref{ieutxo.arrow} of Definition~\ref{defn.ieutxo.category} is just a special case of condition~1 above, and
the reverse implication is Lemma~\ref{lemm.validity.preserved}.
\end{proof}

\begin{corr}
\label{corr.composition.arrows.defined.ieutxo}
Composition of arrows as given in Definition~\ref{defn.ieutxo.category}(\ref{ieutxo.arrow.composition})
is well-defined; that is, the composition $\lf'\,\lf$ really is a map from transactions to chunks. 
\end{corr}
\begin{proof}
Continuing the notation of Definition~\ref{defn.ieutxo.category}(\ref{ieutxo.arrow.composition}), by assumption $\lf$ maps $\tx\in\tf{Transaction}_{\mathbb S}$ to some $\lf(\tx)\in\tf{Chunk}_{\mathbb S'}$, and then by Lemma~\ref{lemm.validity.preserved} the action of $\lf'$ maps $\lf(\tx)$ to a valid chunk $\lf'(\lf(\tx))\in\tf{Chunk}_{\mathbb S''}$. 
\end{proof}

\begin{rmrk}[Comment on design]
\label{rmrk.lf.design}
We briefly discuss the design decisions embedded in Definition~\ref{defn.ieutxo.category}:
\begin{enumerate}
\item
The conditions in Corollary~\ref{corr.ab} are more readable than condition~\eqref{eq.condition.a} 
of Definition~\ref{defn.ieutxo.category}(\ref{ieutxo.arrow}), but this comes at the cost of an additional universally quantified parameter $n$.
It is a matter of taste which version we take as primitive: the one in the Definition has fewest parameters and is easiest to check (a higher-level view will be taken later when we develop abstract chunk systems in Section~\ref{sect.acs}).
\item\label{lf.partial}
We could relax the condition to allow $\lf$ to be a partial map.

This would exhibit $\tf{IEUTxO}$ as a subcategory of a larger category with the same objects but more arrows, and in particular it would allow chunks in $\mathbb S$ to cease to be valid when mapped to $\mathbb T$ --- we would still insist that $\lf$ be a \emph{partial} monoid homomorphism on chunks, where everything is defined.

We did not choose this design for this paper, but it might be useful for future work; e.g. following an intuition that $\mathbb S$ is a liberal universe of chunks, and $\lf$ maps it to a stricter universe $\mathbb T$ in which additional restrictions are appended to validators.
Thus, chunks in the liberal world might cease to be valid in the stricter universe. 
\item 
We could also restrict $\lf$ further so that $\lf:\tf{Transaction}_{\mathbb S}\to\tf{Transaction}_{\mathbb T}$.

This would yield fewer arrows, and we prefer to allow the flexibility of mapping a single transaction in $\mathbb S$ to a chunk of transactions in $\mathbb T$; following an intuition that $\mathbb S$ is a coarse-grained representation which $\lf$ maps into a finely-grained representation where something that was considered a single transaction is now a chunk.
\end{enumerate}
\end{rmrk}

\subsection{Idealised UTxO}
\label{subsect.iutxo}

One special case of IEUTxO deserves its own discussion:

\begin{rmrk}[Idealised UTxO]
\label{rmrk.utxo}
Recall from Figure~\ref{fig.ieutxo} that validators take as input a \emph{pointed} transaction:
$$
\tf{Transaction}_!\subseteq \finpow_!(\tf{Input})\times\finpow(\tf{Output}) .
$$ 
Recall also from Notation~\ref{nttn.pointed}(\ref{pointed.finite.set})
and that a \emph{pointed transaction} is a transaction with one distinguished input of that transaction.
For convenience we will call this the \deffont{input-point} of the transaction. 

The UTxO model --- on which Bitcoin is based --- is the special case of EUTxO where validators just examine the input-point.  
So intuitively, in the UTxO model a validator of an output sees just the input that points to that output, in the sense of Notation~\ref{nttn.ty.points.to}, and it does not pay any attention to the transaction in which that input occurs.

We therefore obtain an \deffont{Idealised UTxO} (\deffont{IUTxO}) model from Figure~\ref{fig.ieutxo} just by changing the line for validators to:
$$
\tf{Validator} \subseteq \powerset(\beta\times\tf{Input}) .
$$ 
 
There is an easy embedding map which we can write $1$, taking an IUTxO model to an IEUTxO model, derived from the embedding
$$
\powerset(\beta\times\tf{Input}) \longrightarrow
\powerset(\beta\times\tf{Transaction}_!) 
$$
which is itself derived from the projection taking a pointed transaction to its input-point:
$$
\tf{Transaction}_! \longrightarrow\tf{Input} .
$$ 
Then we can define a category of \deffont{IUTxO models} such that 
\begin{itemize*}
\item
objects are IUTxO models, and 
\item
arrows are functions exactly as defined in Definition~\ref{defn.ieutxo.category}.
\end{itemize*}
\end{rmrk}

\begin{prop}
\label{prop.1}
The mapping $\f{e}$ extends to a categorical embedding\footnote{A functor that is injective on objects and bijective on arrows.} $\tf{IUTxO}\to\tf{IEUTxO}$.
\end{prop}
\begin{proof}
Direct from the construction, since an $\tf{IUTxO}$ model is identified with an $\tf{IEUTxO}$ model whose validators ignore the transaction and just look at the input-point.
\end{proof}

\begin{rmrk}
\label{rmrk.identify}
For convenience, we may treat $\tf{IUTxO}$ as a direct subset of $\tf{IEUTxO}$ --- abusing notation we could write $\tf{IUTxO}\subseteq\tf{IEUTxO}$ --- thus identifying an IUTxO model with an IEUTxO model whose validators only check the input-point of their transaction.
Thus for instance we wrote `is identified with' in Proposition~\ref{prop.1}.
It will always be clear what is intended and we could always unroll the injections if required.
\end{rmrk}

\section{Abstract chunk systems: \texorpdfstring{$\tf{ACS}$}{ACS}}
\label{sect.acs}

\subsection{Basic definitions}

\begin{rmrk}
\label{rmrk.promise.good}
IEUTxO models are good, because they abstract key features of blockchain architectures in a simple and (I would argue) clear manner: output, input, and (valid) combination of transactions to form chunks and then blockchains.

However, IEUTxOs are concrete.
An IEUTxO model is full of internal structure, by its very construction as a solution to type equations in Figure~\ref{fig.ieutxo}.
We will now set about developing an axiomatic, algebraic account of the essential features that make IEUTxOs interesting.
\end{rmrk}

We recall some basic definitions:
\begin{defn}
Suppose $\ns X$ is a set and ${\leq}\subseteq\ns X^2$ is a relation on $\ns X$.
Call $(\ns X,\leq)$ a \deffont{well-ordering} when:
\begin{enumerate*}
\item
$\leq$ is a partial order (reflexive, transitive, anti-symmetric), and
\item
$\leq$ is well-founded (every descending chain is eventually stationary).\footnote{Alternative and equivalent definition: every \emph{strictly} descending chain is finite.}
\end{enumerate*}
As per Notation~\ref{nttn.equivariant}, $\ns X$ and $\leq$ are also assumed equivariant.
\end{defn}

\begin{xmpl}
This should be familiar, but we give examples: 
\begin{itemize}
\item
$(\mathbb Z,\leq)$ is not well-founded.
\item
$(\powerset(\atoms),\subseteq)$ and $(\finpow(\atoms),\subseteq)$ are well-founded.
\item
$[\mathbb N]$ (lists of numbers) with sublist inclusion is well-founded.
\end{itemize}
\end{xmpl}

\begin{defn}
\label{defn.atomic.elements}
Suppose $(\ns X,\mbot,\mtop,\leq)$ is a partial order with an equivariant least element $\mbot$ and an equivariant greatest element $\mtop$.
Call $x\in \ns X$ \deffont{atomic} when 
\begin{enumerate*}
\item\label{atomic.proper}
$\mbot\lneq x\lneq \mtop$ and 
\item
for every $x'\in\ns X$ if $x'\leq x$ then either $x'=\mbot$ or $x'=x$. 
\end{enumerate*}
Write $\f{atomic}(\ns X)$ for the set of atomic elements of $\ns X$ (see also Definition~\ref{defn.atomic}).

If we call $x\in\ns X$ \deffont{proper} when it is neither $\mbot$ nor $\mtop$ (following the standard terminology of \emph{proper subset}), then an atomic element is ``a minimal proper element''.
\end{defn}

\begin{rmrk}
The set of atomic elements $\f{atomic}(\ns X)$ is not to be confused with the set of atoms $\atoms$ from Definition~\ref{defn.zfa.atoms}.
This name collision is just a coincidence. 
\end{rmrk}

Lemma~\ref{lemm.atomic.IEUTxO} will be useful later:
\begin{lemm}
\label{lemm.atomic.IEUTxO}
Suppose $\mathbb T\in\tf{IEUTxO}$ and consider $\tf{Chunk}_{\mathbb T}$ (valid lists of transactions; see Definition~\ref{defn.chunk}) as a partial order under sublist inclusion $\leq$.

Then the atomic elements in $(\tf{Chunk}_{\mathbb T},\leq)$ are precisely the singleton chunks (Notation~\ref{nttn.chunk}).
\end{lemm}
\begin{proof}
Using Corollary~\ref{corr.sublist.inclusion.chunks}.
\end{proof}

\subsection{Monoid of chunks}

\begin{defn}
\label{defn.monoid.of.chunks}
Assume we have equivariant data $(\ns X,\mbot,\mtop,\leq,\mact)$ where:
\begin{itemize*}
\item
$\ns X$ is a set. 
\item
$\mbot,\mtop\in\ns X$ are called \deffont{unit} and \deffont{fail} respectively. 
\item
${\leq} \subseteq\ns X^2$ is a relation.
\item
$\mact : \ns X^2\to\ns X$ is a \deffont{composition}. 
\end{itemize*}
Call $(\ns X,\mbot,\mtop,\leq,\mact)$ a \deffont{monoid of chunks} when:
\begin{enumerate}
\item\label{mbot.cdot}
$\mbot\mact x=x=x\mact \mbot$. 
\item\label{mtop.cdot}
$\mtop\mact x=\mtop=x\mact\mtop$.
\item
$\leq$ is a well-ordering for which the unit $\mbot$ is a bottom element and the paradoxical element $\mtop$ is a top element.
\item
Composition $\mact$ is \deffont{associative}, and \deffont{monotone} in both components, meaning that 
$$
\begin{array}{r@{\ }l@{\quad\text{implies}\quad}r@{\ }l}
x'\leq& x   &   x'\mact y\leq& x\mact y\quad\text{and}
\\
y\leq& y'   &   x\mact y\leq& x\mact y'.
\end{array}
$$
\item
Composition is \deffont{increasing} in the sense that
$$
x\leq x\mact y
\quad\text{and}\quad y\leq x\mact y.
$$
\item\label{acs.locality}
If $x_1,\dots,x_n\in\ns X$ and $x_1\mact \ldots\mact x_n=\mtop$, then there must exist $1\leq i< j\leq n$ such that $x_i\mact x_j=\mtop$.
\end{enumerate} 
\end{defn}

\begin{rmrk}
A few comments on Definition~\ref{defn.monoid.of.chunks}:
\begin{enumerate}
\item
This is a clearly an abstraction of IEUTxO structure, where $\mact$ is chunk composition and $\leq$ is list inclusion (proof in Proposition~\ref{prop.lF.monoid.atomic}).

This is the key instance of the axioms that motivates the definition --- IEUTxOs have more structure, but monoids of chunks is where we start.
See also Example~\ref{xmpl.partial.monoids}(\ref{ieutxo.monoid.chunks}).
\item
$x\mact y$ is not necessarily a least upper bound for $\{x,y\}$.

Take $X=\{1,2\}$ and $x=[1]$ and $y=[2]$ in Example~\ref{xmpl.partial.monoids}(\ref{finite.lists.with.top}) (finite lists with a top element).
Then $x\mact y$ and $y\mact x$ are distinct and incomparable, so both are upper bounds for $\{x,y\}$ but $x\mact y\not\leq y\mact x$ and $y\mact x\not\leq x\mact y$. 
\item
We see that condition~\ref{acs.locality} of Definition~\ref{defn.monoid.of.chunks} closely resembles Lemma~\ref{lemm.invalidity.must.be.somewhere}, and indeed the condition is inspired by that very Lemma.
We will use this in Proposition~\ref{prop.lF.monoid.atomic}.
\end{enumerate}
\end{rmrk}

\begin{nttn}
As is standard, we may write $\ns X$ for both a monoid of chunks and its carrier set.
See for instance the first line of Definition~\ref{defn.factor}.
\end{nttn}

\begin{defn}
\label{defn.generated.as.a.monoid}
\label{defn.factor}
Suppose $\ns X=(\ns X,\mbot,\mtop,\leq,\mact)$ is a monoid of chunks.
\begin{enumerate}
\item
If $x\in \ns X$ and $[x_1,\dots,x_n]\in[\f{atomic}(\ns X)]$ is a finite list of atomic elements\footnote{Functional programmers, who may be used to distinguishing between types (which are primitive) and sets (which inhabit powerset types), may perceive this definition as subtly broken, since it appears to apply a type-former $[\dots]$, to a set $\f{atomic}(\ns X)$.  This is a culture clash and is not an issue with the maths as set up in this paper. 

We are working in ZFA; the carrier set $\ns X$ is a set and so is $\f{atomic}(\ns X)$ (and both are equivariant); the list set-former $[\text{\dots}]$ is a \emph{set}-former, not a type-former.  Thus, $[\f{atomic}(\ns X)]$ is well-defined by Notation~\ref{nttn.pointed}(\ref{list.concat}) as the set of finite lists of atomic elements from $\ns X$.}
in $\ns X$ and 
$$
x = x_1\mact\ldots\mact x_n
\quad\text{then say that}\quad 
\text{$x$ \deffont{factorises} as $[x_1,\dots,x_n]$.} 
$$
\item
Say that $\ns X$ is \deffont{generated by its atomic elements} when every $x\in\ns X\setminus\{\mtop\}$ has a (possibly non-unique) factorisation into atomic elements.
\end{enumerate}
\end{defn}

\begin{defn}\leavevmode
\label{defn.atomic}
\begin{enumerate}
\item
Call a monoid of chunks $\ns X=(\ns X,\mbot,\mtop,\leq,\mact)$ \deffont{atomic} when:
\begin{enumerate*}
\item\label{atomic.generated}
$\ns X$ is generated as a monoid by its atomic elements (Definition~\ref{defn.generated.as.a.monoid}).
\item\label{atomic.factorisation}
There exists a \deffont{factorisation function} $\f{factor}:\ns X\setminus\{\mtop\}\to[\f{atomic}(\ns X)]$ such that for every $x,y\in\ns X\setminus\{\mtop\}$ 
\begin{enumerate*}
\item
$\f{factor}(x)$ factorises $x$ (Definition~\ref{defn.generated.as.a.monoid}) and 
\item\label{factorisation.monoidal}
$\f{factor}(x\mact y)=\f{factor}(x)\mact\f{factor}(y)$ (the right-hand $\mact$ denotes list concatentation; the left-hand $\mact$ is the monoid action in $\ns X$). 
\end{enumerate*}
\end{enumerate*}
In words we say that $\ns X$ is atomic when there is a homomorphism of partially-ordered monoids from $\ns X\setminus\{\mtop\}$ to the space of possible factorisations of its elements.
The relevance of this condition is discussed in Remark~\ref{rmrk.why.factorisation}.
\item\label{perfectly.atomic}
Call $\ns X$ \deffont{perfectly atomic} when it is atomic and furthermore:
\begin{enumerate*}
\item\label{perfectly.atomic.unique}
factorisations into atom elements are unique and 
\item\label{perfectly.atomic.leq}
if $x\leq y<\mtop$ and $x=x_1\mact\ldots\mact x_m$ and 
$y=y_1\mact\ldots\mact y_n$ then 
$[x_1,\dots,x_m]\leq[y_1,\dots,y_n]$ (sublist inclusion).
\end{enumerate*}
The relevance of this condition is discussed in Proposition~\ref{prop.why.pure}.
\end{enumerate}
\end{defn}

\begin{xmpl}
\label{xmpl.partial.monoids}
Suppose $X$ is an equivariant set.
Then:
\begin{enumerate}
\item
$\powerset(X)$ forms a monoid of chunks, where $\mbot=\varnothing$ and $\mtop=X$, and $\leq$ is subset inclusion, and composition $\mact$ is sets union.
It is atomic if and only if $X$ is finite (recall: factorisations must be finite).

We obtain a factorisation function by choosing any order on $X$, and listing elements of any $X'\subseteq X$ in order. 
\item
$\powerset(X)$ forms a monoid of chunks, where:
\begin{itemize*}
\item
$\mbot=\varnothing$ and $\mtop=X$.
\item
$\leq$ is subset inclusion.
\item
$x\mact y=x\cup y$ if $x\cap y=\varnothing$, and $x\mact y=\mtop$ otherwise.
\end{itemize*}
It is atomic if and only if $X$ is finite.
\item
$\finpow(X)\cup \{X\}$ (finite sets of atoms, with a top element) forms an atomic monoid of chunks, using either of the two definitions above for $\powerset(X)$.
\item
\label{finite.lists.with.top}
Finite lists with a top element $[X]^\top$ --- meaning finite lists of elements from $X$, plus one extra `top' element $\top$ --- form a perfectly atomic monoid of chunks as follows:
\begin{itemize*}
\item
$\mbot = []$ and $\mtop = \top$.
\item
$\leq$ is sublist inclusion (Notation~\ref{nttn.pointed}(\ref{sublist.inclusion})) and $l\leq\mtop$ for every finite list $l$.
\item
Composition $\mact$ is list concatenation on lists, and $x\mact \mtop = \mtop = \mtop\mact x$ for any $x$ (list, or $\mtop$).
\end{itemize*}
\item
Finite lists with a top element $[X]^\top$ form an atomic (but not perfectly atomic) monoid of chunks as above, where $\leq$ and $\mact$ are defined as follows: 
\begin{itemize*}
\item
$l\leq l'$ holds when $l$ is not a singleton list and $l$ is a sublist of $l'$.

So $[]\leq [x]$ and $[]\leq [x,z]\leq [x,y,z]$ but $[x]\not\leq [x,y]$; and the proper atomic elements are singleton and two-element lists.
\item
$[]\mact l = l\mact [] = l$ for any list.
\item
$\mtop\mact x = \mtop = x\mact\mtop$ for any $x$.
\item
If $l$ and $l'$ are non-empty lists, then $l\mact l'$ is $l_{init}$ concatenated with $l'_{tail}$, where $l_{init}$ is everything except for the last element of $l$, and $l'_{tail}$ is everything except for the first element of $l'$. 
\end{itemize*}
\item\label{ieutxo.monoid.chunks}
As touched on above, if $\mathbb T$ is an IEUTxO model then $\mathbb T$ gives rise to a perfectly atomic monoid of chunks.
See Proposition~\ref{prop.lF.monoid.atomic}.
\end{enumerate}
\end{xmpl}

\begin{rmrk}
It might seem counterintuitive to make failure $\mtop$ a \emph{top} element in Definition~\ref{defn.monoid.of.chunks}, especially if we are used to seeing domain models where `failure' is intuitively `non-termination' and features $\bot$ as a bottom element.  

We have a concrete reason for this: our canonical IEUTxO models are based on lists ordered by sublist inclusion, so bottom is already occupied by the empty list $[]$ which plays the role of $\mbot$ (see Definition~\ref{defn.lF}).  

But also we have abstract justifications: if we think of a chunk system as a many-valued logic (in which truth-values are chunks or blockchains and $\leq$ reflects how they accumulate transactions over time), then to exhibit a $\top$ is to \emph{fail} to exhibit a concrete witness.
Or (thinking perhaps of callCC~\cite{clinger:schhls}) we can think of $\mtop$ as a `final' or `escape' element.
\end{rmrk}

\subsection{Behaviour, positions, and equivalence}

\subsubsection{Left- and right-behaviour} 

\begin{defn}
\label{defn.left.right.app}
Suppose $\ns X=(\ns X,\mbot,\mtop,\leq,\mact)$ is a monoid of chunks.
Then we have natural \deffont{left-} and \deffont{right-behaviour} functions: 
$$
\begin{array}{l@{\ }c@{\ }l@{\qquad}l@{\ }c@{\ }l}
\f{leftB}:\ns X&\to& \powerset(\ns X)
&
\f{rightB}:\ns X&\to& \powerset(\ns X)
\\
\f{leftB}: x &\mapsto&\{y{\in}\ns X \mid y\mact x<\mtop\}
&
\f{rightB}:x &\mapsto&\{y{\in}\ns X \mid x\mact y<\mtop\}
\end{array}
$$
\end{defn}

\begin{lemm}
Suppose $\ns X$ is a monoid of chunks.  Then we have:
\begin{enumerate*}
\item
$\f{leftB}(\mbot)=\f{rightB}(\mbot)=\ns X\setminus\{\mtop\}$. 
\item
$\f{leftB}(\mtop)=\f{rightB}(\mtop)=\varnothing$. 
\end{enumerate*}
\end{lemm}
\begin{proof}
A fact of Definition~\ref{defn.monoid.of.chunks}(\ref{mbot.cdot}\&\ref{mtop.cdot}).
\end{proof}

\begin{rmrk}
\label{rmrk.observe.failure}
If we think of $\mtop$ as a failure element, and we think of $x\mact y$ as being a composition of which we can observe whether it fails or succeeds,
then
\begin{itemize*}
\item
$\f{rightB}$ maps $x\in\ns X$ to its right-observational behaviour, and 
\item
$\f{leftB}$ maps $x\in\ns X$ to its left-observational behaviour.
\end{itemize*}
\end{rmrk}

\begin{rmrk}
\label{rmrk.lambda.pi}
Parallels can be made in Definition~\ref{defn.left.right.app} with 
the $\lambda$-calculus~\cite{barendregt:lamcss}
and 
the $\pi$-calculus~\cite{Milner:comms}: 
\begin{enumerate}
\item
In the $\lambda$-calculus, a standard observable is non-termination.

Here we are doing something similar, except that (as noted in Remark~\ref{rmrk.observe.failure}) instead of failure to \emph{terminate} ($\bot$) we observe failure to \emph{compose} ($\mtop$), and we consider combination both to the left and to the right. 

Continuing the analogy, in the untyped $\lambda$-calculus, the left-behaviour set of a term $t$ would be those $s$ such that $st$ terminates; and the right-behaviour set of $t$ would be those $s$ such that $ts$ terminates.
\item
The $\pi$-calculus has notions of communications across channels, and as noted in Remark~\ref{rmrk.chunks.and.channels} we see a resemblance with communication of an input and output on a position. 
However there are differences, including:
\begin{enumerate*}
\item
\emph{Validation} is not primitive in the $\pi$-calculus but it is a core precept here. 
\item
Communication in the plain $\pi$-calculus (without considering dialects) is non-deterministic --- one channel name can be invoked by multiple inputs and outputs --- whereas here a key assumption is that every channel name (i.e. position) must have one input and one output --- and if not, the chunk collapses to a failure error-state $\mtop$ (cf. the conditions in Definitions~\ref{defn.chunk} and~\ref{defn.blockchain}).
\item
Name-restriction in the $\pi$-calculus is not automatic but instead is managed by an explicit restriction term-former.
In contrast here a communicating channel (an output-input pair) automatically closes when used once.
We say `closed' and not `bound' because the name remains visible in $\f{up}$ (see also $\stx$ in the IEUTxO models); it is just that no further communication may occur along it.
We discuss garbage-collecting names in chunks in Subsection~\ref{subsect.garbage-collection}.
\end{enumerate*}
\end{enumerate}
\end{rmrk}

\begin{rmrk}
\label{rmrk.surprising.amount}
Definitions~\ref{defn.posi} and~\ref{defn.left.right.up} will build on Definitions~\ref{defn.monoid.of.chunks} and~\ref{defn.left.right.app} to derive a full notion of an \emph{observable interface} of a monoid element, all derived just from the partiality of composition.
We will make good use of this in the rest of the development, for instance Definition~\ref{defn.rG} depends on it.
\end{rmrk}

\subsubsection{Positions} 

\begin{defn}
\label{defn.posi}
Suppose $\ns X$ is a monoid of chunks.
Define $\posi(x)\subseteq\atoms$ the \deffont{positions} of $x\in\ns X$ as follows:
$$
\begin{array}{r@{\ }l}
\posi(\mtop_{\ns X})=&\varnothing
\\
\posi(x)
=&
\{a\in\atoms \mid \Forall{\pi{\in}\f{fix}(a)}\pi\pact x\not\in\f{leftB}(x)\cup\f{rightB}(x) \}
\quad (x\in\ns X\setminus\{\mtop_{\ns X}\}).
\end{array}
$$
($\f{fix}(a)$ from Definition~\ref{defn.fix}.)

Thus $a\in\posi(x)$ when $x\neq\mtop_{\ns X}$ and $\pi\pact x\not\in\f{leftB}(x)\cup\f{rightB}(x)$, for any $\pi$ such that $\pi(a)=a$.
\end{defn}

\begin{rmrk}
Note that the $\pact$ in $\pi\pact x$ in Definition~\ref{defn.posi} above refers to the atoms-permutation action from Definition~\ref{defn.zfa.perm}, not to the partial monoid action $\mact$ from Definition~\ref{defn.popm}.
\end{rmrk}

\begin{rmrk}
\label{rmrk.pos.posi}
In words, $\posi(x)$ from Definition~\ref{defn.posi} is those atoms 
such that there is no permutation fixing $a$ such that $\pi\pact x$ can be successfully combined (left or right) with $x$.

What is the intuition here?  

The name $\posi$ reminds us of $\pos$ from Definition~\ref{defn.pos}, though the definitions are quite different.
They are indeed related; in fact, they are equal in a sense made formal in Proposition~\ref{prop.pos.eq.posi} (see also Lemma~\ref{lemm.tx.inj}(\ref{tx.posi.pos})). 

We do not have all the machinery in place yet, so it may be helpful to point forwards here and observe that conditions~\ref{oriented.lr.empty} and~\ref{oriented.fresh.commute} of Definition~\ref{defn.oriented} can be read as a way to make name-clash into an observable.

So intuitively, Definition~\ref{defn.posi} --- once combined with the notion of an oriented monoid from Definition~\ref{defn.oriented} --- can use permutations to observe name-clash: it measures the live communication channels $a\in\atoms$ in an element $x$ by forcing name-clashes between $a$-channels with $\pi$-renamed variants $\pi\pact x$ for $\pi\in\f{fix}(a)$. 
More details will follow, and see Remark~\ref{rmrk.pos.posi.2}.
\end{rmrk}

\begin{defn}
\label{defn.left.right.up}
Suppose $\ns X=(\ns X,\mbot,\mtop,\leq,\mact)$ is a monoid of chunks, and suppose $x\in \ns X$ and $a\in\atoms$.
\begin{enumerate}
\item\label{points.left}
If 
$$
a\in\posi(x)
\quad\text{and}\quad
\Exists{y{\in}\f{leftB}(x)}a\in\posi(y), 
$$
then say that $a$ \deffont{points left} in $x$. 

Write $\f{left}(x)\subseteq\atoms$ for the set of atoms that point left in $x$.
\item\label{points.right}
If 
$$
a\in\posi(x)
\quad\text{and}\quad 
\Exists{y{\in}\f{rightB}(x)}a\in\posi(y),
$$
then say that $a$ \deffont{points right} in $x$.

Write $\f{right}(x)\subseteq\atoms$ for the set of atoms that point right in $x$.
\item\label{points.up}
If $a$ points neither left nor right in $a$ and yet $a\in\posi(x)$, so that
$$
a\in\posi(x)
\quad\text{and}\quad
\Forall{y{\in}\f{leftB}(x)\cup\f{rightB}(x)}a\not\in\posi(y),
$$
then say that $a$ \deffont{points up} in $x$.

Write $\f{up}(x)\subseteq\atoms$ for the set of atoms that point up in $x$.
\end{enumerate}
\end{defn}

Lemma~\ref{lemm.up.fail} expresses intuitively that atoms that point `up' in a transaction cannot engage in successful (non-failing) combination; they are `stuck interfaces':
\begin{lemm}
\label{lemm.up.fail}
Suppose $\ns X$ is a monoid of chunks and $x,y\in\ns X$ and $a\in\f{up}(x)$.
Then 
$$
a\in\posi(y)
\quad\text{implies}\quad x\mact y=y\mact x=\mtop.
$$
\end{lemm}
\begin{proof}
Direct from Definition~\ref{defn.left.right.up}(\ref{points.up}).
\end{proof}

\begin{rmrk}
The reader who sees similarities between the $\f{left}$, $\f{right}$, and $\f{up}$ of Definition~\ref{defn.left.right.up}, and the $\utxi$, $\utxo$, and $\stx$ of Definition~\ref{defn.utxo.utxi} is right: see Propositions~\ref{prop.pos.eq.posi}, Lemma~\ref{lemm.acs.up.stx}, and Proposition~\ref{prop.left.to.blocked}.
\end{rmrk}

A simple lemma will be helpful:
\begin{lemm}
\label{lemm.supp.lru}
Suppose $\ns X$ is a monoid of chunks and $x\in\ns X$.
Then:
$$
\begin{array}{r@{\ }l}
\f{up}(x)=&\posi(x)\setminus(\f{left}(x)\cup\f{right}(x))
\\
\varnothing=&\f{left}(x)\cap\f{up}(x)
\\
\varnothing=&\f{right}(x)\cap\f{up}(x)
\\
\posi(x)=&\f{left}(x)\cup\f{right}(x)\cup\f{up}(x) 
\end{array}
$$
\end{lemm}
\begin{proof}
This just rephrases clause~\ref{points.up} of Definition~\ref{defn.left.right.up}.
\end{proof}

\begin{rmrk}
\label{rmrk.supp.pos}
Lemmas~\ref{lemm.supp.lru} and~\ref{lemm.utxi.utxo.empty}(\ref{utxo.cup.utxo}) are similar but note that the status of the underlying datatypes is somewhat different: 
\begin{itemize*}
\item
A chunk $\f{ch}\in\tf{Chunk}$ of an IEUTxO model is full of internal structure, and operations on it are defined in terms of that structure, whereas
\item
an element $x\in\ns X$ in a monoid of chunks is an abstract entity and we assume nothing about its internal structure.
\end{itemize*}
Thus, a similarity between them has significance: it is a sanity check on our model and indicates that something rather abstract (monoids of chunks) is accurately following the behaviour of something more concrete (IEUTxO models).
\end{rmrk}

\begin{rmrk}
Lemma~\ref{lemm.supp.lru} expresses that every position in some $x\in\ns X$ (Definition~\ref{defn.posi}) must point in a direction in $\{\f{left},\f{right},\f{up}\}$, and it cannot point both left and up, or both right and up.

Note that Definition~\ref{defn.left.right.up} admits a possibility that an atom could point both left and right; this \emph{cannot} happen in the IEUTxO models (see Lemma~\ref{lemm.utxi.utxo.empty}(\ref{utxi.cap.utxo})).
We will exclude this when we introduce the notion of an oriented monoid of chunks; see Corollary~\ref{corr.lr.tx.empty}.
\end{rmrk}

\begin{lemm}\leavevmode
\label{lemm.left.right.mbot.mtop}
\begin{enumerate*}
\item
$\f{left}(\mbot)=\f{right}(\mbot)=\f{up}(\mbot)=\varnothing$.
\item
$\f{left}(\mtop)=\f{right}(\mtop)=\f{up}(\mtop)=\varnothing$. 
\item
As a corollary using Lemma~\ref{lemm.supp.lru}, $\posi(\mbot)=\posi(\mtop)=\varnothing$.\footnote{$\posi(\mtop)=\varnothing$ is also immediate from Definition~\ref{defn.posi}.}
\end{enumerate*}
\end{lemm}
\begin{proof}
We check the behaviour of $\mbot$ and $\mtop$ as specified in Definition~\ref{defn.monoid.of.chunks} against the definitions of $\f{left}$, $\f{right}$, and $\f{up}$ in Definition~\ref{defn.left.right.up} and see that this is true.
\end{proof}

\subsubsection{Observational equivalence} 
\label{subsect.obs.eq}

\begin{defn}
\label{defn.circ}
Suppose $\ns X$ is a monoid of chunks. 
\begin{enumerate}
\item\label{obs.equiv}
Call $x$ and $x'$ in $\ns X$ \deffont{observationally equivalent} and write 
$$
x\sim x'
\quad\text{when}\quad
\f{leftB}(x)=\f{leftB}(x')\land \f{rightB}(x)=\f{rightB}(x') .
$$
\item\label{circ.commute}
Say that $x$ and $y$ \deffont{commute (up to observational equivalence)} when
$$
x\mact y\sim y\mact x .
$$ 
\end{enumerate}
\end{defn}

We start with a simple but useful sanity check:
\begin{lemm}
\label{lemm.commute.commute}
Suppose $\ns X$ is a monoid of chunks and $x,y\in\ns X$.
Then if $x$ and $y$ commute then 
$$
x\mact y<\mtop \liff y\mact x<\mtop
\quad\text{and}\quad 
x\mact y=\mtop \liff y\mact x=\mtop .
$$ 
\end{lemm}
\begin{proof}
We unpack Definitions~\ref{defn.circ}(\ref{obs.equiv}\&\ref{circ.commute}) and~\ref{defn.left.right.app} and conclude that
$$
x\mact y\mact\mbot <\mtop \liff y\mact x\mact\mbot <\mtop .
$$ 
The result follows, since $\mbot$ is the unit for $\mact$. 
\end{proof}

\begin{lemm}
Suppose $\ns X$ is a monoid of chunks. 
Then if $x\sim x'$ (Definition~\ref{defn.circ}) then 
$$
\f{left}(x)=\f{left}(x')
\quad\text{and}\quad
\f{right}(x)=\f{right}(x')
\quad\text{and}\quad
\f{up}(x)=\f{up}(x').
$$
\end{lemm}
\begin{proof}
A fact of Definitions~\ref{defn.circ}(\ref{obs.equiv}) and~\ref{defn.left.right.up}.
\end{proof}

\subsection{Oriented monoids}

\subsubsection{Definition and properties}

\begin{defn}
\label{defn.oriented}
Suppose $\ns X=(\ns X,\mbot,\mtop,\leq,\mact)$ is a monoid of chunks.

Call $\ns X$ \deffont{oriented} when for all $x,y\in\ns X$:
\begin{enumerate*}
\item\label{oriented.posi.finite}
$\posi(x)\finsubseteq\atoms$.
\item\label{oriented.empty}
If $\posi(x)=\varnothing$ then $x\in\{\mbot,\mtop\}$. 
\item\label{oriented.lr.empty}
If $\f{left}(x)\cap\f{right}(y)\neq\varnothing$ then $x\mact y=\mtop$.
\item\label{oriented.fresh.commute}
If $\posi(x)\cap\posi(y)=\varnothing$ then $x$ and $y$ commute up to observational equivalence (Definition~\ref{defn.circ}(\ref{circ.commute})). 
\item\label{oriented.fresh.defined}
If $\posi(x)\cap\posi(y)=\varnothing$ and $\mtop\not\in\{x,y\}$ then $x\mact y<\mtop$.
\end{enumerate*} 
\end{defn}

\begin{rmrk}
We discuss the conditions of Definition~\ref{defn.oriented} in turn:
\begin{enumerate}
\item
An element $x\in\ns X$ can only be accessible on finitely many channel interfaces.
\item 
The only elements without any interface (meaning atoms that point left right or up) are the unit element ($\leq$-bottom) and the failure element ($\leq$-top).
Compare with the IEUTxO property Lemma~\ref{lemm.supp.empty.empty}.

We use this in Lemmas~\ref{lemm.one.fails} and~\ref{lemm.atomic.x.supp}.
\item
Interfaces always try to connect, but can only \emph{successfully} connect if the directions of their interfaces match up; if not, the whole combination fails.

We use this in Lemma~\ref{lemm.one.fails}, which is required for Proposition~\ref{prop.partial.converse.fresh}.
\item
This condition echoes Lemma~\ref{lemm.fresh.chunks}(\ref{fresh.chunks.obs}).
We use it in Proposition~\ref{prop.partial.converse.fresh}.
\item
Elements with no channels in common, cannot fail to compose. 
\end{enumerate}
We will show later that the IEUTxO models from Definition~\ref{defn.ieutxo.model} are models of Definition~\ref{defn.oriented} in a suitable sense; see Proposition~\ref{prop.lF.oriented}.
\end{rmrk}

We can strengthen Definition~\ref{defn.oriented}(\ref{oriented.empty}) to a logical equivalence:
\begin{lemm}
\label{lemm.oriented.empty.iff}
Suppose $\ns X$ is an oriented monoid of chunks and $x\in\ns X$.
Then 
$$
\posi(x)=\varnothing
\quad\text{if and only if}\quad x\in\{\mbot,\mtop\} .
$$
\end{lemm}
\begin{proof}
The right-to-left implication is direct from Definition~\ref{defn.oriented}(\ref{oriented.empty}).
The left-to-right implication is Lemma~\ref{lemm.left.right.mbot.mtop}.
\end{proof}

Lemma~\ref{lemm.make.link} is a nice way to repackage Definition~\ref{defn.oriented}(\ref{oriented.lr.empty}) in a slightly more accessible wrapper.  
In its form it resembles Lemma~\ref{lemm.utxi.utxo.empty}(\ref{utxi.utxo}), and we use it for Corollary~\ref{corr.lr.tx.empty}:
\begin{lemm}
\label{lemm.make.link}
Suppose $\ns X$ is an oriented monoid of chunks and suppose $x,y\in\ns X$.
Then
$$
x\mact y<\mtop
\quad\text{implies}\quad
\posi(x)\cap\posi(y)\subseteq\f{right}(x)\cap\f{left}(y) 
.
$$
\end{lemm}
\begin{proof}
We consider the possibilities, using Lemma~\ref{lemm.supp.lru}:
\begin{itemize}
\item
\emph{Suppose $a\in\f{left}(x)\cap\posi(y)$.}\ 
From Definition~\ref{defn.oriented}(\ref{oriented.lr.empty}) $a\not\in\f{right}(x)$, and 
by Lemma~\ref{lemm.supp.lru} $a\in\posi(x)$.
It follows from Definition~\ref{defn.left.right.up}(\ref{points.right}) that $x\mact y=\mtop$. 
\item
\emph{Suppose $a\in\f{right}(y)\cap\posi(x)$.}\ 
From Definition~\ref{defn.oriented}(\ref{oriented.lr.empty}) $a\not\in\f{left}(y)$, and
by Lemma~\ref{lemm.supp.lru} $a\in\posi(y)$.
It follows from Definition~\ref{defn.left.right.up}(\ref{points.left}) that $x\mact y=\mtop$. 
\item
\emph{Other cases}\ 
are from Lemma~\ref{lemm.up.fail} (or by direct reasoning from Definition~\ref{defn.left.right.up}(\ref{points.up})).
\end{itemize}  
\end{proof}

Corollary~\ref{corr.lr.tx.empty} is a slightly magical result, in the sense that it is perhaps not immediately obvious that it should follow from our definitions so far.
In its form, if not its proof, it clearly resembles Lemma~\ref{lemm.utxi.utxo.empty}(\ref{utxi.cap.utxo}).
We need it for Lemma~\ref{lemm.lF.gives.valid.singleton.chunks}, that atomic elements in $\ns X$ generate valid singleton chunks under a mapping to IEUTxO models $\lF$: 
\begin{corr}
\label{corr.lr.tx.empty}
Suppose $\ns X$ is an oriented monoid of chunks and $x\in\ns X$.
Then:\footnote{It would be nice to write this as $\f{left}(x)\#\f{right}(x)$ or $\f{left}(x)\bot\f{right}(x)$, but we prefer to trade off notational elegance for clarity and explicitness here, so we will write out our sets disjointness conditions in full.} 
$$
\f{left}(x)\cap\f{right}(x)=\varnothing . 
$$
\end{corr}
\begin{proof}
Suppose $a\in\f{left}(x)$; we will show that $a\in\f{right}(x)$ is impossible.
Consider some $y\in\ns X$ with $a\in\posi(y)$, so that by Lemma~\ref{lemm.supp.lru} $a\in\f{left}(y)\cup\f{right}(y)\cup\f{up}(y)$.
Then:
\begin{itemize}
\item
If $a\in\f{left}(y)$ then $a\in\f{left}(x)\cap\f{left}(y)$ and by Lemma~\ref{lemm.make.link} $x\mact y=\mtop$.
\item
If $a\in\f{right}(y)$ then $a\in\f{left}(x)\cap\f{right}(y)$ and by Lemma~\ref{lemm.make.link} $x\mact y=\mtop$.
\item
If $a\in\f{up}(y)$ then $a\in\f{left}(x)\cap\f{up}(y)$ and by Lemma~\ref{lemm.make.link} $x\mact y=\mtop$.
\end{itemize}
Thus $a\in\posi(y)$ implies $x\mact y=\mtop$ and so $y\not\in\f{rightB}(x)$. 
It follows from Definition~\ref{defn.left.right.up}(\ref{points.right}) that $a\not\in\f{right}(x)$ as required.
\end{proof}

We use Lemma~\ref{lemm.one.fails} for Proposition~\ref{prop.partial.converse.fresh}:
\begin{lemm}
\label{lemm.one.fails}
Suppose $\ns X=(\ns X,\mbot,\mtop,\leq,\mact)$ is an oriented monoid of chunks.
Then
at least one of the following must hold:
$$
x\mact y=\mtop
\qquad
y\mact x=\mtop
\qquad
\posi(x)\cap\posi(y)=\varnothing
$$
\end{lemm}
\begin{proof}
If $x$ or $y$ are equal to $\mbot$ or $\mtop$ then $\posi(x)\cap\posi(y)=\varnothing$ is immediate from Lemma~\ref{lemm.oriented.empty.iff}.

So suppose $x,y\not\in\{\mbot,\mtop\}$, from which it follows by Lemma~\ref{lemm.oriented.empty.iff} (or direct from Definition~\ref{defn.oriented}(\ref{oriented.empty})) that $\posi(x)\neq\varnothing$ and $\posi(y)\neq\varnothing$.
Suppose we have some $a\in\posi(x)\cap\posi(y)$.
We reason by cases using our assumption that $\ns X$ is oriented (Definition~\ref{defn.oriented}):
\begin{itemize*}
\item
If $a\in\f{left}(x)\cap\f{right}(y)$ then $x\mact y=\mtop$
by Definition~\ref{defn.oriented}(\ref{oriented.lr.empty}).
\item
If $a\in\f{right}(x)\cap\f{left}(y)$ then $y\mact x=\mtop$ 
by
by Definition~\ref{defn.oriented}(\ref{oriented.lr.empty}).
\item
If $a\in\f{up}(x)$ or $a\in\f{up}(y)$ then $x\mact y=y\mact x=\mtop$ by Lemma~\ref{lemm.up.fail}.
\item
Other cases are no harder.
\end{itemize*} 
\end{proof}

\begin{rmrk}
Proposition~\ref{prop.partial.converse.fresh} is a partial converse to Definition~\ref{defn.oriented}(\ref{oriented.fresh.commute}) (compare also with Proposition~\ref{prop.fresh.chunks.iff}, which is the same result but for a different structure, and with a very different proof).
It is significant because it equates 
\begin{itemize*}
\item
a static property of positions, with 
\item
a local property of being combinable in either order, with
\item
a global operational property of being commutative up to observation.  
\end{itemize*}
Commutativity is of particular interest in the context of blockchains, because they are by design intended to be distributed, so that we cannot in general know or assume in what order transactions get appended. 
\end{rmrk}

\begin{prop}
\label{prop.partial.converse.fresh}
Suppose $\ns X=(\ns X,\mbot,\mtop,\leq,\mact)$ is an oriented monoid of chunks, and $x,y\in\ns X$.
Suppose further that $x\mact y<\mtop$ or $y\mact x<\mtop$ (not necessarily both).
Then the following conditions are equivalent:
\begin{enumerate*}
\item
$\posi(x)\cap\posi(y)=\varnothing$.
\item
$x\mact y<\mtop \land y\mact x<\mtop$.
\item
$x$ and $y$ commute up to observational equivalence (Definition~\ref{defn.circ}(\ref{circ.commute})).
\end{enumerate*}
\end{prop}
\begin{proof}
First, note that since $x\mact y<\mtop$ or $y\mact x<\mtop$, it must from Definition~\ref{defn.monoid.of.chunks}(\ref{mtop.cdot}) be the case that $x<\mtop$ and $y<\mtop$.

If $x\mact y<\mtop\land y\mact x<\mtop$ then $\posi(x)\cap\posi(y)=\varnothing$ by Lemma~\ref{lemm.one.fails}.
Conversely if $\posi(x)\cap\posi(y)=\varnothing$ then (since $x<\mtop$ and $y<\mtop$) by Definition~\ref{defn.oriented}(\ref{oriented.fresh.defined}) $x\mact y<\mtop$ and $y\mact x<\mtop$.
 
If $\posi(x)\cap\posi(y)=\varnothing$ then $x$ and $y$ commute by Definition~\ref{defn.oriented}(\ref{oriented.fresh.commute}).
Conversely, if $\posi(x)\cap\posi(y)\neq\varnothing$ then by Lemma~\ref{lemm.one.fails} (since $x\mact y<\mtop$ or $y\mact x<\mtop$) we have that $y\mact x=\mtop$ or $x\mact y=\mtop$ respectively, and in particular we have that $x\mact y\neq y\mact x$.
Using Lemma~\ref{lemm.commute.commute} we conclude that $x$ and $y$ do not commute.
\end{proof}

\begin{rmrk}[Comment on design]
There is design freedom to Definition~\ref{defn.oriented}, and we mention this briefly for future work. 
One plausible condition is:
$$
\f{up}(x\mact y)\subseteq\f{up}(x)\cup\f{up}(y)\cup(\f{right}(x)\cap\f{left}(y)).
$$
Intuitively, this states that combining $x$ and $y$ can \emph{only} bind positions that point right in $x$ and point left in $y$.
An even stricter variant would be to insist on equality provided that $x\mact y\neq\mtop$.
\end{rmrk}

\subsubsection{A brief discussion}
\label{subsect.brief.discussion}

We might ask:
\begin{quote}
\emph{Why bother with monoids of chunks?
Why not just work with IEUTxOs?}
\end{quote}
The answer is that we need both: IEUTxOs are the motivating concrete model, and monoids of chunks are their abstraction.

Of course, the IEUTxO equations from Figure~\ref{fig.ieutxo} themselves are an abstraction and generalisation of a concrete inductive definition in~\cite[Figure~3]{chakravarty:extum}, so in overview this paper has the following hierarchy of models, given in increasing order of generality ---
\begin{itemize}
\item
The inductive EUTxO structures from~\cite{chakravarty:extum}.
\item
IEUTxO models from Definition~\ref{defn.solution}.
\item
Abstract chunk systems (ACS) from Definition~\ref{defn.acs.system}.
\end{itemize} 
--- though there is also more going on, because we also map from ACS back down to IEUTxO (Theorem~\ref{thrm.adjoints} and surrounding discussion).

Programmers can think of APIs, which abstract away from internal structure of a concrete implementation; this makes programs more modular, and easier to test and document.

Modern programming languages make it easier to program efficiently using abstract denotations, instantiating only when needed --- one might call this \emph{just-in-time instantiation}.
So even if we have just one concrete model and want one implementation, a good algebraic theory is still relevant to producing working code, because
\begin{itemize*}
\item
it may help structure the mathematics and the code --- the \emph{algebraic graphs with class} Haskell package illustrates how effective this marriage of theory and practice can be~\cite{mokhov:alggwc} --- and 
\item
an abstraction can be read as a library of testable properties, against which implementations can be checked.  
If an implementation fails an axiom --- it's wrong.
\end{itemize*}
So perhaps a better question is this:
\begin{quote}
\emph{What are the essential properties of $\tf{IEUTxO}$ that make it interesting?
How are these properties layered, and nested; and how can they be compactly represented?}
\end{quote}
Definitions~\ref{defn.monoid.of.chunks} and~\ref{defn.oriented} are one set of answers to these questions.
And then, a further question is this:
\begin{quote}
\emph{What other models, aside from $\tf{IEUTxO}$ models, exist of the ACS axioms?} 
\end{quote} 
This paper contains one answer (plus examples; see Subsection~\ref{subsect.acs.examples}): in Definitions~\ref{defn.rG} and~\ref{defn.rG.arrow} and Proposition~\ref{prop.rG.functor} we give a functor taking \emph{any} oriented monoid of chunks to IEUTxOs.
We also map functorially in the other direction in Definitions~\ref{defn.lF} and~\ref{defn.lF.arrow} and Theorem~\ref{thrm.lF.functor}.

The overall results are packaged up in Theorem~\ref{thrm.adjoints}.

\subsection{The category \texorpdfstring{$\tf{ACS}$}{ACS} of abstract chunk systems}

We are now ready to give the algebraic account of IEUTxOs promised in Remark~\ref{rmrk.promise.good}:\footnote{It will take more work to \emph{prove} that this is so.  See the functors $\lF$ and $\rG$ which we define in Sections~\ref{sect.functor.lF} and~\ref{sect.functor.rG}, and Theorem~\ref{thrm.adjoints}.}
\begin{defn}
\label{defn.acs.system}
An \deffont{abstract chunk system} (\deffont{ACS}) is an oriented atomic monoid of chunks (Definitions~\ref{defn.oriented}, \ref{defn.atomic}, and~\ref{defn.monoid.of.chunks}).
\end{defn}

\begin{defn}
\label{defn.acs}
Define $\tf{ACS}$ the \deffont{category of abstract chunk systems} by:
\begin{enumerate}
\item\label{acs.object}
Objects are abstract chunk systems (Definition~\ref{defn.acs.system}).
\item\label{acs.arrow}
Arrows $\rg:\ns X\to\ns Y$ are sets functions from $\ns X$ to $\ns Y$ such that: 
\begin{enumerate*}
\item\label{acs.mbot.mtop}
$\rg(\mbot_{\ns X})=\mbot_{\ns Y}$ and $\rg(\mtop_{\ns X})=\mtop_{\ns Y}$ 
\item\label{abs.less.than.mtop}
$x\leq y<\mtop_{\ns X}$ implies $\rg(x)\leq \rg(y)<\mtop_{\ns Y}$ 
\item\label{f.cdot}
$\rg(x)\mact \rg(y)= \rg(x\mact y)$ 
\end{enumerate*}
\end{enumerate}
Composition of arrows is composition of functions, and the identity arrow is the identity function.
\end{defn}

\begin{rmrk}[Comment on design]\leavevmode
\begin{enumerate*}
\item
We do not insist in Definition~\ref{defn.acs} that $\rg(x)$ must be atomic if $x$ is.
This corresponds to our choice in Definition~\ref{defn.ieutxo.category}(\ref{ieutxo.arrow}) to let $\lf$ map from transactions to chunks, and not from transactions to transactions.
\item
We do insist in Definition~\ref{defn.acs}(\ref{abs.less.than.mtop}) that if $x$ is not the failure element $\mtop_{\ns X}$ in $\ns X$ then $\rg(x)$ is also not the failure element $\mtop_{\ns Y}$ in $\ns Y$.
This corresponds to our choice in Definition~\ref{defn.ieutxo.category}(\ref{ieutxo.arrow}) to make $\lf$ a total function from transactions to chunks, rather than a partial one (cf. discussion in Remark~\ref{rmrk.lf.design}(\ref{lf.partial})).

We could relax this condition by allowing $\rg$ to map $x<\mtop_{\ns X}$ to $\mtop_{\ns Y}$.
There would be nothing wrong with this and it would just exhibit $\tf{ACS}$ as embedded in a larger category with the same objects but more arrows.  
\end{enumerate*}
\end{rmrk}

Lemma~\ref{lemm.atomic.factor} just repackages Definition~\ref{defn.atomic} for objects in $\tf{ACS}$: 
\begin{lemm}
\label{lemm.atomic.factor}
Suppose $\ns X,\ns Y\in\tf{ACS}$ and $x\in\ns X$ and $\rg:\ns X\to\ns Y\in\tf{ACS}$.
Then:
\begin{itemize*}
\item
If $x\neq\mtop$ then there exists some finite (possibly empty, possibly non-unique) list of atomic elements $x_1,\dots,x_n\in \f{atomic}(\ns X)$ such that $x=x_1\mact\ldots\mact x_n$.
\item
If $x\neq\mtop$ (and in particular by Definition~\ref{defn.atomic.elements}(\ref{atomic.proper}) if $x$ is atomic) then there exists some finite (possibly empty, possibly non-unique) list of atomic elements $y_1,\dots,y_n\in \f{atomic}(\ns Y)$ such that $\rg(x)=y_1\mact\ldots\mact y_n$.
\end{itemize*}
\end{lemm}
\begin{proof}
Immediate since by Definition~\ref{defn.acs.system} $\ns X$ is atomic (Definition~\ref{defn.atomic}).
\end{proof}

\begin{prop}\leavevmode
\label{prop.ACS.morphism.atomic}
\begin{enumerate}
\item
An arrow $\rg:\ns X\to\ns Y\in\tf{ACS}$ (Definition~\ref{defn.acs}(\ref{acs.arrow})) is uniquely determined by its action on $\f{atomic}(\ns X)$.
\item\label{atomic.equality.test}
As a corollary, if $\rg,\rg':\ns X\to\ns Y$ are two arrows, then to check the equality of arrows $\rg=\rg'$ it suffices to check that $\rg(x)=\rg'(x)$ for every $x\in\f{atomic}(\ns X)$.
\end{enumerate}
\end{prop}
\begin{proof}
\begin{enumerate}
\item
Consider some $x\in\ns X$.
If $x\in\{\mtop,\mbot\}$ then the action of $f$ is determined by Definition~\ref{defn.acs}(\ref{acs.mbot.mtop}).

Otherwise, using Lemma~\ref{lemm.atomic.factor} write $x=x_1\mact \dots \mact x_n$ for atomic $x_1,\dots,x_n\in \f{atomic}(\ns X)$.
We then have from Definition~\ref{defn.acs}(\ref{f.cdot}) that 
$$
f(x_1)\mact\ldots\mact \rg(x_n)= \rg(x).
$$
Thus $\rg(x)$ is determined by the values of $\rg(x_1)$, \dots, $\rg(x_n)$.
\item
By a routine argument from Definition~\ref{defn.atomic} (since $\ns X$ is atomic) and Definition~\ref{defn.acs}(\ref{f.cdot}).
\end{enumerate}
\end{proof}

\subsection{Examples of abstract chunk systems}
\label{subsect.acs.examples}

We created abstract chunk systems in (Definition~\ref{defn.acs.system}) to abstract away the internal structure of IEUTxO models (Definition~\ref{defn.solution}).

This is a standard move in mathematics: axiomatise, then generalise.
As we argued in Subsections~\ref{subsect.map} and~\ref{subsect.brief.discussion}, \emph{even if} the reader only cares about practical hands-on implementation, for which internal structure is available because we have the code, it is still useful --- and arguably \emph{indispensable} --- to take a moment to get a good higher-view of what it is that is implemented, since axioms provide a language for the higher-level enterprises of comparison, description, specification, correctness, testing, and exposition.

Now we take a moment to go back and consider what the space of concrete models of the ACS axioms looks like.
We give a list of examples, which is not intended to be exhaustive but 
which we hope may illustrate the scope, character, and structure of our definition:

\begin{enumerate}
\item
Consider the set of all finite sets of atoms $A=\{a_1,\dots,a_n\}\finsubseteq\mathbb A$ and $\atoms$ itself:
$$
\ns X=\finpow(\atoms)\cup\{\atoms\} .
$$ 
$\ns X$ forms an ACS such that
\begin{itemize*}
\item
$A\leq B$ when $A\subseteq B$.
\item
$A\mact B=A\uplus B$ (disjoint union) if $A$ and $B$ are disjoint, and $A\mact B=\mathbb A$ otherwise.
\end{itemize*}
Unpacking definitions, we can check that:
\begin{itemize*}
\item
Atomic elements are singletons $\{a\}$.
\item
A factorisation function $\f{factor}$ (Definition~\ref{defn.atomic}(\ref{atomic.factorisation})) is obtained by ordering atoms and listing the finite set in order.
\item
$\posi(A)=\f{up}(A)=A$ and $\f{left}(A)=\f{right}(A)=\varnothing$.
\end{itemize*}
\item
Consider some datatype $\tf{Terms}$ of term syntax over variable symbols $a,b,c,\dots$ (terms of first-order logic $s,t::= a \mid f(t,\dots,t)$ for some term-formers $f$, or terms of the untyped $\lambda$-calculus $s,t::=a \mid ss\mid\lambda a.s$ would suffice). 

Let \deffont{finite substitutions} be finite partial maps from variable symbols to terms generated by \deffont{atomic substitutions} 
$$
\sigma=[a\ssm t] . 
$$ 
Composition is defined in a standard way by acting on terms as $s\sigma$.

If $a\not\in\f{dom}(\sigma)$ then $a\sigma=a$ (so the substitution acts as the identity on variable symbols not in its domain). 

Now, add a \emph{fail} top element $\mtop$, such that $s\mtop=s$ always (so $\mtop$ is a formal element that acts as the identity).  
If the reader likes, we could take $\mtop$ to be $[a\ssm a\mid \text{ all }a]$.

Then substitutions with $\mtop$ form an ACS, where 
\begin{itemize*}
\item
$\sigma\leq\sigma'$ when $\f{dom}(\sigma)\subseteq\f{dom}(\sigma')$, and $\sigma\leq\mtop$ always.
\item
If $\f{dom}(\sigma)\cap\f{dom}(\sigma')=\varnothing$ then $\sigma\mact\sigma'$ is the finite substitution that maps $a$ to $a\sigma\sigma'$ for each $a\in\f{dom}(\sigma)\cup\f{dom}(\sigma')$, and 
\item
if $\f{dom}(\sigma)\cap\f{dom}(\sigma')\neq\varnothing$ then $\sigma\mact\sigma'=\mtop$
\end{itemize*}
The `disjoint domains or fail' condition ensures that composition is monotone in $\leq$.
Unpacking definitions, we check that:
\begin{itemize*}
\item
Atomic elements are the generators $[a\ssm s]$.
\item
A factorisation function $\f{factor}$ (Definition~\ref{defn.atomic}(\ref{atomic.factorisation})) is obtained by ordering atoms and listing a substitution as a list of atomic substitutions in order of the atom on the left.\footnote{This can be a little subtle: if $b$ is ordered before $a$ then $[a\ssm b]\mact [b\ssm c]$ would need to be written as $[b\ssm c]\mact[a\ssm c]$.}
\item
$\f{posi}(\sigma)=\f{left}(\sigma)=\f{dom}(\sigma)$ and $\f{right}(\sigma)=\f{up}(\sigma)=\varnothing$.
\end{itemize*}
\end{enumerate}

\section{The functor \texorpdfstring{$\lF:\tf{IEUTxO}\to\tf{ACS}$}{F : IEUTxO -> ACS}}
\label{sect.functor.lF}

\subsection{Action on objects}

\begin{defn}
\label{defn.lF}
Suppose $\mathbb T=(\alpha,\beta,\tf{Transaction},\tf{Validator})\in\tf{IEUTxO}$ is a model (Definition~\ref{defn.ieutxo.model}) and recall $\tf{Chunk}_{\mathbb T}$ from Notation~\ref{nttn.chunk}.
Then we define an abstract chunk system 
$$
\lF(\mathbb T) = (\tf{Chunk}_{\mathbb T}\cup\{\f{fail}\},[],\f{fail},\leq,\mact) \in \tf{ACS} 
$$ 
as follows:
\begin{itemize*}
\item
The underlying set is $\tf{Chunk}_{\mathbb T}\cup\{\f{fail}\}$ --- so if we write ``$x\in \lF(\mathbb T)$'' this means ``either $x=\f{fail}$ or $x\in\tf{Chunk}_{\mathbb T}$''.
\item
$\mbot=[]$ and $\mtop=\f{fail}$. 
\item
$\leq$ is sublist inclusion (Notation~\ref{nttn.pointed}(\ref{sublist.inclusion})) on chunks that are lists of transactions, with $\mtop$ as a top element, so $\f{ch}\leq\mtop$ always.
\item
$x\mact\mtop=\mtop=\mtop\mact x$, and if $x$ and $y$ are both chunks then $\mact$ is validated concatentation of chunks, defaulting to the failure element $\mtop$ if validation fails.
\end{itemize*}
\end{defn}

\begin{rmrk}
So if $x,y\in\tf{Chunk}_{\mathbb T}$ but $x\mact y\in[\tf{Transaction}_{\mathbb T}]\setminus\tf{Chunk}_{\mathbb T}$ in $\mathbb T$ --- so the two chunks cannot be validly combined in $\mathbb T$ --- then $x\mact y=\f{fail}$ in $\lF(\mathbb T)$.  
We will always be clear whether $\mact$ means `concatenate as lists' or `compose in $\lF(\mathbb T)$'.

$\lF(\mathbb T)$ resembles $\mathbb T$, except with an explicit failure element $\f{fail}$ added.
This makes monoid composition total: not every pair of chunks composes to form a list of transactions that is a chunk --- as per the criteria in Definition~\ref{defn.chunk} --- thus any combination of chunks that is not a chunk, can be set to $\f{fail}$.

We still need to prove that $\lF(\mathbb T)$ is an abstract chunk system (Definition~\ref{defn.acs.system}) in the category $\tf{ACS}$ (Definition~\ref{defn.acs}(\ref{acs.object})) thus:
\begin{itemize*}
\item
$\lF(\mathbb T)$ is a monoid of chunks (Definition~\ref{defn.monoid.of.chunks}), 
\item
$\lF(\mathbb T)$ is atomic (Definition~\ref{defn.atomic}) --- in fact it is perfectly atomic (Definition~\ref{defn.atomic}(\ref{perfectly.atomic})) --- and 
\item
$\lF(\mathbb T)$ is oriented (Definition~\ref{defn.oriented}).
\end{itemize*}
We do this next, culminating with Theorem~\ref{thrm.lF.object.acs}.

Note that this gives us a dual view of a chunk:
\begin{itemize*}
\item
as a chunk in the IEUTxO universe, and
\item
as an abstract element in the ACS universe.
\end{itemize*}
To what extent do these two views correspond, and how can we make this formal?
We prove Proposition~\ref{prop.pos.eq.posi}, which verifies that the positions in a chunk as an element in the IEUTxO universe, coincide with its positions as an element in the ACS universe. 
This is refined in further results; notably Lemma~\ref{lemm.acs.up.stx} and its sharper corollary Proposition~\ref{prop.left.to.blocked}.
\end{rmrk}

\begin{prop}\leavevmode
\label{prop.lF.monoid.atomic}
\begin{enumerate}
\item
$\lF(\mathbb T)$ from Definition~\ref{defn.lF} is a perfectly atomic (Definition~\ref{defn.atomic}(\ref{perfectly.atomic})) monoid of chunks (Definition~\ref{defn.monoid.of.chunks}).

Unpacking Definition~\ref{defn.atomic}(\ref{perfectly.atomic}), any $x\in \lF(\mathbb T)\setminus\{\mtop\}$ can be uniquely decomposed as a (possibly empty) finite list of atomic elements 
$$
x = x_1\mact\ldots\mact x_n,
$$
and if $x\leq y<\mtop_{\lF(\mathbb T)}$ then the factorisation of $x$ is a sublist of the factorisation of $y$.
\item
Furthermore, the atomic elements $x_i\in\f{atomic}(\lF(\mathbb T))$ are singleton lists of the form $[\tx]$ for $\tx\in\tf{Transaction}_{\mathbb T}$.
\end{enumerate}
\end{prop}
\begin{proof}
Most of the properties in Definition~\ref{defn.monoid.of.chunks} are facts of sublist inclusion and concatenation.
Condition~\ref{acs.locality} of Definition~\ref{defn.monoid.of.chunks} is Lemma~\ref{lemm.invalidity.must.be.somewhere}.

Recall that being a chunk is down-closed by Corollary~\ref{corr.sublist.inclusion.chunks}, so in a perfectly atomic monoid of chunks, every chunk is above its atomic chunks.
It is just a fact of lists, sublist inclusion, and the construction in Definition~\ref{defn.lF} that $\lF(\mathbb T)$ is perfectly atomic with atomic elements singleton chunks of the form $[\tx]$ for $\tx\in\tf{Transaction}_{\mathbb T}$.
\end{proof}

\subsection{Relation between the partial monoid \texorpdfstring{$\tf{Chunk}_{\mathbb T}$}{Chunk(T)} and the monoid of chunks \texorpdfstring{$\lF(\mathbb T)$}{F(T)}}

\begin{lemm}
\label{lemm.perm.utxi}
Suppose that:
\begin{itemize*}
\item
$T\in\tf{IEUTxO}$ and $\f{ch}\in\tf{Chunk}_{\mathbb T}$ and $a\in\pos(\f{ch})$.
\item
$\pi\in\f{Perm}$ is a permutation and $\pi(a)=a$, and write $\f{ch}'=\pi\pact\f{ch}$ (Definition~\ref{defn.perm}).
\end{itemize*}
It is a structural fact that 
$\f{ch}\mact\f{ch}',\f{ch}'\mact\f{ch}\in[\tf{Transaction}_{\mathbb T}]$
since a concatentation of two lists is a list.
However:
$$
\f{ch}'\mact \f{ch} \not\in\tf{Chunk}_{\mathbb T} 
\quad\text{and}\quad 
\f{ch}\mact\f{ch}'\not\in\tf{Chunk}_{\mathbb T} .
$$
\end{lemm}
\begin{proof}
By assumption $a\in\pos(\f{ch})$, so using Lemmas~\ref{lemm.utxi.utxo.empty}(\ref{utxo.cup.utxo}) and~\ref{lemm.perm.pos} (since $\pi(a)=a$) precisely one of 
$$
\begin{array}{r@{\ }l}
a\in&\utxi(\f{ch})\cap\utxi(\f{ch}'),
\\
a\in&\utxo(\f{ch})\cap\utxo(\f{ch}'),
\\
\text{or}\quad
a\in&\stx(\f{ch})\cap\stx(\f{ch}')
\end{array}
$$
must hold.
In each of these three cases the result follows from Lemma~\ref{lemm.utxi.utxo.empty}(\ref{utxi.utxo}).
\end{proof}

\begin{prop}
\label{prop.pos.eq.posi}
Suppose $\mathbb T$ is an IEUTxO model and $x\in\tf{Chunk}_{\mathbb T}$, and
\begin{itemize*}
\item
recall $\pos(x)$ from Definition~\ref{defn.pos}, and
\item
noting from Proposition~\ref{prop.lF.monoid.atomic} that $x$ can also be viewed as an element $x\in \lF(\mathbb T)$, recall also $\posi(x)$ from Definition~\ref{defn.posi}.
\end{itemize*}
Then
$$
\pos(x) = \posi(x) .
$$
\end{prop}
\begin{proof}
By a routine calculation from Definitions~\ref{defn.chunk} and~\ref{defn.posi} using Lemma~\ref{lemm.perm.utxi}.
\end{proof}

We continue Remark~\ref{rmrk.pos.posi}:
\begin{rmrk}
\label{rmrk.pos.posi.2}
$\pos$ from Definition~\ref{defn.pos} and $\posi$ from Definition~\ref{defn.posi} have different constructions yet give equal results, in a sense made formal in Proposition~\ref{prop.pos.eq.posi}.
Aside from being a useful equality, what does this tell us?

$\pos$ is \emph{intensional} --- it has and requires full access to the internal structure of its argument --- whereas $\posi$ is \emph{extensional} --- it treats its argument as a black box in which it can only permute atoms and observe compositional behaviour.
See also a similar observation in Remark~\ref{rmrk.supp.pos}.

A significance of Proposition~\ref{prop.pos.eq.posi} is as a (non-trivial) correctness assertion about the overall algebraic framework in which these definitions have been embedded: that the abstract interface of $x$ viewed extensionally, matches the concrete interface of $x$ when viewed intensionally.

Put another way, Proposition~\ref{prop.pos.eq.posi} has the flavour of being a weak but indicative soundness and completeness result relating a class of concrete models with a class of abstract ones. 

Lemma~\ref{lemm.acs.up.stx} refines this, and we will continue to build on these ideas, culminating with Theorem~\ref{thrm.adjoints}.
\end{rmrk}

\subsection{\texorpdfstring{$\lF(\mathbb T)$}{F(T)} is oriented, so \texorpdfstring{$\lF(\mathbb T)\in\tf{ACS}$}{F(T) in ACS}}

\begin{lemm}
\label{lemm.acs.up.stx}
Suppose $\mathbb T\in\tf{IEUTxO}$.
By Proposition~\ref{prop.lF.monoid.atomic} $\lF(\mathbb T)=\tf{Chunk}_{\mathbb T}\cup\{\f{fail}\}$ is a monoid of chunks, so it has notions of $\f{left}$, $\f{right}$, and $\f{up}$ from Definition~\ref{defn.left.right.up}.

Then for $x\in\tf{Chunk}_{\mathbb T}$ we have:\footnote{\dots so we are looking at a `real chunk' here, for which $\utxi$, $\utxo$, and $\stx$ are defined, and excluding our extra failure element, for which they are not defined \dots}
$$
\begin{array}{r@{\ }l}
\f{left}(x)\subseteq& \utxi(x)
\\
\f{right}(x)\subseteq& \utxo(x)
\\
\f{stx}(x)\subseteq&\f{up}(x)
\end{array}
$$
\end{lemm}
\begin{proof}
We consider each line in turn:
\begin{enumerate}
\item
If $a\in\f{left}(x)$ then by Lemma~\ref{lemm.supp.lru} $a\in\posi(x)$ so by Proposition~\ref{prop.pos.eq.posi} also $a\in\pos(x)$.
By Definition~\ref{defn.left.right.up} there exists $y\in\tf{Chunk}_{\mathbb T}$ such that $a\in\posi(y)$, so by Proposition~\ref{prop.pos.eq.posi} also $a\in\pos(y)$, and $y\mact x\in\tf{Chunk}_{\mathbb T}$.
It follows from Lemma~\ref{lemm.utxi.utxo.empty}(\ref{utxi.utxo}) that ($a\in\utxo(y)$ and) $a\in\utxi(x)$ as required.
\item
If $a\in\f{right}(x)$ then $a\in\pos(x)$ and by Definition~\ref{defn.left.right.up} there exists $y\in\tf{Chunk}_{\mathbb T}$ such that $a\in\pos(y)$ and $x\mact y\in\tf{Chunk}_{\mathbb T}$, so by Lemma~\ref{lemm.utxi.utxo.empty}(\ref{utxi.utxo}) ($a\in\utxi(y)$ and) $a\in\utxo(x)$ as required.
\item
The reasoning to prove $\f{stx}(x)\subseteq\f{up}(x)$ is no harder.
\end{enumerate}
\end{proof}

\begin{rmrk}
\label{rmrk.blocked}
Lemma~\ref{lemm.acs.up.stx} is interesting as much for what it is \emph{not}, namely it is not the equality $\f{left}=\utxi$ and $\f{right}=\utxi$ and $\f{up}=\f{stx}$ that one might initially expect.
Why?
\begin{enumerate*}
\item
Consider a chunk $\f{ch}\in\tf{Chunk}_{\mathbb T}$ with an IEUTxO output located at $a$ but with an empty validator (one which validates no inputs).
Then $a\in\utxo(\f{ch})$ in $\mathbb T$, but $a\in\f{up}(\f{ch})$ in $\lF(\mathbb T)$.
\item
Similarly consider a chunk $\f{ch}$ with an input located at $a$ but such that no validator will validate it --- just because an input exists, does not mean a validator must exist to accept it.
Then $a\in\utxi(\f{ch})$ in $\mathbb T$, but $a\in\f{up}(\f{ch})$ in $\lF(\mathbb T)$.
\end{enumerate*}
We return to this with the notion of a \emph{blocked channel}, in Subsection~\ref{subsect.blocked.channel}.
\end{rmrk}

\begin{prop}
\label{prop.lF.oriented}
$\lF(\mathbb T)$ from Definition~\ref{defn.lF} is oriented (Definition~\ref{defn.oriented}).
\end{prop}
\begin{proof}
We check each condition of Definition~\ref{defn.oriented} in turn:
\begin{enumerate}
\item
\emph{We check that $\posi(x)\finsubseteq\atoms$.}

It is a structural fact of Definition~\ref{defn.pos} that $\pos(x)\finsubseteq\atoms$.
We use Proposition~\ref{prop.pos.eq.posi}.
\item
\emph{We check that $\posi(x)=\varnothing$ implies $x=\mbot_{\lF(\mathbb T)}$ or $x=\mtop_{\lF(\mathbb T)}$.}

An element $x\in \lF(\mathbb T)$ is either a chunk or the failure element:
\begin{itemize*}
\item
If $x$ is a chunk and $\posi(x)=\varnothing$ then by Proposition~\ref{prop.pos.eq.posi} $\posi(x)=\varnothing$ and by Lemma~\ref{lemm.supp.empty.empty} $x=[]$, which by Definition~\ref{defn.lF} is $\mbot_{\lF(\mathbb T)}$.
\item
If $x=\f{fail}$ then there is nothing to prove, since $\f{fail}=\mtop_{\lF(\mathbb T)}$.
\end{itemize*}
\item
\emph{We check that $\f{left}(x)\cap\f{right}(y)\neq\varnothing$ implies $x\mact y=\mtop$.}

If $x$ or $y$ are $\f{fail}$ then $x\mact y=\mtop$.

So suppose that $x$ and $y$ are chunks, that is, suppose $x,y\in\tf{Chunk}_{\mathbb T}$.
By Lemma~\ref{lemm.acs.up.stx} $\utxi(x)\cap\utxo(y)\neq\varnothing$.
We use Lemma~\ref{lemm.utxi.utxo.empty}(\ref{utxi.utxo}).
\item
\emph{We check that if $\posi(x)\cap\posi(y)=\varnothing$ then $x$ and $y$ commute (Definition~\ref{defn.circ}(\ref{circ.commute})).}

From Proposition~\ref{prop.pos.eq.posi} and Lemma~\ref{lemm.fresh.chunks}(\ref{fresh.chunks.obs}).
\item
\emph{We check that if $\posi(x)\cap\posi(y)=\varnothing$ and $\mtop\not\in\{x,y\}$ then $x\mact y<\mtop_{\lF(\mathbb T)}$.}

Using Proposition~\ref{prop.pos.eq.posi}, this just rephrases Lemma~\ref{lemm.fresh.chunks.defined} in the language of a monoid of chunks.
\end{enumerate}
\end{proof}

\begin{thrm}
\label{thrm.lF.object.acs}
$\lF(\mathbb T)$ (Definition~\ref{defn.lF}) is an abstract chunk system in $\tf{ACS}$ (Definition~\ref{defn.acs.system}).
In symbols:
$$
\lF(\mathbb T)\in\tf{ACS} .
$$ 
\end{thrm}
\begin{proof}
From Propositions~\ref{prop.lF.monoid.atomic} and~\ref{prop.lF.oriented}.
\end{proof}

\subsection{Action of \texorpdfstring{$\lF$}{F} on arrows}

\begin{defn}
\label{defn.lF.arrow}
Suppose $\lf:\mathbb S\to\mathbb T\in\tf{IEUTxO}$ is an arrow (Definition~\ref{defn.ieutxo.category}(\ref{ieutxo.arrow})).
We define an arrow 
$$
\lF(\lf) : \lF(\mathbb S)\to \lF(\mathbb T)\in\tf{ACS}
$$
by
\begin{equation}
\label{eq.the.construction}
\begin{array}{l@{\ =\ }l}
\lF(\lf)([\tx_1,\dots,\tx_n]) &
\lf(\tx_1)\mact\ldots\mact \lf(\tx_n) 
\\
\lF(\lf)(\f{fail}_{\mathbb S}) &
\f{fail}_{\mathbb T} .
\end{array}
\end{equation}
\end{defn}

\begin{lemm}
$\lF(\lf)$ from Definition~\ref{defn.lF.arrow} does indeed map from $\lF(\mathbb S)$ to $\lF(\mathbb T)$.
\end{lemm}
\begin{proof}
We need to check that validity is preserved, meaning that if $[\tx_1,\dots,\tx_n]$ is a chunk then so is $\lf(\tx_1)\mact\ldots\mact \lf(\tx_n)$.
This is Lemma~\ref{lemm.validity.preserved}.
\end{proof}

\begin{prop}
\label{prop.lF.acts.on.arrow}
Continuing Definition~\ref{defn.lF.arrow}, we have that
$$
\lf:\mathbb S\to\mathbb T\in\tf{IEUTxO}
\quad\text{implies}\quad 
\lF(\lf):\lF(\mathbb S)\to \lF(\mathbb T)\in\tf{IEUTxO}.
$$
Furthermore, $\lF(\lf'\,\lf)=\lF(\lf')\,\lF(\lf)$ and $\lF(\f{id}_{\mathbb S})=\f{id}_{\lF(\mathbb S)}$.
\end{prop}
\begin{proof}
We check the properties in Definition~\ref{defn.acs}(\ref{acs.arrow}) in turn:
\begin{enumerate*}
\item 
\emph{We check that $\lF(\lf)(\mbot_{\lF(\mathbb S)})=\mbot_{\lF(\mathbb T)}$ and $\lF(\lf)(\mtop_{\lF(\mathbb S)})=\mtop_{\lF(\mathbb T)}$.}

This is just the fact that $\lF(\lf)([])=[]$ and $\lF(\lf)(\f{fail}_{\mathbb S})=\f{fail}_{\mathbb T}$. 
\item
\emph{We check that $x\leq y<\f{fail}_{\lF(\mathbb S)}$ implies $\lF(\lf)(x)\leq \lF(\lf)(y)<\mtop_{\lF(\mathbb T)}$.} 

It is a fact of the construction in Definition~\ref{defn.lF.arrow} that $x\leq y$ ($x$ is a sublist of $y$) implies $\lF(\lf)(x)\leq \lF(\lf)(y)$.
\item
\emph{We check that $\lF(\lf)(x)\mact \lF(\lf)(y)= \lF(\lf)(x\mact y)$.}

A fact of the first clause of equation~\eqref{eq.the.construction} in Definition~\ref{defn.lF.arrow}.
\end{enumerate*}
We can check $\lF(\lf'\,\lf)=\lF(\lf')\,\lF(\lf)$ and $\lF(\f{id}_{\mathbb S})=\f{id}_{\lF(\mathbb S)}$ by routine calculations which we elide.
\end{proof}

\begin{thrm}
\label{thrm.lF.functor}
The map $\lF$, with the action on IEUTxO models from Definition~\ref{defn.lF}, and with the action on arrows from Definition~\ref{defn.lF.arrow}, is a functor
$$
\lF : \tf{IEUTxO}\to\tf{ACS} .
$$
\end{thrm}
\begin{proof}
This is Theorem~\ref{thrm.lF.object.acs} and Proposition~\ref{prop.lF.acts.on.arrow}.
\end{proof}

\subsection{Blocked channels}
\label{subsect.blocked.channel}

Recall from Remark~\ref{rmrk.chunks.and.channels} that we can think of positions as communication channels in the $\pi$-calculus sense.
We conclude this Section by taking a little time to refine the subset inclusions from Lemma~\ref{lemm.acs.up.stx}.
For this, we need to consider the possibility of a channel which is (intuitively) \emph{blocked}, in the sense that no successful validation can occur across it:

\begin{defn}
\label{defn.blocked.utxo}
Suppose that $\mathbb T$ is an IEUTxO model and $\f{ch}\in\tf{Chunk}_{\mathbb T}$ and $a\in\atoms$.
\begin{enumerate}
\item
Suppose that
\begin{itemize*}
\item
$a\in\f{utxi}(\f{ch})$ and 
\item
for every $\f{ch}'\in\tf{Chunk}_{\mathbb T}$ with $a\in\f{utxo}(\f{ch}')$,\ $\f{ch}'\mact\f{ch}$ is not a chunk.
\end{itemize*}
Then call $a$ a \deffont{blocked utxi} in $\f{ch}$.
Write $\f{blockedUtxi}(\f{ch})$ for the blocked utxis of $\f{ch}$. 
\item
Similarly define $\f{blockedUtxo}(\f{ch})$ the \deffont{blocked utxos} of $\f{ch}$ to be those $a\in\atoms$ such that
\begin{itemize*}
\item
$a\in\f{utxo}(\f{ch})$ and
\item
for every $\f{ch}'\in\tf{Chunk}_{\mathbb T}$ with $a\in\f{utxi}(\f{ch}')$, $\f{ch}'\mact\f{ch}$ is not a chunk.
\end{itemize*}
\end{enumerate}
\end{defn}

\begin{rmrk}
So a blocked UTxI or UTxO in a chunk is an input or output that exists, but which fails if you try to interact with it.
This could happen for an output whose validator is the empty set (it fails on any input), or for an input such that no validator in $\mathbb T$ exists to validate it (see Remark~\ref{rmrk.blocked}).
\end{rmrk}

\begin{lemm}
\label{lemm.oi.cap.rightleft.cap}
Suppose $\mathbb T$ is an IEUTxO model and $x,y\in\tf{Chunk}_{\mathbb T}$ and $x\mact y\in\tf{Chunk}_{\mathbb T}$.
Then 
$$
\utxo(x)\cap\utxi(y)\subseteq\f{right}(x)\cap\f{left}(y).
$$
\end{lemm}
\begin{proof}
Suppose $a\in\utxo(x)\cap\utxi(y)$.
In particular then by Lemma~\ref{lemm.utxi.utxo.empty}(\ref{utxo.cup.utxo})
$a\in\pos(x)\cap\pos(y)$ so by Proposition~\ref{prop.pos.eq.posi} also $a\in\posi(x)\cap\posi(y)$.

$\lF(\mathbb T)$ is a monoid of chunks by Proposition~\ref{prop.lF.monoid.atomic}, and since $x\mact y\in\tf{Chunk}_{\mathbb T}$ it follows that $x\mact y<\mtop_{\lF(\mathbb T)}$.
It follows from Definition~\ref{defn.left.right.app} that $y\in\f{rightB}(x)$ and $x\in\f{leftB}(y)$.

The result now follows by Definition~\ref{defn.left.right.up}.
\end{proof}

\begin{prop}
\label{prop.left.to.blocked}
Suppose $\mathbb T\in\tf{IEUTxO}$ and $x\in \lF(\mathbb T){\setminus}\{\mtop\}$ (that is, $x\in\tf{Chunk}_{\mathbb T}$).
Then: 
$$
\begin{array}{r@{\ }l}
\f{left}(x)=&\f{utxi}(x)\setminus\f{blockedUtxi}(x)
\\
\f{right}(x)=&\f{utxo}(x)\setminus\f{blockedUtxo}(x)
\\
\f{up}(x)=&\f{stx}(x)\cup\f{blockedUtxi}(x)\cup\f{blockedUtxo}(x)
\end{array}
$$
\end{prop}
\begin{proof}
We know by Lemma~\ref{lemm.acs.up.stx} that $\f{left}(x)\subseteq\utxi(x)$ and $\f{right}(x)\subseteq\utxo(x)$.
Now suppose $a\in\f{utxi}(x)$ and $a\not\in\f{blockedUtxi}(x)$; unpacking Definition~\ref{defn.blocked.utxo} this means that there exists a $y\in\tf{Chunk}_{\mathbb T}$ such that $a\in\utxo(y)$ and $y\mact x\in\tf{Chunk}_{\mathbb T}$.
By Lemma~\ref{lemm.oi.cap.rightleft.cap} it follows that ($a\in\f{right}(y)$ and) $a\in\f{left}(x)$.

The case of $\f{right}(x)$ is similar, and the case of $\f{up}(x)$ follows from the previous two cases and Lemma~\ref{lemm.acs.up.stx}.
\end{proof}

\section{The functor \texorpdfstring{$\rG:\tf{ACS}\to\tf{IEUTxO}$}{G : ACS -> IEUTxO}}
\label{sect.functor.rG}

\subsection{A brief discussion: why represent?}

In Subsection~\ref{subsect.brief.discussion} we observed a hierarchy of models, from concrete EUTxO structures to IEUTxO models to abstract chunk systems.

The mapping from IEUxO to ACS is the functor $\lF:\tf{IEUTxO}\to\tf{ACS}$ from Section~\ref{sect.functor.lF}.
We will now exhibit a functor $\rG:\tf{ACS}\to\tf{IEUTxO}$ \emph{back down} from the abstract to the concrete structures.\footnote{Note that $\rG$ consists of an action on objects, and an action on arrows, and we can usefully have the former without the latter. See Remark~\ref{rmrk.why.factorisation}.}

This is interesting for two reasons: one specific, and one general. 
We consider each in turn.

$\rG$ is interesting because:
\begin{enumerate}
\item
It gives a sense in which the abstraction reasonably represents the concrete models.
That is, there is nothing the abstract model could do that is so crazy that it cannot be engineered back down to a concrete structure.
This may involve some ugly concrete fiddling, emulation, and choices --- as one might expect going from an abstract to a concrete object --- but it can be done, and seeing how, can be helpful for understanding both worlds.
\item
Sometimes, theorems are better proved in the concrete world than the abstract world.
This can be particularly useful to prove negative properties, that something \emph{cannot} happen in the abstract world, because it would correspond to something that would be impossible in the concrete world.

A well-known example is that every Boolean Algebra can be represented as a powerset, and thus every finite Boolean Algebra has cardinality a power of two.
Thus, to prove that some abstract structure does \emph{not} admit any Boolean Algebra structure, it suffices if its carrier set is finite and has cardinality that is \emph{not} a power of two.\footnote{We do not exhibit any such application of our result in this paper; we are just making the general observation.  Still, it is possible that in future work our constructions might be put to such use.}
\end{enumerate}
Now to understand the relevance of $\rG$ specifically for this paper, consider the following question:
\begin{quote}
\emph{In what sense is Definition~\ref{defn.acs.system} a good abstraction of Definition~\ref{defn.solution}?}
\end{quote}
Design decisions are embedded in the conditions of Definition~\ref{defn.acs.system}, and some of these were not trivial and had more than one plausible outcome.
Why did we choose as we did?  How do these choices interact?  In what sense were they appropriate?

To answer these questions, $\lF$ is not necessarily the greatest help on its own.
To illustrate why, consider that we can obtain a general `theory of blockchains' merely by insisting that an `abstract chunk system' is a set.
We impose no further structure: \emph{et voil\`a}: instant generality!\footnote{This really happened.  An author lifted an algebra from one of my papers, deleted crucial structure, and claimed superior generality.  When the paper went to me to referee, I observed that deleting this structure also deleted all the interesting theorems.  This was not necessarily fatal; but what other theorems or properties were there to replace them?  No reply was forthcoming.} 
But this tells us little; e.g. $\lF$ would just be the forgetful functor, mapping an IEUTxO model to its underlying set.
We could map just to monoids, if we add the failure element, and again an $\lF$ would exist, but this would be only slightly less uninformative.

So where is the sweet spot, and why?  As we observed, merely exhibiting an $\lF$-style functor does not help: we need to get an algebraic measure of what it is about Figure~\ref{fig.ieutxo} that gives it its essential nature. 

We get a formally meaningful measure of an appropriate level of abstraction by locating one at which we can build a sensible functor $\rG$ \emph{going back}, and seeing how conditions in Definition~\ref{defn.acs.system} interact with its construction --- and, we can observe how tweaking them can affect, or even break, these constructions.
A discussion of such tweaks, and their effects, is in Remarks~\ref{rmrk.why.factorisation} and~\ref{rmrk.factor}, and Proposition~\ref{prop.why.pure}. 

\begin{rmrk}[Comment on design]
\label{rmrk.sensible}
What counts as a `sensible' $\rG$ is a design decision in itself.
We consider several options in this paper (listed here in increasing order of size of the category of denotations $\tf{ACS}$): 
\begin{enumerate*}
\item
a categorical equivalence (Proposition~\ref{prop.why.pure}); or 
\item
a categorical embedding (Theorem~\ref{thrm.adjoints} and Remark~\ref{rmrk.loop.of.embeddings}); or 
\item
just an injection on objects (Remark~\ref{rmrk.why.factorisation}).
\end{enumerate*}
All these possibilities are justifiable.\footnote{A comparison: when giving a denotation to $\mathbb N$, the domain of denotations could be $\omega$ (equivalence), ordinals (embedding), or just an arbitrary infinite set (injection on objects).  All three possibilities are reasonable, depending on the context.}  
So to be clear: $\rG$ and the choices we make in building it are not intended as direct value judgements; they are a way to measure and explore the structure of a large, abstract, and interesting design space.
\end{rmrk}

\subsection{Action on objects}

Recall from Definition~\ref{defn.solution} the notion of an IEUTxO model, and the accompanying discussion in Remark~\ref{rmrk.NQR} about the status of the injection $\nu:\tf{Validator}\hookrightarrow\powerset(\beta\times\tf{Transaction}_!)$.

Continuing that Remark, in Definition~\ref{defn.rG} we must be explicit about $\nu$:
\begin{defn}
\label{defn.rG}
Suppose $(\ns X,\mbot,\mtop,\leq,\mact)\in\tf{ACS}$.
We define an IEUTxO model $\rG(\ns X)$  
$$
\rG(\ns X)=(\alpha,\beta,\tf{Transaction},\tf{Validator},\nu:\tf{Validator}\hookrightarrow\powerset(\beta\times\tf{Transaction}_!)) 
$$ 
as follows:
\begin{enumerate}
\item
We take $\alpha=\beta=\f{atomic}(\ns X)$ (Definition~\ref{defn.atomic.elements}).
\item
We take:
$$
\tf{Validator}=\{\ast\} 
$$
where $\{\ast\}$ is a unit type.
\item\label{F.arrow}
For each atomic $x\in\f{atomic}(\ns X)$ we admit a transaction 
$$
\tf{tx}(x)\in\tf{Transaction}
$$ 
such that:
$$
\begin{array}{r@{\ }l}
\f{input}(\tf{tx}(x)) =& \{(a,x) \mid a\in\f{left}(x)\}
\\
\f{output}(\tf{tx}(x)) =&\{(b,x,\ast)\mid b\in\f{right}(x)\cup\f{up}(x)\} .
\end{array}
$$
Thus:
$$
\tf{Transaction} = \{\tf{tx}(X) \mid x\in\f{atomic}(\ns X)\}
$$
for $\tf{tx}$ defined as above.
\item\label{F.nu}
We define $\nu:\tf{Validator}\hookrightarrow\powerset(\beta\times\tf{Transaction}_!)$ to map $\ast\in\tf{Validator}$ as follows ($\at i$ from Notation~\ref{nttn.ty.points.to}): 
$$
\nu(\ast) =
\{(x,\tx\at (p,y)) \mid (p,y)\in\f{input}(\tx),\ x\mact y<\mtop\} .
$$
Thus, $\nu(\ast)$ is the function that inputs $x$ and a pointed transaction $\tx\at (p,y)$, extracts the data $y$ from the input, and then checks that $x\mact y<\mtop$ in $\ns X$.\footnote{In fact, the only pointed transactions possible in this system are $\tf{tx}(y)\at (p,y)$, so we could also extract $y$ from $\tx$.}
\end{enumerate}
\end{defn}

\begin{rmrk}
\label{rmrk.its.utxo}
Continuing Remark~\ref{rmrk.utxo}, we see that the validator used by $\rG$ in Definition~\ref{defn.rG} is UTxO-style; it only examines the (pointed) input of the transaction to be validated.
So in fact, $\rG$ maps not just to IEUTxO models but to the \emph{IUTxO} models noted in Subsection~\ref{subsect.iutxo}.
We will use this observation in Theorem~\ref{thrm.adjoints}(3).
\end{rmrk}

An easy sanity check:
\begin{lemm}
Suppose $\ns X\in\tf{ACS}$ and $x,y\in\ns X$.
Then $\tf{tx}(x)\mact\tf{tx}(y)$ is a chunk if and only if $x\mact y<\mtop$. 
\end{lemm}
\begin{proof}
By construction, unravelling Definition~\ref{defn.rG}.
\end{proof}

\begin{rmrk}[Comment on design]
In Definition~\ref{defn.rG}(\ref{F.arrow}) we set 
\begin{itemize*}
\item
$\f{input}(\tf{tx}(x)) = \{(a,x) \mid a\in\f{left}(x)\}$ and 
\item
$\f{output}(\tf{tx}(x)) =\{(b,x,\ast)\mid b\in\f{right}(x)\cup\f{up}(x)\}$.
\end{itemize*}
So $\f{up}$-atoms in $x$ map to $\f{output}$-atoms in $\tf{tx}(x)$.
Why? 
For two reasons: 
\begin{itemize}
\item
\emph{The short reason} is that it makes Lemma~\ref{lemm.tx.inj} work: all atoms in $\posi(x)$ get recorded in an input or output (even the ones in $\f{up}(x)$, which cannot participate in a non-failing interaction), along with a copy of $x$ (to get injectivity).
\item
\emph{The longer reason} is as follows:

An ACS element $x\in\ns X$ has no internal structure and thus no explicit structural notion of inputs or outputs.
Our only interaction with $x$ is by combining it with other elements and observing partiality (cf. Remark~\ref{rmrk.surprising.amount}).

But suppose our ACS $\ns X$ was obtained concretely from an IEUTxO model using $\lF$ from Definition~\ref{defn.lF}, so that $x$ is `secretly' a singleton chunk, presented as an atomic element in an ACS.
Then $p\in\f{up}(x)$ could occur for two reasons: 
\begin{itemize*}
\item
either $p$ is the position of an input which no output will accept (perhaps it is labelled with some data that all validators disapprove of); 
\item
or $p$ is the position of an output that will not validate any available input (e.g. it has the empty validator).
\end{itemize*}
When we come to map $x$ back to a transaction $\tf{tx}(x)$, the simplest way to record $p$ is to attach it to an IEUTxO output located at $p$, with a validator that always fails.

The other option would be to attach $p$ to an \emph{input} in $\tf{tx}(x)$, to tag the data carried by that input with some special `fail-me' tag, and remember to create only validators that recognise this tag and reject the input.
But this is clearly a more complicated way of doing things, and the design adopted in Definition~\ref{defn.rG}(\ref{F.arrow}) seems the natural and simple approach.
\end{itemize} 
\end{rmrk}

Several things about Definition~\ref{defn.rG} need checked.
We start with Lemma~\ref{lemm.lF.gives.valid.singleton.chunks}:
\begin{lemm}\leavevmode
\label{lemm.lF.gives.valid.singleton.chunks}
\begin{enumerate}
\item
If $\ns X\in\tf{ACS}$ and $x\in\f{atomic}(\ns X)$ then $\tf{tx}(x)$ has the right type to be a transaction as per Figure~\ref{fig.ieutxo}.
\item
If $\ns X\in\tf{ACS}$ and $x\in\f{atomic}(\ns X)$ then $[\tf{tx}(x)]$ is a chunk.
\item\label{atomic.lF}
As a corollary, atomic elements in $\lF(\ns X)$ are precisely the singleton chunks of $\tf{tx}(x)$, where $x$ ranges over atomic elements of $\ns X$ --- or more concisely in symbols: 
$$
\f{atomic}(\lF(\ns X))=\{[\tf{tx}(x)]\mid x\in\f{atomic}(\ns X)\}.
$$
\end{enumerate}
\end{lemm}
\begin{proof}\leavevmode
\begin{enumerate}
\item
From Definitions~\ref{defn.oriented}(\ref{oriented.posi.finite}) and~\ref{defn.rG}(\ref{F.arrow}) and Lemma~\ref{lemm.supp.lru}, $\tf{tx}(x)$ has finitely many inputs and outputs; so as per Figure~\ref{fig.ieutxo} it is indeed a pair of a \emph{finite} set of inputs and a \emph{finite} set of outputs.
\item
By Lemma~\ref{lemm.singleton.chunk.valid}, to show $[\tf{tx}]$ is a valid chunk
it would suffice to show that $\f{input}(\tf{tx}(x))\cap\f{output}(\tf{tx}(x))=\varnothing$.
This follows from Lemma~\ref{lemm.supp.lru} and Corollary~\ref{corr.lr.tx.empty}. 
\item
By construction and Lemma~\ref{lemm.atomic.IEUTxO}, noting that in lists ordered by subset inclusion, atomic elements are singleton lists.
\end{enumerate}
\end{proof}

\begin{rmrk}[Comment on design]
We mention an alternative definition of $\rG$ from Definition~\ref{defn.rG}, just to illustrate that more than one encoding is possible:
\begin{enumerate}
\item
We take $\alpha=\beta=\tf{Validator}=\f{atomic}(\ns X)$ (Definition~\ref{defn.atomic.elements}).
\item
For each atomic $x\in\f{atomic}(\ns X)$ we admit a transaction $\tf{tx}(x)\in\tf{Transaction}$ such that:
$$
\begin{array}{r@{\ }l}
\f{input}(\tf{tx}(x)) =& \{(a,x) \mid a\in\f{left}(x)\}
\\
\f{output}(\tf{tx}(x)) =&\{(b,x,x)\mid b\in\f{right}(x)\cup\f{up}(x)\}
\end{array}
$$
\item
We define $\nu:\tf{Validator}\hookrightarrow\powerset(\beta\times\tf{Transaction}_!)$ to map 
$x\in\tf{Validator}$ as follows: 
$$
\nu(x) =
\{(x,\tf{tx}(y)\at i) \mid x\mact y<\mtop,\ i\in\f{input}(\tf{tx}(y))\} .
$$
\end{enumerate}
\end{rmrk}

\begin{rmrk}
\label{rmrk.it's.ok}
It remains to prove that $\nu$ is well-defined and (as required by Definition~\ref{defn.solution}(\ref{validator.injection})) is injective, and that $\rG(\ns X)$ is indeed an IEUTxO model.
See Corollaries~\ref{corr.nu.inj} and~\ref{corr.rG.IEUTxO}.
\end{rmrk}

\subsection{\texorpdfstring{$\nu$}{nu} is injective}

\begin{lemm}
\label{lemm.atomic.x.supp}
Suppose $\ns X\in\tf{ACS}$ (Definition~\ref{defn.acs}).
Then 
$$
x\in\f{atomic}(\ns X)
\quad\text{implies}\quad
\posi(x)\neq\varnothing.
$$
\end{lemm}
\begin{proof}
Suppose $x$ is atomic.
By Definition~\ref{defn.atomic.elements}(\ref{atomic.proper}) $x\not\in\{\mbot,\mtop\}$. 
We use Lemma~\ref{lemm.oriented.empty.iff}.
\end{proof}

\begin{lemm}
\label{lemm.tx.inj}
Suppose $\ns X\in\tf{ACS}$.
Then:
\begin{enumerate*}
\item\label{tx.inj}
The assignment  
$$
x\in\f{atomic}(\ns X) \longmapsto \tf{tx}(x) \in \tf{Transaction}_{\rG(\ns X)}
$$
from Definition~\ref{defn.rG}(\ref{F.arrow}) is injective. 
\item\label{tx.posi.pos}
$\posi(x)=\pos(\tf{tx}(x))$.
\end{enumerate*}
\end{lemm}
\begin{proof}
\begin{enumerate}
\item
By Lemma~\ref{lemm.atomic.x.supp} (since $x$ is atomic) $\posi(x)\neq\varnothing$, so by Lemma~\ref{lemm.supp.lru} at least one of $\f{left}(x)$ or $\f{right}(x)$ or $\f{up}(x)$ must be nonempty.

If $\f{left}(x)$ is nonempty then $\tf{tx}(x)$ has an input and we can read $x$ off the data in that input.

Otherwise $\tf{tx}(x)$ has an output and we can read $x$ off the data in that output.
\item
It follows from Definition~\ref{defn.rG}(\ref{F.arrow}) and Figure~\ref{fig.positions} that $\pos(\tf{tx}(x))=\f{left}(x)\cup(\f{right}(x)\cup\f{up}(x))$.
Also, by Lemma~\ref{lemm.supp.lru} $\pos(x)=\f{left}(x)\cup\f{right}(x)\cup\f{up}(x)$.
\end{enumerate}
\end{proof}

As promised in Remark~\ref{rmrk.it's.ok}, we prove:
\begin{corr}
\label{corr.nu.inj}
The map 
$$
\nu:\tf{Validator}\hookrightarrow\powerset(\beta\times\tf{Transaction}_!)
$$ 
from Definition~\ref{defn.rG}(\ref{F.nu}) is well-defined and injective. 
\end{corr}
\begin{proof}
By Lemma~\ref{lemm.tx.inj}(\ref{tx.inj}) we can deduce $y$ from $\tf{tx}(y)$.
The result follows. 
\end{proof}

\begin{corr}
\label{corr.rG.IEUTxO}
If $\ns X$ is an abstract chunk system (Definition~\ref{defn.acs.system}) then $\rG(\ns X)$ is an IEUTxO model (Definition~\ref{defn.solution}).
\end{corr}
\begin{proof}
We just need to check the conditions of Definition~\ref{defn.solution}; the only nontrivial part is that $\nu$ is an injection, and that is Corollary~\ref{corr.nu.inj}.
\end{proof}

\subsection{Action on arrows}

\begin{defn}
\label{defn.rG.arrow}
Suppose $\rg:\ns X\to\ns Y\in\tf{ACS}$ and recall from Definition~\ref{defn.ieutxo.category}(\ref{ieutxo.arrow}) that an arrow 
$$
\rG(\rg)=\rG(\ns X)\to\rG(\ns Y) \in \tf{IEUTxO} 
$$
should be a mapping from $\tf{Transaction}_{\rG(\ns X)}$ to $\tf{Chunk}_{\rG(\ns Y)}$.
Recall $\f{factor}$ from Definition~\ref{defn.atomic}(\ref{atomic.factorisation}), which factorises non-failure elements into atomic constituents, and recall from Definition~\ref{defn.acs}(\ref{abs.less.than.mtop}) that $\rg$ maps non-failure elements to non-failure elements.

Then define $\rG(\rg)$ by 
\begin{multline*}
\rG(\rg) : \tf{tx}(x) \longmapsto \tf{tx}(y_1)\mact\ldots\mact \tf{tx}(y_n)\in\tf{Chunk}_{\rG(\ns Y)} 
\\
\quad\text{where}\quad \f{factor}(\rg(x))=[y_1,\dots,y_n]\in[\f{atomic}(\ns Y)].
\end{multline*}
\end{defn}

\begin{lemm}
\label{lemm.rGf.well-defined}
$\rG(\rg)$ from Definition~\ref{defn.rG.arrow} is well-defined.
\end{lemm}
\begin{proof}
We must check that $\tf{tx}(y_1)\mact\ldots\mact \tf{tx}(y_n)$ is a chunk;
this follows by Lemma~\ref{lemm.validity.preserved}.
\end{proof}

\begin{prop}
\label{prop.rG.functor}
The map $\rG$, with the action on abstract chunk systems from Definition~\ref{defn.rG}, and with the action on arrows from Definition~\ref{defn.rG.arrow}, is a functor
$$
\rG : \tf{ACS}\to\tf{IEUTxO} .
$$
\end{prop}
\begin{proof}
Given the results above, the only remaining thing to check is that if $\rg:\ns X\to \ns Y$ and $\rg':\ns Y\to\ns Z$ then $\rG(\rg'\,\rg)=\rG(\rg')\rG(\rg)$.
This follows by a routine argument from the definitions, using Definition~\ref{defn.atomic}(\ref{factorisation.monoidal}).
\end{proof}

\begin{rmrk}[Comment on design]
\label{rmrk.why.factorisation}
The significance of condition~\ref{atomic.factorisation} of Definition~\ref{defn.atomic} is not that elements can be factored into atomic elements --- this follows already from condition~\ref{atomic.generated} --- but that a factorisation can be \emph{selected}, as a monoid homomorphism. 

We use this condition in just one place: to define the action of $\rG$ on arrows in Definition~\ref{defn.rG.arrow} and prove it functorial in Proposition~\ref{prop.rG.functor}.

It would be legitimate to remove condition~\ref{atomic.factorisation} of Definition~\ref{defn.atomic}.
This would exhibit our category $\tf{ACS}$ as a subcategory of a larger category which would include more objects --- `even more abstract' abstract chunk systems --- at a cost of no longer being able to functorially map this larger space back down to $\tf{IEUTxO}$.

Intuitively, this larger space behaves more like a space of all possible denotations rather than a space of $\tf{IEUTxO}$-representable ones, which might be worthwhile if it admits other interesting examples; whether or not this will be the case cannot be predicted at time of writing.

Note that the action on \emph{objects} from Definition~\ref{defn.rG} would still be well-defined even without Definition~\ref{defn.atomic}(\ref{atomic.factorisation}), so that we can still represent our `even more abstract' abstract chunk systems concretely in IEUTxO models: this just would not correspond to a functor.\footnote{Thus, the $\rG$ in the loop of embeddings in Remark~\ref{rmrk.loop.of.embeddings} would weaken a mapping on objects.}
\end{rmrk}

\section{An adjunction between \texorpdfstring{$\lF:\tf{IEUTxO}\to\tf{ACS}$}{F : IEUTxO -> ACS} and \texorpdfstring{$\rG:\tf{ACS}\to\tf{IEUTxO}$}{G : ACS -> IEUTxO}}

\subsection{The counit map \texorpdfstring{$\epsilon_{\ns X}:\lF\rG(\ns X)\to\ns X$}{epsilon : FG(X) -> X} exists and is a surjection}

\begin{rmrk}
\label{rmrk.epsilon}
Suppose $\ns X\in\tf{ACS}$; we wish to define an arrow $\epsilon_{\ns X}:\lF\rG(\ns X)\to\ns X\in\tf{ACS}$.
Unpacking Definitions~\ref{defn.rG} and~\ref{defn.lF}, we see that an $x\in \lF\rG(\ns X)$ has one of the following forms:
\begin{itemize*}
\item
$x=\f{fail}_{\lF\rG(\ns X)}$ for $\f{fail}_{\lF\rG(\ns X)}$ the \emph{failure} element added by $\lF$ to $\tf{Chunk}_{\rG(\ns X)}$ in Definition~\ref{defn.lF}.
\item
$x=[\tf{tx}(x_1),\dots,\tf{tx}(x_n)]$ for some unique $[x_1,\dots,x_n]\in[\f{atomic}(\ns X)]$.
\end{itemize*}
We also know from Lemma~\ref{lemm.tx.inj} that $\tf{tx}:\f{atomic}(\ns X)\to\tf{Transaction}(\rG(\ns X))$ is injective, and it follows that we can recover each $x_i$ from the unique corresponding $[\tf{tx}(x_i)]$ above.
\end{rmrk}

\begin{defn}
\label{defn.epsilon}
Let $\epsilon_{\ns X}:\lF\rG(\ns X)\to\ns X$ be determined by:
\begin{enumerate*}
\item
$\epsilon_{\ns X}([])=\mbot_{\ns X}$ (this would be a special case of the next clause, for $n=0$)
\item
$\epsilon_{\ns X}([\tf{tx}(x_1),\dots,\tf{tx}(x_n)]) = x_1\mact\ldots\mact x_n$ for $n\geq 1$ and $x_1,\dots,x_n\in\f{atomic}(\ns X)$ 
\item
$\epsilon_{\ns X}(\mtop_{\lF\rG(\ns X)}) = \mtop_{\ns X}$
\end{enumerate*}
\end{defn}

\begin{lemm}
Definition~\ref{defn.epsilon} is well-defined and determines an arrow in $\tf{ACS}$.
\end{lemm}
\begin{proof}
Well-definedness follows as per Remark~\ref{rmrk.epsilon} from Lemma~\ref{lemm.tx.inj}, since we can recover each $x_i$ from its $[\tf{tx}(x_i)]$.
It remains to check the arrow conditions from Definition~\ref{defn.acs}(\ref{acs.arrow}): 
\begin{enumerate*}
\item
\emph{We check that $\epsilon_{\ns X}(\mbot_{\lF\rG(\ns X)})=\mbot_{\ns X}$ and $\epsilon(\mtop_{\lF\rG(\ns X)})=\mtop_{\ns X}$.} 

A fact of Definition~\ref{defn.epsilon}.
\item
\emph{We check that $x\leq y<\mtop_{\lF\rG(\ns X)}$ implies $\epsilon_{\ns X}(x)\leq \epsilon_{\ns X}(y)<\mtop_{\ns X}$.} 

By Proposition~\ref{prop.lF.monoid.atomic} $\lF\rG(\ns X)$ is perfectly atomic, and if $x\leq y<\mtop_{\lF\rG(\ns X)}$ then $x$ and $y$ factorise uniquely into a composition of lists of singleton chunks, 
which by construction in Definition~\ref{defn.rG}(\ref{F.arrow}) have the form $\tf{tx}(x_i)$ and $\tf{tx}(y_j)$, 
such that 
the factorisation of $x$ is a sublist of the factorisation of $y$.  

The result follows by a routine calculation from Definition~\ref{defn.epsilon}.
\item
$\epsilon_{\ns X}(x)\mact \epsilon_{\ns X}(y)= \epsilon_{\ns X}(x\mact y)$ 

A fact of Definition~\ref{defn.epsilon}, again using the fact that by Proposition~\ref{prop.lF.monoid.atomic} $\lF\rG(\ns X)$ is perfectly atomic.
\end{enumerate*}
\end{proof}

\begin{prop}
\label{prop.epsilon.bijective}
The counit map $\epsilon_{\ns X}:\lF\rG(\ns X)\to\ns X\in\tf{ACS}$ is a surjection on underlying sets.
\end{prop}
\begin{proof}
Suppose we are given $x\in\ns X$; we want to exhibit an element in $\lF\rG(\ns X)$ that maps to it under $\epsilon_{\ns X}$.

If $x=\mtop_{\ns X}$ then by construction in Definition~\ref{defn.lF} $\f{fail}_{\lF\rG(\ns X)} = \mtop_{\lF\rG(\ns X)}$ and 
also by construction $\epsilon_{\ns X}(\f{fail}_{\lF\rG(\ns X)})=\mtop_{\ns X}$ so we are done.

Otherwise by Proposition~\ref{prop.lF.monoid.atomic} (or just by Definition~\ref{defn.atomic}) we can write 
$$
x=x_1\mact\ldots\mact x_n
$$ 
for some atomic $x_1,\dots,x_n\in\f{atomic}(\ns X)$, and 
looking at Definition~\ref{defn.epsilon} we immediately have that 
$$
\epsilon_{\ns X} ([\tf{tx}(x_1),\dots,\tf{tx}(x_n)]) = x_1\mact\ldots\mact x_n = x.
$$
\end{proof}

\begin{rmrk}[Comment on design] 
\label{rmrk.factor}
By Proposition~\ref{prop.epsilon.bijective} $\epsilon_{\ns X}$ surjects $\lF\rG(\ns X)$ onto $\ns X$ as sets, but is it not necessarily surjective on the $\leq$ structure, meaning that $\epsilon_{\ns X}(\f{ch})\leq_{\ns X}\epsilon_{\ns X}(\f{ch}')$ does not imply $\f{ch}\leq_{\lF\rG(\ns X)} \f{ch}'$.

We can have this if we add condition~\ref{perfectly.atomic.leq} of Definition~\ref{defn.atomic}, that:
if $x\leq y<\mtop$ then $\f{factor}(x)\leq\f{factor}(y)$ (the right-hand $\leq$ denotes sublist inclusion; the left-hand $\leq$ is the partial order on $\ns X$).  
More on this in Proposition~\ref{prop.why.pure}.
Further tweaks to the design of abstract chunk systems are also mentioned in Remark~\ref{rmrk.why.factorisation}.
\end{rmrk}

\subsection{The unit map \texorpdfstring{$\eta_{\mathbb T}:\mathbb T\to \rG\lF(\mathbb  T)$}{eta : T -> GF(T)} exists and is an isomorphism}

\begin{rmrk}
Suppose $\mathbb T=(\alpha,\beta,\tf{Transaction},\tf{Validator})\in\tf{IEUTxO}$.
We wish to define an arrow 
$$
\eta_{\mathbb T}:\mathbb T\to \rG\lF(\mathbb T) .
$$
We can make some observations:
\begin{itemize}
\item
By construction in Definition~\ref{defn.lF}, $\lF(\mathbb T)$ is isomorphic as a partial ordering to $\mathbb T$, with the addition of the $\f{fail}$ top element.

As noted in Proposition~\ref{prop.lF.monoid.atomic}
it follows that the atomic elements of $\lF(\mathbb T)$ correspond precisely with the singleton chunks in $\tf{Chunk}_{\mathbb T}$, and thus with transactions in $\tf{Transaction}_{\mathbb T}$.

The chunks in $\lF(\mathbb T)$ are then determined by combining the singleton chunks, subject to the well-formedness conditions of Definition~\ref{defn.chunk} and the locality properties noted in Lemma~\ref{lemm.invalidity.must.be.somewhere}.
\item
By construction in Definition~\ref{defn.rG} the atomic elements of $\rG\lF(\mathbb T)$ are isomorphic to $\f{atomic}(\lF(\mathbb T))$, and by Lemma~\ref{lemm.lF.gives.valid.singleton.chunks}(\ref{atomic.lF}) also to $\tf{Transaction}_{\mathbb T}$.
\end{itemize}
\end{rmrk}

\begin{defn}
\label{defn.eta}
Let $\eta_{\mathbb T}:\mathbb T\to \rG\lF(\mathbb T)$ be determined by mapping $\tx\in\tf{Transaction}_{\mathbb T}$ to $[\tf{tx}(\tx)]\in \tf{Chunk}_{\rG\lF(\mathbb T)}$.
Thus using Lemma~\ref{lemm.validity.preserved} we have:
$$
\eta_{\mathbb T}([\tx_1,\dots,\tx_n]) = [\tf{tx}(\tx_1),\dots,\tf{tx}(\tx_n)] .
$$ 
\end{defn}

\begin{lemm}
\label{lemm.eta.up}
If 
\begin{itemize*}
\item
$x=[\tx_1,\dots,\tx_n]\in\tf{Chunk}_{\mathbb T}$, then 
\item
$\eta_{\mathbb T}(x)=[\tf{tx}(\tx_1),\dots,\tf{tx}(\tx_n)]\in\tf{Chunk}_{\rG\lF(\mathbb T)}$.
\end{itemize*}
As a corollary, Definition~\ref{defn.eta} does indeed map chunks to chunks. 
\end{lemm}
\begin{proof}
By a routine check on Definition~\ref{defn.rG}(\ref{F.arrow}), using Proposition~\ref{prop.left.to.blocked}.
\end{proof}

\begin{lemm}
Recall the notions of $\utxi$, $\utxo$, and $\stx$ from Definition~\ref{defn.utxo.utxi} and the notion of $\f{blockedUtxi}$ from Definition~\ref{defn.blocked.utxo}.
Recall $\eta$ from Definition~\ref{defn.eta} and suppose $x\in\ns X\in\tf{ACS}$.
Then:
\begin{enumerate*}
\item
$\f{utxi}(\eta_{\mathbb T}(x))=\f{utxi}(x)\setminus\f{blockedUtxi}(x)$ 
\item
$\f{utxo}(\eta_{\mathbb T}(x))=\f{utxo}(x)\cup\f{blockedUtxi}(x)$ 
\item
$\f{stx}(\eta_{\mathbb T}(x))=\f{stx}(x)$ 
\end{enumerate*}
\end{lemm}
\begin{proof}
By routine calculations using Proposition~\ref{prop.left.to.blocked}.
\end{proof}

\begin{prop}
\label{prop.unit.eta.iso}
The unit map $\eta_{\mathbb T}:\tf{Chunk}_{\mathbb T}\to \tf{Chunk}_{\rG\lF(\mathbb T)}\in\tf{IEUTxO}$ is a bijection on underlying sets.
\end{prop}
\begin{proof}
Using Lemma~\ref{lemm.tx.inj} and the fact that from Definitions~\ref{defn.lF} and~\ref{defn.rG}, any element of $\rG\lF(\mathbb T)$ has the form $[\tf{tx}(\tx_1),\dots,\tf{tx}(\tx_n)]$ for some transactions $\tx_1,\dots,\tx_n\in\tf{Transaction}_{\mathbb T}$. 
\end{proof}

\subsection{\texorpdfstring{$\lF$}{F} is left adjoint to \texorpdfstring{$\rG$}{G}}

\begin{rmrk}
\label{rmrk.diagram.chasing}
Naturality of $\epsilon$ and $\eta$ is Propositions~\ref{prop.epsilon.natural} and~\ref{prop.eta.natural}.
The proofs are by diagram-chasing.
We do need to be a little careful because of the choice made in the $\f{factor}$ function (Definition~\ref{defn.atomic}(\ref{atomic.factorisation})), which propagates to the action of $\rG$ on arrows (Definition~\ref{defn.rG.arrow}).
In the event, the diagram-chasing is all standard and it works fine.\footnote{Let's pause on this `it works fine'. 
This paper is populated by structures whose definitions can be quite different, and yet which mesh together: consider Proposition~\ref{prop.pos.eq.posi} and Remark~\ref{rmrk.pos.posi.2}, and the adjunction here. 

If aspects of the proof of the equivalence follow without fuss, then this tells us that the different elements mesh correctly and with minimal friction. 
It might even have taken much thought by a certain author, and patient finessing of widely-separated definitions and proofs, for us to enjoy this happy state of affairs.
Thus: the property of this argument that it is \emph{fairly straightforward}, may in and of itself have some \emph{mathematical} significance.} 
\end{rmrk}

\begin{prop}
\label{prop.epsilon.natural}
The counit map $\epsilon$ is a natural transformation from $\lF\rG$ to $1_{\tf{ACS}}$.
\end{prop}
\begin{proof}
Consider some arrow $\rg:\ns X\to\ns Y\in\tf{ACS}$.
We must check a commuting square in $\tf{ACS}$ that
$$
\epsilon_{\ns Y}\,\lF\rG(\rg) = \rg\,\epsilon_{\lF\rG(\ns X)}.
$$
Using Proposition~\ref{prop.ACS.morphism.atomic}(\ref{atomic.equality.test}) it suffices to check for each $x'\in\f{atomic}(\lF\rG(\ns X))$ that 
$$
\epsilon_{\ns Y}(\lF\rG(\rg)\,(x')) = g(\epsilon_{\lF\rG(\ns X)}\,x').
$$
From Lemma~\ref{lemm.lF.gives.valid.singleton.chunks}(\ref{atomic.lF}), $x'=[\tf{tx}(x)]$ for $x\in\f{atomic}(\ns X)$.

Write $\rg(x)=y_1\mact\ldots\mact y_n$ where $\f{factor}(\rg(x))=[y_1,\dots,y_n]\in[\f{atomic}(\ns Y)]$ (Definition~\ref{defn.atomic}(\ref{atomic.factorisation})), so that
$$
(\lF\rG(\rg))([\tf{tx}(x)]) = [\tf{tx}(y_1),\dots,\tf{tx}(y_n)]\in\lF\rG(\ns Y) .
$$
Then 
$$
\epsilon_{\ns Y}(\lF\rG(\rg)([\tf{tx}(x)]))=\epsilon_{\ns Y}([\tf{tx}(y_1),\dots,\tf{tx}(y_n)]) = y_1\mact\ldots\mact y_n = y 
$$ 
and
$$
\rg(\epsilon_{\ns X}([\tf{tx}(x)])) =\rg(x) = y 
$$ 
as required.
\end{proof}

Consider $\lf:\mathbb S\to\mathbb T$ and $\lF(\lf):\lF(\mathbb S)\to\lF(\mathbb T)$ and some $[\tx]\in\tf{Chunk}_{\mathbb S}$ and $\lf(\tx)=[\tx_1,\dots,\tx_n]\in\tf{Chunk}_{\mathbb T}$.
Then $\eta_{\mathbb S}([\tx_1,\dots,\tx_n]) = [\tf{tx}(\tx_1),\dots,\tf{tx}(\tx_n)]$ and
$$
\lF(\lf)([\tx])\,([\tx_1,\dots,\tx_n])
$$ 

\begin{prop}
\label{prop.eta.natural}
The unit map $\eta$ is a natural transformation from $1_{\tf{IEUTxO}}$ to $\rG\lF$.
\end{prop}
\begin{proof}
Consider some arrow $\lf:\mathbb S\to\mathbb T\in\tf{IEUTxO}$.
We must check a commuting square in $\tf{IEUTxO}$ that
$$
\eta_{\mathbb T}\,\lf = \rG\lF(\lf)\,\eta_{\mathbb S}.
$$
Using Lemma~\ref{lemm.validity.preserved} it suffices to check for each $\tx\in\tf{Transaction}_{\mathbb S}$ that 
$$
\eta_{\mathbb T}(\lf([\tx])) = \rG\lF(\lf)(\eta_{\mathbb S}([\tx])) .
$$
Suppose $\lf(\tx)=[\tx_1,\dots,\tx_n]\in\tf{Chunk}_{\mathbb T}$.
Then
$$
\eta_{\mathbb T}(f([\tx])) = \eta_{\mathbb T}([\tx_1,\dots,\tx_n]) = [\tf{tx}(\tx_1),\dots,\tf{tx}(\tx_n)]
$$
and
$$
\rG\lF(\lf)(\eta_{\mathbb S}([\tx]))
=
\rG\lF(\lf)([\tf{tx}(\tx)]) 
= 
[\tf{tx}(\tx_1),\dots,\tf{tx}(\tx_n)]  
$$
as required.
\end{proof}

\begin{thrm}\leavevmode
\label{thrm.adjoints}
\begin{enumerate}
\item
The functors 
$$
\lF:\tf{IEUTxO}\to\tf{ACS}
\quad\text{and}\quad
\rG:\tf{ACS}\to\tf{IEUTxO}
$$ 
from Definitions~\ref{defn.lF} and~\ref{defn.lF.arrow} (for $\lF$) and from Definitions~\ref{defn.rG} and~\ref{defn.rG.arrow} (for $\rG$) form an adjoint pair.
In symbols:
$$
\lF\dashv \rG : \tf{IEUTxO}\to\tf{ACS}
$$
\item
$\lF$ is full and faithful (a bijection on homsets), and $\rG$ is full (a surjection on homsets).

Thus $\lF$ is a \emph{full embedding} of $\tf{IEUTxO}$ into $\tf{ACS}$.\footnote{There seem to be various definitions in the literature of what an `embedding of categories' should mean.  By \emph{full embedding} here, we mean a functor that is injective on objects, and bijective on arrows.}
\item
The image $\rG(\ns X)$ of $\ns X\in\tf{ACS}$ is in fact a pure IUTxO model --- meaning that its validators examine only the input-point of the transaction passed to them, not the entire transaction --- and $\ns G$ maps more specifically to the full subcategory $\tf{IUTxO}$ of $\tf{IEUTxO}$ (Subsection~\ref{subsect.iutxo}).

Thus have the following short chain of maps
\begin{equation}
\label{eq.adj}
\tf{IUTxO}
\ \stackrel{\f{e}\ \dashv\ \rG\lF}\longrightarrow\  
\tf{IEUTxO} 
\ \stackrel{\lF\ \dashv\ \rG}\longrightarrow\  
\tf{ACS}
\end{equation}
where on the left $\f{e}$ denotes the trivial inclusion/embedding map 
(Remark~\ref{rmrk.utxo}) 
with right adjoint $\rG\lF$.\footnote{Does $\rG$ map to $\tf{IEUTxO}$, or to $\tf{IUTxO}$?  Both, because we embedded the latter in the former: $\tf{IUTxO}\subseteq\tf{IEUTxO}$.  See Remark~\ref{rmrk.identify}.}
\item\label{adj.iso}
The image $\rG\lF(\mathbb T)$ of $\mathbb T\in\tf{IEUTxO}$ is a pure IUTxO model and as a corollary, every IEUTxO model $\mathbb T$ is isomorphic via $\eta_{\mathbb T}:\mathbb T\to\rG\lF(\mathbb T)$ to an IUTxO model $\rG\lF(\mathbb T)$, 
\end{enumerate}
\end{thrm}
\begin{proof}\leavevmode
\begin{enumerate}
\item
The natural transformations are $\epsilon:\lF\rG\to 1$ and $\eta:1\to\rG\lF$ from Definitions~\ref{defn.epsilon} and~\ref{defn.eta}.
They are natural by Propositions~\ref{prop.epsilon.natural} and~\ref{prop.eta.natural}.
\item
By Proposition~\ref{prop.unit.eta.iso} the unit $\eta$ is a bijection and it follows that $\lf$ is full and faithful.
By Proposition~\ref{prop.epsilon.bijective} the counit $\epsilon$ is a surjection and it follows that $\rG$ is full (this also follows from the fact that $\eta$ is a bijection and so an injection).
\item
This is a structural fact of Definition~\ref{defn.rG}, as observed in Remark~\ref{rmrk.its.utxo}.
\item
From Proposition~\ref{prop.unit.eta.iso} and part~3 of this result.
\end{enumerate} 
\end{proof}

\begin{rmrk}
\label{rmrk.loop.of.embeddings}
We can present the diagram in~\eqref{eq.adj} in Theorem~\ref{thrm.adjoints}(3) (recall that $1$ and $\tf{IUTxO}$ are from Subsection~\ref{subsect.iutxo}) quite nicely as a loop of embeddings:
\begin{center}
\begin{tikzpicture}[start chain=going left,node distance=2cm]   
\node[on chain]                 (2) {$\tf{ACS}$};
\node[on chain]                 (1) {$\tf{IEUTxO}$};
\node[on chain]                 (0) {$\tf{IUTxO}$};
\draw[
    >=latex,
    auto=right,                      
    loop above/.style={out=75,in=105,loop},
    every loop,
    ]
     (0)   edge             node {$1$}     (1)
     (1)   edge             node {$\lF$}   (2)
     (2)   edge[bend right]             node {$\rG$}   (0)
     ;
\end{tikzpicture}
\end{center}
%
\end{rmrk}

Intuitively, $\tf{ACS}$ denotations seem to be the largest denotational class which conveniently maps back down into $\tf{IEUTxO}$ models as above (this is an intuitive observation, not a theorem).
But if we wish to optimise differently and make our denotation more specific, then we can be rewarded with a tighter result:
\begin{prop}
\label{prop.why.pure}
If we strengthen Definition~\ref{defn.acs.system} so that an abstract chunk system is a 
\emph{perfectly atomic} oriented monoid of chunks (Definition~\ref{defn.atomic}(\ref{perfectly.atomic})) instead of just being an atomic one, then 
\begin{itemize*}
\item
$\epsilon_{\ns X}$ in Proposition~\ref{prop.epsilon.bijective} becomes a bijection, and 
\item
the embedding $\lF\dashv \rG$ in Theorem~\ref{thrm.adjoints} becomes an equivalence of categories, and
\item
the loop illustrated in Remark~\ref{rmrk.loop.of.embeddings} becomes a loop of equivalences.
\end{itemize*}
Thus, the full subcategory in $\tf{ACS}$ of \emph{perfectly atomic} abstract chunk systems, and all arrows between them, is the `properly IEUTxO-like' abstract chunk systems.\footnote{$\rG$ and $\lF$ still have non-trivial work to do, as we see e.g. from Propositions~\ref{prop.left.to.blocked} and~\ref{prop.pos.eq.posi}.
}
\end{prop}
\begin{proof}
The bijection is just a structural fact: if by Definition~\ref{defn.atomic}(\ref{perfectly.atomic.unique}) an element $x\in\ns X\in\tf{ACS}$ factorises uniquely into a list of atomic elements $x_1\mact \dots \mact x_n$, then $x$ can be \emph{identified} with that list and $\rG$ just maps it to a list of transactions $[\tf{tx}(x_1),\dots,\tf{tx}(x_n)]$ --- which, by design in Definition~\ref{defn.rG}, has the same composition behaviour in $\rG(\ns X)$ as $x$ does in $\ns X$.

Then Definition~\ref{defn.atomic}(\ref{perfectly.atomic.leq}) is exactly what is required for $\lF\rG(\ns X)$ to `remember' all the $\leq$-structure of $\ns X$.
\end{proof}

\section{Conclusions}
\label{sect.discussion}

We have presented the EUTxO blockchain model in a novel and compact form and derived from it an algebra-style theory of blockchains as partially-ordered partial monoids with channel name communication.
This builds on previous work~\cite{gabbay:utxabs}.

We hope this paper will make two contributions: 
\begin{enumerate}
\item
its specific definitions and theorems, but also, 
\item
an \emph{idea} 
that blockchain structures can be subjected to this kind of analysis.
\end{enumerate}
Or to put it another way: we illustrate that an algebraic theory of blockchains is possible, and what it might look like.

The reader can apply these ideas to their favourite blockchain architecture, and if this were widespread practice then this might help make the field accessible to an even broader audience, ease technical comparisons between systems, add clarity to a fast-changing field --- and as we have argued, it might suggest structures and tests for practical programs, as is already reflected in a recent work~\cite{gabbay:utxabs,gabbay:idealisedeutxo}.

We now reflect on the design decisions made along the way and suggest possibilities for future work. 

\subsection{Observational equivalence}

We touched on notions of observational equivalence in Subsections~\ref{subsect.obs.obs} and~\ref{subsect.obs.eq}.

The theory of EUTxO observational equivalence is a little weaker than we might like, in the sense that not as many things get identified as one might first anticipate.

This is because a validator gets access to the whole transaction from which an input emanates --- see the line for $\tf{Validator}$ in Figure~\ref{fig.ieutxo}.
(This is specific to EUTxO; UTxO is more local, see Remark~\ref{rmrk.utxo}.)

This makes it hard to factor out internal structure.
For instance, even if all but one of the inputs and outputs of a transaction have been spent, so that it has just one dangling input --- then that entire transaction is still observable at the final dangling input.
Recall Proposition~\ref{prop.fresh.chunks.iff}, and consider $\tx,\tx'\in\tf{Transaction}$ such that $\pos(\tx)\cap\pos(\tx')=\varnothing$, so that $\tx$ and $\tx'$ are commuting.
We might reasonably wish to identify $\tx\mact\tx'$ and $\tx'\mact\tx$ with a composite transaction which we might write $\tx\cup\tx'$, but we cannot do this because a validator could see the difference.

Thus, currently the EUTxO framework can observe the difference between 
\begin{itemize}
\item
one large transaction, and 
\item
a composite chain (i.e. a chunk; see Definition~\ref{defn.chunk}) of smaller ones 
\end{itemize}
--- even if they are `morally' the same.
This is not good for developing compositional theories of observational equivalence.

We propose it might be helpful if the EUTxO model allowed us to limit access to the channels in a transaction, such that an input can limit validators to access only certain inputs and outputs in a transaction, e.g. by nominating positions that are `examinable' along that input.\footnote{So a validator would be passed the \emph{restriction} of a transaction obtained by including only nominated inputs and outputs in that transaction, and withholding the rest.} 
The pure UTxO model (Remark~\ref{rmrk.utxo}) would correspond to the special case when an input nominates only itself for validators to access.

Then, we would in suitable circumstances be able to switch between a single large transaction, and a composite chunk with the same inputs and outputs.

Another observable of a transaction is the names of its spent transaction channels, and this brings us to garbage-collection:

\subsection{Garbage-collection}
\label{subsect.garbage-collection}

Our model has no garbage-collection of spent positions --- where by \emph{spent positions} we mean the $\f{stx}$-atoms from Definition~\ref{defn.utxo.utxi}; see also the $\f{up}$-atoms from Definition~\ref{defn.left.right.up} (a connection is precisely stated in Proposition~\ref{prop.left.to.blocked}).

There is nothing wrong with this; the EUTxO presentation in~\cite{chakravarty:extum} does not garbage-collect either (i.e. positions on the blockchain can be occupied at most once, and can never be un-occupied).

However, since we name our positions, it might be nice to consider garbage-collecting (i.e. locally binding) the names of spent positions, by removing or $\alpha$-converting them in some way.
We do this in the implementation in~\cite{gabbay:idealisedeutxo}.

However, the maths indicates that this does come at a certain price. 
For instance, consider a very simple model of finite lists of atoms in which the first atom is an `input' and the last atom is an `output', and we garbage-collect by removing matching atoms, like so:
$$
[a,b]\mact [b,c] = [a,c]
\qquad
[a,c]\mact [c,b] = [a,b] .
$$ 
Thus,
$$
([a,b]\mact [b,c])\mact [c,b] = [a,b] .
$$
Now if we bracket the other way, then $[b,c]\mact [c,b]$ is ill-defined, and therefore so is $[a,b]\mact([b,c]\mact [c,b])$.
We cannot have $[b,b]$ because this would violate the condition that an input must point to an earlier (not a later) output (Definition~\ref{defn.chunk}(\ref{chunk.earlier}); we consider relaxing this condition in item~\ref{loops} of Subsection~\ref{subsect.future.work}).

The partiality is no issue --- chunks are already a partial monoid (Theorem~\ref{thrm.popm}).
But, our monoid of garbage-collecting chunks would not be \emph{associative}.
It would still be \emph{nearly} associative, meaning that $x\mact (y\mact z)=(x\mact y)\mact z$ where both sides are defined. 
However this is a weaker property that would be messier to work with, and for now we do not need it to make our case.

(This could also be read as a mathematical hint that our uniqueness conditions are a little too strict, and that $[b,c]\mact[c,b]$ should be permitted, where the leftmost $b$ does not point to the later $b$ but instead points backwards in time to some earlier $b$, and the rightmost $b$ points forwards towards some later $b$ --- in the style of an \emph{interleaved scope} \cite{gabbay:leatnn}, for example.  In the presence of garbage-collection this would make perfect sense, since we would expect both $b$s to eventually get bound, i.e. `spent'.)
 
In this paper we do not garbage-collect.
We leave deciding whether we should, and if so how, for future work.

\subsection{Tests}
\label{subsect.tests}

A few more words on the equations in Figure~\ref{fig.ieutxo} (IEUTxO type equations) vs. those in Definition~\ref{defn.acs} (its algebraic counterpart of abstract chunk systems).

We want the equations in Figure~\ref{fig.ieutxo} to be short and sweet, so that the system looks simple and solutions are easy to build and manipulate.

The design parameters in Definition~\ref{defn.acs} are somewhat different: we do not mind if there are plenty of algebraic properties, because this means that we have captured as many interesting properties as possible.
The adjunction in Theorem~\ref{thrm.adjoints} gives a formal sense in which the two correspond.

(This leaves it for future work to see how these conditions could be relaxed, as discussed e.g. in Remarks~\ref{rmrk.why.factorisation} and~\ref{rmrk.factor} and Proposition~\ref{prop.why.pure}.)

This has mathematical interest, but not only that.

As noted in Subsection~\ref{subsect.brief.discussion}, modern programming languages support efficient programming on abstract denotations, thus delaying instantiating to specific instances until truly necessary. 
Also, they allow us to express and test against properties --- and equality properties in particular can be helpful for optimising transformations.

So an algebraic theory can be relevant to producing concrete working code, because: 
\begin{enumerate}
\item
it can help structure code; and 
\item
an axiom can be read both as a testable property and as a program transformation; so that 
\item
the more axioms we have, the more transformations and tests are available, and the more scope we have to transform, structure, and test our programs.
\end{enumerate}

\subsection{Connections with nominal techniques}
\label{subsect.nominal}

This paper borrows ideas from \emph{nominal techniques}~\cite{gabbay:newaas-jv} and in particular it follows the ideas on Equivariant ZFA from~\cite{gabbay:equzfn} in handling the atoms which we use to name positions. 
We use atoms to name positions in IEUTxO models and in abstract chunk systems (ACS).

Partly, this is just using nominal-style names and permutations as a standard vocabulary. 
We can think of this application of nominal ideas as reaching for a familiar API, typeclass, or algebra, and there is nothing wrong with that.

The reader can find this implemented in~\cite{gabbay:idealisedeutxo}, where the IEUTxO equations in Figure~\ref{fig.ieutxo} are combined with this author's nominal datatypes package to create first a Haskell typeclass, and then a working implementation of chunks and blockchains, following the IEUTxO model of this paper. 

But in parts of this paper, something deeper is also taking place.
For example:
\begin{enumerate}
\item
The notion of abstract chunk system --- which makes up the more abstract half of this paper --- depends on that of an oriented monoids of chunks from Definition~\ref{defn.oriented}, which depends on the notion of $\posi$ from Definition~\ref{defn.posi}, whose definition depends on the permutation action. 

$\posi$ is a nameful definition, in the nominal sense, and it is not clear how would express it, and thus the notion of ACS, were it not for our nominal use of names.
\item
In~\cite{gabbay:idealisedeutxo} we go somewhat further than in this paper in developing IEUTxO models, in that we permit $\alpha$-conversion to garbage-collect spent positions as discussed in Subsection~\ref{subsect.garbage-collection}.
This too is a fully nominal definition, which uses the nominal model of binding in specific ways.
\end{enumerate}

\begin{rmrk}[No support assumed]
Experts on nominal techniques should note that no finite support conditions are imposed; for example, there is nothing to insist that a validator $\nu(v)\subseteq\beta\times\tf{Transaction}_!$ should be a finitely-supported subset.
This paper is in ZFA, not Fraenkel-Mostowski set theory.

We could construct a finitely-supported account of the theory in this paper in FM, just by imposing finite support conditions appropriately --- but it would cost us complexity, and since we do not seem to need this, we do not do it.
We hope we have struck a good balance in the mathematics between being rigorous in our treatment of names, and not drowning the reader in detail. 

Note that finiteness conditions on sets of names are still important: notably in Definition~\ref{defn.oriented}(\ref{oriented.posi.finite}); and name-management is key to some nontrivial results, notably Corollary~\ref{corr.lr.tx.empty}.
So there is a `finite support' flavour to the maths, just not as a direct translation of the Fraenkel-Mostowski notion of finite support.
\end{rmrk}

\begin{rmrk}[Disjointness conditions]
Continuing the previous remark, `freshness-flavoured' set disjointness conditions like 
\begin{itemize*}
\item
$\f{input}(\tx)\cap\f{output}(\tx)=\varnothing$ in Lemma~\ref{lemm.singleton.chunk.valid}, and 
\item
$\pos(\f{ch})\cap\f{pos}(\f{ch}')=\varnothing$ in Lemma~\ref{lemm.fresh.chunks.defined}, and 
\item
$\posi(x)\cap\posi(y)=\varnothing$ in Definition~\ref{defn.oriented} and elsewhere, 
\end{itemize*}
are quite important in this paper.

We could have imported a nominal notation and written these $\#$, as in $\f{ch}\#\f{ch}'$ or $x\#y$.
This would not be wrong, but it might mislead because, as discussed above, we do not assume nominal notions of support and freshness.
Thus, $\pos(\f{ch})$ and $\posi(x)$ are \emph{not} necessarily equal to the support of $\f{ch}$ and $x$ respectively, and if e.g. $x$ and $y$ were in an abstract chunk system that happened also to be finitely-supported (a plausible scenario), it would not be guaranteed that $x\#y$ would coincide with $\posi(x)\cap\posi(y)=\varnothing$.\footnote{$\posi(x)\#\posi(y)$ in the nominal sense would still mean $\posi(x)\cap\posi(y)=\varnothing$, just because nominal freshness coincides with sets disjointness on finite sets of atoms.}

We could insist that $\supp(x)$ must exist and coincide with $\posi(x)$, of course --- but that would be an additional restriction.
\end{rmrk}

\subsection{Concrete formalisation}
\label{subsect.formalisation}

Can we implement all or parts of this paper in a theorem-prover, and use that to verify properties of a blockchain system?

This paper has a lot of moving parts, and it is not all or nothing: a user can import whichever components (IEUTxO? ACS?) they wish --- though we also hope that the overall mathematical vision could provide useful guidance. 
 
How the ideas in this paper might be brought to bear in the setting of a formal verification cannot have a clear-cut answer, because it depends on interactions between the maths, what we want to verify, the resources available,\footnote{Very important: a single person working alone will make different tradeoffs than a large team, and a short-term project will make different tradeoffs than a longer-term project that can e.g. afford to invest in building basic tooling.} and what facilities are offered by a particular theorem-proving environment. 

A reasonable (but naive) implementation of the EUTxO inductive definition in~\cite{chakravarty:extum} would suggest modelling positions by numbers which are essentially de Bruijn indices.
The maths in this paper suggests against this:
\begin{itemize}
\item
We want to talk about chunks. 
Indices only make sense in a blockchain which has an initial genesis block from which we start to index. 
\item
We want to rearrange transactions and chunks, e.g. to talk about how they commute, as in (for example) Definition~\ref{defn.circ}.
With indices, a transaction in a different place is a \emph{different transaction}, and potentially subject to different validation if a validator were to directly inspect its indices.
\item
Because we use names, we never have to worry about reindexing functions, as can be an issue with the de Bruijn indexed approach.
\end{itemize} 
This is borne out by the practice: the developers of the EUTxO implementation underlying~\cite{chakravarty:extum} have explored theorem-provers and they do not reference transactions by position (even though the mathematical description in the literature makes it look like they do, at least to the uninformed reader). 

What they actually do is maintain an explicit naming context.
This is ongoing research, but the interested reader can find a brief description in~\cite{formal-eutxo}.

The difficulty with this hands-on context-based approach is that we end up having to curate our context of names, using explicit context-weakening and context-combining operations.
These get limited, extremely tedious to formulate and prove, and clutter the proofs of properties.%
\footnote{\emph{Properties} often get called \emph{meta-theorems} in this field.  So wherever we write `(algebraic) property', a reader with a theorem-proving background can read `meta-theorem', and they will not go too far wrong.}

It is standard and known that maintaining contexts of `known names' can get painful.
Indeed, this is one of the issues that nominal techniques were developed to alleviate.

In the context of this paper, we would advise looking at how the implementation in~\cite{gabbay:idealisedeutxo} manages names and binding using permutations and the $\tt{Nom}$ binding context (which is the nominal construct that closely corresponds to a local context of known names).

So it works in Haskell --- but determining the extent to which this can be translated to a \emph{theorem-prover} remains to be seen. 
We could imitate the nominal datatypes package; there is a nominal package in Isabelle~\cite{urban:nomrti}; and this author has written some recommendations on implementing nominal techniques in theorem-proving environments~\cite{gabbay:equzfn}.\footnote{This last paper essentially says: make sure that permuting names in properties is a pushbutton operation.
Once you have that, the rest should follow; and if you do not then there may be trouble.} 
Exploring this is future work.

\subsection{Future work} 
\label{subsect.future.work}

We discussed future work above and in the body of the paper.
We conclude with some further observations:
\begin{enumerate}
\item
The authors of~\cite{chakravarty:extum} map their concrete models to \emph{Constraint Emitting Machines}, which are a novel variant of Mealy machines.
It is future work to see whether the idealised EUTxO solutions from Definition~\ref{defn.solution} admit corresponding descriptions.
For the interested reader, we can note that a body of work on nominal automata does exist~\cite{bojanczyk:auttns,bojancyk:sliis}.
\item\label{loops}
In condition~\ref{chunk.earlier} of Definition~\ref{defn.chunk} we restrict inputs to point to strictly earlier outputs.
This restriction makes operational sense for a blockchain, and it is unavoidable and required for well-definedness in~\cite{chakravarty:extum} because of its inductive construction.

However, there is no mathematical necessity to retain it here; it would be perfectly valid and possible to contemplate a generalisation of Definition~\ref{defn.chunk} which permits loops from a transaction to itself, or even forward pointers from inputs to later outputs (i.e. `feedback loops').

Mathematically and structurally, loops are perfectly admissible in the framework of this paper.

Loops from a transaction to itself are particularly interesting, because this would remove the need for a \emph{genesis block} --- which has no inputs, and therefore exists \emph{sui generis} in that it is not subject to the action of any validation from other blocks.
So, we could insist that all transactions have at least one input, but that input may loop to an output on the same transaction.

This might seem like a minor difference, but it is not, because \emph{we} control the set of validators (the injection $\nu$ in Definition~\ref{defn.solution}), so we can enforce checks of good behaviour --- even of the first transaction, which would e.g. be forced to validate itself via a loop --- by controlling the set of validators. 
\item
We have considered EUTxO blockchains in this paper.
It would be natural to attempt a similar analysis for an accounts-based blockchain architecture such as Ethereum.
A start on this is the \emph{Idealised Ethereum} equations in~\cite[Figure~4]{gabbay:utxabs}; developing this further is future work.
\item
As is often the case, in practice there are desirable features of real systems that would break our model.

For instance, it can be useful to make transactions time-sensitive using \emph{slot ranges}, also called \emph{validity intervals}; a transaction can only be accepted into a block whose slot is inside the transaction's slot range.
Clearly, this could compromise results having to do with chunks and commutations, because these are by design `pure' notions, with no notion of time. 

It is future work to see to what extent a theory of chunks might be compatible with explicit time dependence constraints.
In one sense this temporal aspect should be orthogonal to the nominal techniques used so far, since atoms are just used as positions and we impose no finite support conditions; however in another sense it may also be possible to usefully import notions of e.g. ordered atoms from nominal automata, as presented in~\cite{bojanczyk:auttns,bojancyk:sliis}.
In any case, it is common to find a pure sublanguage with a nice theory embedded in a more expressive and larger language with less good behaviour, as e.g. pure SML is embedded in SML with global variables.
\end{enumerate}

\subsection{Final words} 

The slogan of this paper is \emph{blockchains as algebras},
and more specifically \emph{blockchains as nominal algebras of chunks}.
Based on the maths above, which substantiates this slogan, we can note that: 
\begin{enumerate}
\item
UTxO blockchains can be viewed as an algebraic structure, both intensionally (Definition~\ref{defn.solution}) and extensionally (Definition~\ref{defn.acs.system}), and these views are equivalent (Theorem~\ref{thrm.adjoints}).
\item
The natural and fundamental unit of blockchain algebras is \emph{partial} blockchains --- what we call \emph{chunks} in this paper --- and chunks naturally organise themselves into monoids with extra structure.
\item
It can be convenient to address inputs and outputs by \emph{name} in a `nominal' style --- contrast with a de Bruijn-style index (addressing a location by some kind of offset from a genesis block; cf. Subsection~\ref{subsect.formalisation}), or a blockchain hash.
\item
We have proved some high-level properties (some readers might call \emph{meta-theorems}).
Notably: a kind of Church-Rosser property in Theorem~\ref{thrm.practical}; and Theorem~\ref{thrm.adjoints}(\ref{adj.iso}) that every IEUTxO system admits a presentation as an IUTxO system; and that both can be presented as Abstract Chunk Systems, in senses made formal by a loop of embeddings in Remark~\ref{rmrk.loop.of.embeddings}.
These capture global properties of how the structures fit together, which are not immediately obvious just reading the definitions. 
\item
An open question is how this maths might help the blockchain community to simplify proofs, write better programs, more quickly and safely design new systems, or accommodate existing and new extensions (cf. next point).
The results in this paper, coupled with a toy (but perfectly real) blockchain implementation based on these ideas~\cite{gabbay:idealisedeutxo}, suggests that these things are plausible --- but time and future research will reveal more, and that is as it should be. 
\item
We knowingly threw out concrete structure that real blockchains need.
\emph{Real} blockchains have tokens, monetary policies, smart contracts, and more structure that is being invented literally on a daily basis. 
But instantiating algebraic structures is possible as discussed in Subsection~\ref{subsect.what.its.about} and indeed 
the \emph{raison d'\^etre} of algebra is just this: abstract, find instances, and extend.
\item
One path to applying this paper is essentially to take Proposition~\ref{prop.fresh.chunks.iff} and Remark~\ref{rmrk.separation.logic} seriously and use them to design domain-specific resource-aware logics for reasoning about algebras of chunks equipped with specific features of possibly industrial relevance.

These would be logics for verifying domain-specific elaborations of Theorem~\ref{thrm.practical}, and for tracing resources through the evolution of a blockchain.
The nominal method of tracking resources (e.g. as we have seen used to label inputs and outputs and define Abstract Chunk Systems) lends itself naturally to such applications.\footnote{Nominal techniques were originally designed to track resources (namely: variable symbols in inductive definitions with binding) and they have been elaborated considerably since e.g. through nominal rewriting, automata, and nominal algebra.  So there is a body of theory to draw on.}

Thus the semantics here, enriched with data structures of practical interest --- e.g. currencies, NFTs, or other `real' data coming from specific user needs --- could lead to logics that are both powerful and useful, helping to prove commutativity and other resource-based properties, using logics for algebras of chunks in the style of this paper.
\end{enumerate} 
In summary: this paper introduces the idea of \emph{blockchain algebra}, at least as applied to UTxO-style blockchains.
It is reasonable to hope that this might provide a convenient target semantics for new programs and logics for building and reasoning about blockchain systems.

\begin{acks} 
I thank the editors and three anonymous referees for their detailed and constructive feedback: thank you for your time and input.
Thanks also to Lars Br\"unjes, without whom this paper might not have been written.
I dedicate this paper to the memory of our colleague Martin Hofmann.  May he rest in peace. 
 \end{acks}


\begin{thebibliography}{CCM{\etalchar{+}}20}

\bibitem[Acz88]{aczel:nonwfs}
Peter Aczel, \emph{Non-wellfounded set theory}, CSLI lecture notes, no.~14,
  CSLI, 1988.

\bibitem[Bar84]{barendregt:lamcss}
Henk~P. Barendregt, \emph{The lambda calculus: its syntax and semantics
  (revised ed.)}, North-Holland, 1984.

\bibitem[BG20]{gabbay:utxabs}
Lars Br\"unjes and Murdoch~J. Gabbay, \emph{{UTxO}- vs account-based smart
  contract blockchain programming paradigms}, Proceedings of the 9th
  International Symposium On Leveraging Applications of Formal Methods,
  Verification and Validation (ISOLA 2020), Springer, October 2020, See arXiv
  preprint https://arxiv.org/abs/2003.14271.

\bibitem[BKL14]{bojanczyk:auttns}
Miko\l{}aj Boja\'nczyk, Bartek Klin, and S\l{}awomir Lasota, \emph{Automata
  theory in nominal sets}, Logical Methods in Computer Science \textbf{10}
  (2014).

\bibitem[Boj18]{bojancyk:sliis}
Miko\l{}aj Boja\'nczyk, \emph{Slightly infinite sets (book draft)}, 2018.

\bibitem[CCM{\etalchar{+}}20]{chakravarty:extum}
Manuel~M.T. Chakravarty, James Chapman, Kenneth MacKenzie, Orestis Melkonian,
  Michael {Peyton Jones}, and Philip Wadler, \emph{The {Extended UTXO} model},
  Proceedings of the 4th Workshop on Trusted Smart Contracts (WTSC'2020), LNCS,
  vol. 12063, Springer, 2020.

\bibitem[CFW86]{clinger:schhls}
William~D. Clinger, Dan~P. Friedman, and Mitchell Wand, \emph{A scheme for a
  higher-level semantic algebra}, p.~237–250, Cambridge University Press,
  USA, 1986.

\bibitem[Gab01]{gabbay:thesis}
Murdoch~J. Gabbay, \emph{\href{http://www.gabbay.org.uk/papers.html\#thesis}{A
  Theory of Inductive Definitions with alpha-Equivalence}}, Ph.D. thesis,
  University of Cambridge, UK, March 2001.

\bibitem[Gab20a]{gabbay:equzfn}
\bysame, \emph{{Equivariant ZFA and the foundations of nominal techniques}},
  Journal of Logic and Computation \textbf{30} (2020), 525--548.

\bibitem[Gab20b]{gabbay:idealisedeutxo}
\bysame, \emph{Implementation of idealised {EUTxO} in the nominal datatypes
  package},
  \url{https://github.com/bellissimogiorno/nominal/blob/6e9c/src/Language/Nominal/Examples/IdealisedEUTxO.hs},
  June 2020.

\bibitem[GGP15]{gabbay:leatnn}
Murdoch~J. Gabbay, Dan~R. Ghica, and Daniela Petrisan, \emph{{Leaving the Nest:
  Nominal Techniques for Variables with Interleaving Scopes}}, 24th EACSL
  Annual Conference on Computer Science Logic (CSL 2015) (Dagstuhl, Germany)
  (Stephan Kreutzer, ed.), Leibniz International Proceedings in Informatics
  (LIPIcs), vol.~41, Schloss Dagstuhl--Leibniz-Zentrum fuer Informatik, 2015,
  pp.~374--389.

\bibitem[GP01]{gabbay:newaas-jv}
Murdoch~J. Gabbay and Andrew~M. Pitts,
  \emph{\href{http://www.gabbay.org.uk/papers.html\#newaas-jv}{A New Approach
  to Abstract Syntax with Variable Binding}}, Formal Aspects of Computing
  \textbf{13} (2001), no.~3--5, 341--363.

\bibitem[Mil99]{Milner:comms}
Robin Milner, \emph{Communicating and mobile systems: the $\pi$-calculus},
  Cambridge University Press, 1999.

\bibitem[Mok17]{mokhov:alggwc}
Andrey Mokhov, \emph{Algebraic graphs with class (functional pearl)}, SIGPLAN
  Not. \textbf{52} (2017), no.~10, 2–13.

\bibitem[MSC19]{formal-eutxo}
Orestis Melkonian, Wouter Swierstra, and Manuel~MT Chakravarty, \emph{Formal
  investigation of the {Extended UTxO} model ({Extended Abstract})},
  \url{https://omelkonian.github.io/data/publications/formal-utxo.pdf}, 2019.

\bibitem[Nes21]{nester:fouls}
Chad Nester, \emph{{A Foundation for Ledger Structures}}, 2nd International
  Conference on Blockchain Economics, Security and Protocols (Tokenomics 2020)
  (Dagstuhl, Germany) (Emmanuelle Anceaume, Christophe Bisi\`{e}re, Matthieu
  Bouvard, Quentin Bramas, and Catherine Casamatta, eds.), Open Access Series
  in Informatics (OASIcs), vol.~82, Schloss Dagstuhl--Leibniz-Zentrum f{\"u}r
  Informatik, 2021, pp.~7:1--7:13.

\bibitem[Rey02]{reynolds:seplls}
John~C. Reynolds, \emph{Separation logic: A logic for shared mutable data
  structures}, Proceedings of the 17th {IEEE} Symposium on Logic in Computer
  Science ({LICS} 2002), IEEE Computer Society Press, 2002, pp.~55--74.

\bibitem[Urb08]{urban:nomrti}
Christian Urban, \emph{Nominal reasoning techniques in {I}sabelle/{HOL}},
  Journal of Automatic Reasoning \textbf{40} (2008), no.~4, 327--356.

\end{thebibliography}

\newcommand{\etalchar}[1]{$^{#1}$}
\hyphenation{Mathe-ma-ti-sche}
\providecommand{\bysame}{\leavevmode\hbox to3em{\hrulefill}\thinspace}
\providecommand{\MR}{\relax\ifhmode\unskip\space\fi MR }
\providecommand{\MRhref}[2]{%
  \href{http://www.ams.org/mathscinet-getitem?mr=#1}{#2}
}
\providecommand{\href}[2]{#2}

\end{document}